\documentclass[a4paper,11pt]{amsart}
\usepackage{amsmath}
\usepackage{amsmath}
\usepackage{amsthm, amsfonts, amssymb, graphics}
\usepackage{bm}
\usepackage[foot]{amsaddr}
\usepackage{verbatim}
\usepackage{hyperref}
\usepackage{xcolor}
\usepackage[normalem]{ulem}
\usepackage{enumerate}
\usepackage{mathtools}
\usepackage{tikz}
\usepackage{geometry}
\usepackage{dsfont}
\usepackage{centernot}
\usepackage{bbm}
\usepackage[utf8]{inputenc}
\UseRawInputEncoding

\newtheorem{theorem}{Theorem}[section]
\newtheorem{lemma}[theorem]{Lemma}
\newtheorem{proposition}[theorem]{Proposition}
\newtheorem{corollary}[theorem]{Corollary}

\theoremstyle{definition}
\newtheorem{definition}[theorem]{Definition}

\newtheorem{remark}[theorem]{Remark}

\newcommand{\la}{\langle}
\newcommand{\ra}{\rangle}
\newcommand{\bi}{\bm{\mathrm{i}}}
\newcommand{\bone}{\mathbf{1}}
\newcommand{\x}{\mathbf{X}}
\numberwithin{equation}{section}
\renewcommand{\d}{\mathrm{d}}
\newcommand \id{\mathds 1}
\newcommand{\bs}{\mathbf{s}}
\newcommand{\btheta}{\bm{\theta}}
\newcommand{\harm}{\mathrm{Harm}}
\newcommand {\R} {\mathbb{R}}
\newcommand{\bmu}{\bm{\mu}}
\newcommand {\C} {\mathbb{C}}

\newcommand{\bsigma}{\bm{\sigma}}
\newcommand {\D} {\mathbb{D}}

\newcommand{\E}{\mathbb{E}}
\newcommand{\g}{\mathfrak{g}}
\renewcommand{\a}{\mathfrak{a}}
\newcommand {\N} {\mathbb{N}}
\renewcommand {\S} {\mathbb{S}}
\newcommand {\T} {\mathbb{T}}

\newcommand{\fl}{\mathrm{FL}}
\renewcommand{\H}{\mathcal{H}}
\newcommand {\Z} {\mathbb{Z}}
\renewcommand{\P} {\mathbb{P}}

\newcommand{\reg}{\mathrm{reg}}

\newcommand{\z}{\mathbf{z}}
\newcommand{\m}{\mathbf{m}}
\newcommand{\bv}{\mathbf{v}}
\renewcommand{\Im}{\mathrm{Im}}
\renewcommand{\Re}{\mathrm{Re}}

\newcommand{\bx}{\mathbf{x}}
\newcommand{\bxi}{\bm{\xi}}
\newcommand{\bzeta}{\bm{\zeta}}
\renewcommand{\v}{\mathrm{v}}
\newcommand{\bbb}{\mathbbm{b}}

\newcommand{\balpha}{\bm{\alpha}}

\newcommand{\A}{\mathcal{A}}

\newcommand{\bbg}{\mathbbm{g}}
\newcommand{\w}{\mathbf{w}}

\newcommand {\eps}{\varepsilon}

\renewcommand{\ne}{n_{\mathfrak{e}}}
\newcommand{\nm}{n_{\mathfrak{m}}}

\renewcommand{\l}{\ell}
\newcommand{\blambda}{\bm{\lambda}}

\usepackage{tikz}
\usepackage{etoolbox}

\newcommand{\norm}[1]{\left\lVert #1 \right\rVert}
\begin{document}
\title{Compactified Imaginary Toda Theory}
\author{Yi-An Yao}
\address{Institut Fourier, Universit\'e Grenoble Alpes, Grenoble, France}
\email{yi-An.yao@etu.univ-grenoble-alpes.fr}

\begin{abstract}
Following \cite{Guillarmou:2023exh}, we construct compactified imaginary Toda theory on closed Riemann surfaces, extending the rank-one construction to the higher-rank setting. This theory is expected to describe critical higher-rank models with extended symmetries, such as web models. We construct the correlation functions and prove that they satisfy the axioms of conformal field theory, as well as Segal's gluing axioms. On the Riemann sphere, we express the correlation functions as Dotsenko--Fateev type integrals. In the case $\mathfrak g=\mathfrak{sl}_n$, under a semidegenerate condition, we obtain a closed formula for the three-point structure constant.
\end{abstract}

\maketitle

\tableofcontents

\section{Introduction}\label{section 1}
\subsection{Toda conformal field theory}
Two-dimensional conformal field theory (CFT in short) provides a universal language for describing
critical phenomena in statistical physics. At a second-order phase transition, a lattice
model is expected to lose its microscopic length scale and, after suitable rescaling, to
converge to a continuum theory whose main quantities of interest, the \emph{correlation functions},
are invariant under conformal transformations.

This principle is especially powerful in two dimensions: the local conformal symmetry is
infinite-dimensional, and its infinitesimal symmetries are encoded by the
\emph{Virasoro algebra}. The seminal work of Belavin, Polyakov and Zamolodchikov
\cite{BELAVIN1984333} showed how to exploit this symmetry to obtain strong constraints on correlation functions. Such constraints lead to the conformal bootstrap, a recursive
procedure which aims to determine higher-point correlation functions from lower-point
data, in particular from three-point functions, also known as the \emph{structure constants}.

Liouville theory is a central example where this bootstrap philosophy can be carried out
explicitly. Its origin goes back to Polyakov's path-integral approach to two-dimensional
quantum geometry and non-critical string theory \cite{POLYAKOV1981207}. The resulting theory is an interacting CFT whose structure constants, described by the DOZZ formula, were later independently predicted by Dorn-Otto \cite{dorn1994two} and Zamolodchikov-Zamolodchikov \cite{zamolodchikov1996conformal}. Together with the conformal block decomposition, these structure constants determine higher-point correlation functions on
the sphere, or more generally on any Riemann surface via Segal's gluing \cite{segal1988definition}, and provide one of the basic non-rational examples of the conformal bootstrap.

In many models, Virasoro symmetry is only part of the structure. Certain critical systems
possess additional conserved quantities, leading to extended chiral symmetry algebras and
higher-spin symmetries. Among the most important examples are $W$-algebras, introduced
by Zamolodchikov \cite{zamolodchikov1995infinite}, which extend the Virasoro algebra by currents
of higher conformal spin; see also \cite{bouwknegt1993w} for a review.
Toda CFTs provide a natural family of CFTs governed by such $W$-symmetries: the Virasoro symmetry of Liouville theory is replaced by the
$W$-algebra associated with the underlying simple Lie algebra. From the point of view
of statistical physics, such extended symmetries are expected to describe continuum
limits of models with richer algebraic or geometric features, for instance, $q$-state Potts model and web models \cite{lafay2022uq,lafay2024integrability}. This provides one motivation for studying Toda CFTs.

\subsection{Compactified imaginary Toda CFT}

The bootstrap approach gives a powerful way to characterize conformal field theories from their symmetry constraints, but it does not by itself provide a direct construction of the path integral. A complementary point of view is the probabilistic construction of CFTs, where the formal path integral is interpreted as an expectation over random fields. In
Liouville theory, this approach is based on the Gaussian free field (see \cite{sheffield2007gaussian,Berestycki_Powell_2025}) and the Gaussian multiplicative chaos (see \cite{RhodesGMC}), and has led to a rigorous construction of correlation functions and their conformal covariance properties \cite{david2016liouville} and also a rigorous proof of the DOZZ formula \cite{KRVDOZZ}, thereby connecting the probabilistic and bootstrap perspectives. Moreover, the gluing axiom and the bootstrap strategy have also been established \cite{guillarmou2021segal,GKRV}.

More recently, the same probabilistic perspective has also been developed for Toda CFTs, leading to a rigorous construction of Toda correlation functions as well as to the study of their conformal covariance and
$W$-symmetry constraints \cite{Baptiste2023,cercle2022ward}. On the bootstrap side, Toda
structure constants are expected to be governed by higher-rank analogs of the DOZZ
formula. In particular, Fateev and Litvinov proposed explicit formulae for certain Toda
structure constants \cite{fateev2007correlation}, which have been proved under the probabilistic framework in the case $\g=\mathfrak{sl}_3$ \cite{Cercle2025-zi}. These results provide the higher-rank background for the compactified imaginary Toda theory considered in this paper.

The goal of this paper is to develop an analogous construction for a compactified imaginary version of Toda theory. Our construction is inspired by, and follows closely,
the framework developed in \cite{Guillarmou:2023exh}. Let $(\Sigma,g)$ be a Riemann surface, let $\mathfrak g$ be a complex simple Lie algebra of
rank $r$, with simple roots $e_1,\ldots,e_r$, and let $\mathfrak a$ be the real
Cartan subspace. For $\gamma>0$, we consider fields $\Phi:\Sigma\rightarrow\T(\gamma)$ taking values in the torus
$$
\mathbb T(\gamma)=\mathfrak a/(2\pi\Lambda),
\qquad
\Lambda=\gamma^{-1}\bigoplus_{i=1}^r\mathbb Z\omega_i^\vee .
$$
Formally, the compactified imaginary Toda action is\begin{equation}\label{eq:Toda equation}
S_{\mathfrak g}(\Phi,g)=\frac1{4\pi}\int_\Sigma\left(\langle \partial_g\Phi,\partial_g\Phi\rangle_g+K_g\langle \bi Q,\Phi\rangle+
4\pi\sum_{i=1}^r\mu_i e^{\bi\gamma\langle e_i,\Phi\rangle}
\right)\d \v_g,
\end{equation}
where $\bi:=\sqrt{-1}$, and $Q=\gamma\rho-\frac{2}{\gamma}\rho^\vee$ is the background charge with $\rho$ the Weyl vector and $\rho^\vee$ its dual.
For a functional $F$ of the field $\Phi$, we define $$\la F\ra_{\Sigma,g}:=\int_{\Phi:\Sigma\rightarrow\T(\gamma)}F(\Phi)e^{-S_\g(\Phi,g) }\mathrm D\Phi,$$ where $\mathrm{D}\Phi$ is formally the Lebesgue measure on the function space $\{\Phi:\Sigma\rightarrow\T(\gamma)\}.$

The compactification makes the interaction terms
$e^{i\gamma\langle e_i,\Phi\rangle}$ globally well defined as characters of
$\mathbb T(\gamma)$, but it also forces the theory to include topological sectors. More precisely, a smooth map $\Phi:\Sigma\rightarrow\T(\gamma)$ induces an $\a$-valued closed $1$-form $\Omega$ with periords in $2\pi\Lambda$, and the Hodge decomposition allows one to uniquely decompose the form as $\Omega=\Omega_{\mathrm{h}}+\d f,$ with $\Omega_{\mathrm{h}}$ a harmonic $\a$-valued 1-form and $f$ a smooth $\a$-valued function. This shows that any smooth $\T(\gamma)$-valued map $\Phi$ admits an orthogonal decomposition: $\Phi=f+I_{x_0}(\Omega_{\mathrm{h}}),$ where $I_{x_0}(\Omega_\mathrm{h}):=\int_{\gamma_{x_0,x}}\Omega_{\mathrm{h}}$ is a multivalued harmonic function, with $\gamma_{x_0,x}$ a path from $x_0$ to $x$, and the orthogonality reads $$\int_\Sigma |\d \Phi|_g^2\d\v_g=\int_\Sigma|\d f|_g^2\d\v_g+\int_\Sigma \Omega_{\mathrm{h}}\wedge*\Omega_{\mathrm{h}}.$$ Thus, formally, we can rewrite the path integral with both $Q$ and $\bmu$ vanishing as $$\begin{aligned}\la F\ra_{\Sigma,g}=&\int_{\Phi:\Sigma\rightarrow\T(\gamma)}F(\Phi)e^{-\frac{1}{4\pi}\int_\Sigma|\d\Phi|_g^2\d\v_g}\mathrm{D}\Phi\\
=&\int F(f+I_{x_0}(\Omega_{\mathrm{h}})+c)e^{-\frac{1}{4\pi}\int_\Sigma (|d f|^2_g+\Omega_{\mathrm{h}}\wedge*\Omega_{\mathrm{h}})\d\v_g}\mathrm{D}f\d\mu(\Omega_{\mathrm{h}})\d c ,\end{aligned}
$$
where $\mathrm{D}f$ is a formal Lebesgue measure over zero-mean $\a$-valued functions, $\d\mu(\Omega_\mathrm{h})$ is the counting measure on the $\a$-valued de Rham cohomology group, and $\d c$ is the Lebesgue measure on the torus $\T(\gamma).$ 
We note that, due to the multivaluedness of the primitive $I_{x_0}(\Omega_\mathrm{h})$, the functional $F$ has to satisfy certain periodic conditions so that the path integral makes sense.

When the background charge $Q$ is nonzero, one has to regularize the curvature term
$$
\int_\Sigma K_g\langle \bi Q,\Phi\rangle\d \v_g,
$$
to deal with the multivalued primitive. Following
the strategy of \cite{Guillarmou:2023exh}, we do so by choosing cuts and
adding explicit geodesic-curvature counterterms, in such a way that the resulting path
integral is independent of the auxiliary choices up to the expected lattice ambiguities (see Section \ref{subsec:curvature term}).

Let us now describe more explicitly the objects constructed in this paper. The basic
observables are correlation functions of electric, magnetic, and electro-magnetic operators. An electric operator with charge $\alpha$ is obtained by Wick-renormalizing
the exponential $e^{\bi\langle \alpha,\Phi(x)\rangle}$ at an insertion point $x\in\Sigma$. A magnetic
operator, on the other hand, inserts a prescribed monodromy $m\in\Lambda$ around a
marked point, and is implemented by a singular closed $1$-form
$\nu_{\mathbf z,\mathbf m}$. Electro-magnetic operators $V^g_{(\alpha,m)}(z,v)$
combine these two effects. When $m\neq0$, they depend on a choice of tangent direction
$v$ at the insertion point and carry a non-trivial spin. We prove that the corresponding
correlation functions are well defined under natural Seiberg-type assumptions and satisfy
the expected conformal anomaly, diffeomorphism covariance, and spin covariance.

The three-point functions on the Riemann sphere, or structure constants, are expected to
play the role of basic building blocks in the conformal bootstrap. In the compactified
imaginary Toda setting, their computation leads to higher-rank Coulomb gas integrals of Dotsenko--Fateev type, with screening variables indexed by the simple roots of
$\mathfrak g$. In special cases, these integrals are
closely related to the imaginary Fateev--Litvinov formula \cite{Dupic_2019} for Toda structure constants. Thus the
correlation functions constructed here provide a probabilistic realization of the
Coulomb gas picture in the compactified imaginary Toda setting.

Finally, we establish Segal's gluing axioms \cite{segal1988definition} for the theory. To do so, we construct
the path integral on surfaces with analytic boundary and view the resulting boundary amplitudes as vectors in the Hilbert space associated with the boundary
field. We then prove that these amplitudes compose under gluing of surfaces. The proof uses the Markov property of the Gaussian free field, but also requires bookkeeping for the compactified topological sectors, magnetic backgrounds, and regularized curvature terms. These features are specific to the
compactified theory and are essential for obtaining a consistent CFT on surfaces of arbitrary topology.

We summarize the above results in Theorem \ref{thm:main-results} and refer to Theorem \ref{thm:electro-magnetic operators} and Proposition \ref{prop:CITT-Segal-gluing} for precise statements.
\begin{theorem}\label{thm:main-results}
Assume that $\gamma^2<1$ and that $Q\in\Lambda^*$, where
$$
Q=\gamma\rho-\frac{2}{\gamma}\rho^\vee.
$$Let $(\Sigma,g)$ be a closed oriented Riemann surface. Let $\bv=((z_1,v_1),\ldots,(z_n,v_n))\in (T\Sigma)^n$ with pairwise distinct base points, and let $\balpha=(\alpha_1,\ldots,\alpha_n)\in (\Lambda^*)^n,
\,\m=(m_1,\ldots,m_n)\in \Lambda^n$ satisfy
$\alpha_j-Q\in\mathcal C_+,\quad 1\le j\le n$
and
$\sum_{j=1}^n m_j=0,$ where $\mathcal{C}_+$ is the open positive Weyl chamber \eqref{eq:weyl chamber}. Then the following hold true:
\begin{enumerate}
\item[\textup{(i)}]The correlation functions $\la V^g_{(\balpha,\m)}(\bv)\ra_{\Sigma,g}$ are defined as limits of regularized observables. These limits are well defined and do not depend on the auxiliary choices entering the construction, i.e., the choice of cohomology basis, the corresponding closed representatives, or the defect graph used to define the magnetic sector.
\item[\textup{(ii)}]
$\la V^g_{(\balpha,\m)}(\bv)\ra_{\Sigma,g}$ satisfy the axioms of Conformal Field Theory, with conformal weights
$$
\Delta_{(\alpha,m)}
=
\langle \frac{\alpha}{2},\frac{\alpha}{2}-Q\rangle+\frac14|m|^2
$$
and central charge $c=\operatorname{rank}(\g)-6\langle Q,Q\rangle.$ Namely, $\la V^g_{(\balpha,\m)}(\bv)\ra_{\Sigma,g}$ is covariant under the action of diffeomorphisms, the Weyl scaling $g\mapsto e^\rho g,\,\rho\in C^\infty(\Sigma)$, and rotation of spins $v_j\mapsto O_jv_j,\,O_j\in\mathrm{SO}(2).$
\item[\textup{(iii)}] $\la V^g_{(\balpha,\m)}(\bv)\ra_{\Sigma,g}$ satisfy Segal's gluing axioms for conformal field theory under cutting of the surface along analytic parametrized simple curves.
\item[\textup{(iv)}] Consider the case where $(\Sigma,g_0)=(\hat{\C},(\max\{z,1\})^{-4}|\d z|^2).$
Under the neutrality condition $$
2Q-\sum_{j=1}^n\alpha_j=\gamma\sum_{i=1}^r s_i e_i,
\quad
(s_1,\ldots,s_r)\in\N^r,
$$
the structure constant is given by
\begin{equation*}
\begin{aligned}
&C_{\gamma,\bmu}(\balpha,\m):=
\la V^{g_0}_{(\alpha_1,m_1)}(0)V^{g_0}_{(\alpha_2,m_2)}(1)V^{g_0}_{(\alpha_3,m_3)}(\infty)\ra_{\hat{\C},g_0}\\
&=\mathrm{Vol}(\T(\gamma))
\left(\frac{\v_{g_0}(\hat{\C})}{\det'(\Delta_{g_0})}\right)^{r/2}
\prod_{i=1}^r \frac{(-\mu_i)^{s_i}}{s_i!}\,
e^{\bi\pi\la\alpha_2,m_1\ra-\bi\pi\la\alpha_3,m_3\ra}\,
\mathcal I_{\bs}(\balpha,\m),
\end{aligned}
\end{equation*}
where $\mathcal{I_{\bs}}(\balpha,\m)$ is the Dotsenko--Fateev integral \eqref{eq:DF-0-1-infty}.
Assume further that $\g=\mathfrak{sl}_{r+1}$, that $\alpha_1=\kappa\omega_r$, and that $m_1=0$. The structure constant satisfies
\begin{equation}\label{eq:main theorem}
     \begin{aligned}
        C_{\gamma,\bmu}(\kappa\omega_r,\alpha_2,\alpha_3,0,m_2,m_3)^2 =&\mathrm{Vol(\T(\gamma))}^{2}
\left(\frac{\v_{g_0}(\hat{\C})}{\det'(\Delta_{g_0})}\right)^{r}
\prod_{i=1}^r (\pi\mu_i)^{2s_i}\\&\times C^{\fl}_\gamma(\kappa\omega_r,\alpha_2+m_2,\alpha_3+m_3)C^{\fl}_\gamma(\kappa\omega_r,\alpha_2-m_2,\alpha_3-m_3)\end{aligned}
\end{equation}
where $C^{\fl}_\gamma$ denotes the imaginary Fateev--Litvinov constant \eqref{eq:FL formula}.
\end{enumerate}
\end{theorem}

We make some remarks on the proof:
\begin{remark}$\,$
    \begin{enumerate}
        \item The proof of the first three assertions is largely parallel to that of the corresponding statements in \cite{Guillarmou:2023exh}. We do not repeat all the arguments here. Let us only point out that several technical steps are obtained directly from the CILT case after passing to coordinates on $\a$. For instance, the curvature term involving the multivalued magnetic primitive is reduced, after choosing an orthonormal basis of $\a$ or a basis adapted to the coweight lattice, to a finite sum of scalar curvature terms of the type treated in \cite{Guillarmou:2023exh}. Consequently, all the required properties of this term follow componentwise from the corresponding results in the Liouville case. The free-field part of the construction does not introduce any essentially new higher-rank analytic difficulty. The only genuinely Toda-specific ingredient is the interaction term, which is expressed as a finite product of imaginary GMC factors associated with the simple-root projections of the field.
        \item For \eqref{eq:main theorem}, we follow the argument in \cite[Appendix A]{fateev2007correlation}. Since the integrand in \eqref{eq:DF-0-1-infty} is no longer the modulus of, say, $1-x^{(i)}_a$ due to the nonzero magnetic charge $m_2,$ we have to establish a complex twin of \cite[eq. (A.7)]{fateev2007correlation}, this is done by a generalization (see Lemma \ref{lem:complex-twin-A7-final}) of \cite[Lemma 2.4]{neretin2024dotsenko}. Moreover, the parameter range that ensures absolute convergence does not match our requirement, and so one has to first argue in an open subset of parameters so that all the applications of Lemma \ref{lem:complex-twin-A7-final} are valid in the proof of Proposition \ref{prop:closed formula} and use meromorphic continuation to cover the desired parameter range. Following Atiyah's method \cite{atiyah1970resolution}, namely using resolution of singularities \cite{Hironaka}, we  first prove the meromorphic continuation for complex local zeta integral (Theorem \ref{thm:local-complex-zeta}) and obtain a similar result for polynomials (Corollary \ref{cor:global-polynomial-zeta}).
    \end{enumerate}
\end{remark}

\subsection{Future directions}

We conclude by describing several directions which naturally follow from the present
construction.

\begin{itemize}

\item \textbf{Hamiltonian spectrum and conformal bootstrap.} The Segal gluing identities proved in Section~\ref{sec:Segal's axioms} suggest a
natural continuation: to study the spectrum of the Hamiltonian, namely the generator of
the semigroup of annuli, as in \cite{guillarmou2021segal}. In the compactified theory,
this spectral analysis should reflect the decomposition of the boundary field into its
compact zero mode, winding sector, and oscillator modes. It is also naturally connected
with the expected $W$-symmetry of the theory: one can use Cerclé's probabilistic
construction of $W$-algebras \cite{CercleW} to study the action of $W$-currents on
the boundary Hilbert space $\H_{\g,\gamma}$ and their compatibility with the amplitudes
constructed here. Together, the spectral analysis of the Hamiltonian and the
$W$-algebraic structure of the boundary theory form part of the analytic and
representation-theoretic input needed for a conformal bootstrap approach to compactified
imaginary Toda theory.
\item \textbf{Scaling limits of higher-rank lattice models.}
A major motivation for this work is its expected
relation with scaling limits of lattice models with extended symmetries, especially web
models. In such models \cite{lafay2022uq,lafay2024integrability}, one expects observables to be governed not only by Virasoro
symmetry but by higher-rank $W$-symmetries. This suggests that the correlation functions
and amplitudes constructed here should appear as continuum limits of suitable lattice observables.

In rank one, i.e., when $\g=\mathfrak{sl}_2$, this has been achieved via conformal loop ensembles (CLE) \cite{SheffieldCLE,Sheffield2012}, which serve as a candidate for scaling limits of critical loop models. More precisely, in \cite{ang2021integrability}, using techniques from Liouville quantum gravity \cite{SheffieldConformalWelding,duplantier}, the multi-points connectivity of CLE are expressed in terms of the imaginary DOZZ formula. For Toda theories, the corresponding random geometry should be richer than a collection of loops. For example, it is expected that the full interface of the critical 3-states Potts model enjoys $W_3$-symmetry. In particular, the interfaces are not characterized by simply conformal invariance and Markov property. One can employ the Edwards--Sokal coupling in the continuum, established in \cite{MILLER_SHEFFIELD_WERNER_2017}, to obtain the desired interface of critical Potts models, but to describe the interface cleanly requires non-trivial effort.

\end{itemize}
\medskip

\paragraph{\textbf{Organization of the paper.}}
Section~2 collects the geometric, analytic, and Lie-theoretic preliminaries used
throughout the paper. In Section~\ref{sec:igmc}, we recall the
$\mathfrak a$-valued Gaussian free field and the imaginary Gaussian multiplicative
chaos, along with their properties. Section~\ref{sec:path integral} is
devoted to the construction of the path integral and of the electric, magnetic, and
electro-magnetic correlation functions on closed surfaces, together with their conformal
covariance properties. In Section~\ref{sec:riemann sphere}, we specialize to the Riemann
sphere, derive the Dotsenko--Fateev type integral representation of the correlation
functions, and prove the closed formula for the semidegenerate three-point structure
constant in type $A$. Section~\ref{sec:Segal's axioms} proves Segal's gluing axioms by
constructing boundary amplitudes and establishing their compatibility under gluing.
Finally, the appendix contains the complex Dotsenko--Fateev integral identity and the
meromorphic-continuation argument used in the proof of the structure constant formula.
\medskip
\paragraph{\textbf{Acknowledgements.}}
The author thanks Baptiste Cercl\'e for proposing this problem and for his support throughout this work.

\section{Preliminary background and notations}
\subsection{Lie-theoretic conventions}

Let $\g$ be a complex simple Lie algebra of rank $r$ and $\mathfrak{h}$ be the Cartan subalgebra.
Let $\a\cong (\R^r,\langle\cdot,\cdot\rangle)$ be a Euclidean space such that $\mathfrak{h}^*\cong\a\oplus \bi\a$. We consider the simple roots $(e_i)_{i=1,\ldots,r}$ of $\g$ which is a basis of $\a$ and satisfies
$$2\frac{\langle e_i,e_j\rangle}{\langle e_i,e_i\rangle}=A_{ij},$$
where $A$ is the Cartan matrix of $\g$.
Following physics convention, we normalize the inner product $\langle\cdot,\cdot\rangle$
so that the longest root has norm $2$.
The renormalizing constant is given by $2h^\vee$, where $h^\vee$ is called the dual Coxeter number.

We further define by
$$
\omega_i=\sum^r_{j=1}(A^{-1})_{ij}e_j,
$$
the basis $(\omega_i)_{i=1,\ldots,r}$ of $\a^*$ (identified with $\a$) dual to $(e^\vee_i)_{i=1,\ldots,r}$
so that $\langle e_i^\vee,\omega_j\rangle=\delta_{ij}$ where
$e_i^\vee:=2\frac{e_i}{\langle e_i,e_i\rangle}$
is the coroot.
Let $\rho=\sum^r_{i=1}\omega_i$ denote the Weyl vector and its dual $\rho^\vee=\sum^r_{i=1}\omega^\vee_i $ with $\la\omega^\vee_i,e_j\ra=\delta_{ij}$.
The Weyl vector has norm \cite[eq. (47.11)]{FdV}
$$
|\rho|^2=\frac{h^\vee\dim\g}{12}.
$$

The Toda field will take values in the torus $\T=\T(\gamma):=\a/(2\pi\Lambda)$ where $$\Lambda=\Lambda_\gamma:=\gamma^{-1}\bigoplus_{i=1}^r\Z\omega_i^\vee$$ is the scaled coweight lattice. We equip $\T(\gamma)=\a/(2\pi\Lambda)$ with the Haar measure $\d c$ induced by the Lebesgue measure on $\a$. Thus
$$
\int_{\T(\gamma)}\d c
=
\operatorname{Vol}(\T(\gamma))
=
\left(\frac{2\pi}{\gamma}\right)^r\left(\det(A)\prod^r_{i=1}\frac{\la e_i,e_i\ra}{2}\right)^{-1/2}.
$$
\subsection{Riemann surfaces, parametrized boundaries, gluing, and cutting} We briefly recall the geometric setup in \cite{Guillarmou:2023exh}.
\label{subsec:surfaces-boundaries-gluing}
\paragraph{\textbf{Closed surfaces}}
We first fix the geometric conventions used throughout the paper. A closed Riemann surface $(\Sigma,[g])$
will mean a compact connected oriented smooth surface $\Sigma$, without boundary, endowed
with a conformal class of Riemannian metrics $[g]:=\{e^\rho:\rho\in C^\infty(\Sigma)\}$. Equivalently, the conformal class determines
a complex structure $J\in C^\infty(\Sigma;\mathrm{End}(T\Sigma))$, $J^2=-\mathrm{Id}$. In local
holomorphic coordinates $z=x+\bi y$, any compatible metric is of the form $g=e^\rho |\d z|^2,$
and the Hodge star on one-forms is characterized by
$*\d x=\d y,\,*\d y=\d x.$
We denote by $K_g$ the scalar curvature, by $\d \v_g$ the Riemannian volume form, and
by $\Delta_g=\d^*\d$ the non-negative Laplace--Beltrami operator. The Gauss--Bonnet formula reads
\begin{equation}\label{eq:Gauss-Bonnet-closed}
\int_\Sigma K_g\,\d \v_g
=
4\pi\chi(\Sigma),
\qquad
\chi(\Sigma)=2-2\bbg
\end{equation} with $\bbg$ the genus of $\Sigma.$ Moreover, if $\widehat g=e^\rho g$, $\rho\in C^\infty(\Sigma)$, then
\begin{equation}\label{eq:curvature-conformal-change}
K_{\widehat g}
=
e^{-\rho}\left(\Delta_g\rho+K_g\right).
\end{equation}

\paragraph{\textbf{Surfaces with analytic parametrized boundary.}}
Let $\mathbb T:=\{e^{\bi\theta}:\theta\in\mathbb R/2\pi\mathbb Z\}$ be the unit circle. A compact Riemann
surface with analytic parametrized boundary is a compact oriented smooth surface
$\Sigma$ with boundary $\partial\Sigma=\bigsqcup_{j=1}^{\bbb}\partial_j\Sigma,$
together with an atlas $(U_j,\omega_j)$ and smooth diffeomorphisms
$\zeta_j:\mathbb T\longrightarrow \partial_j\Sigma,
\, j=1,\ldots,\bbb,$
such that the following hold. The collection $U_j$ is an open cover $\Sigma$, and $U_j\cap\partial_j\Sigma\neq\emptyset$ if and only if $j\in[1,\bbb]$. There exists $\delta\in(0,1)$ such that, for $j\in[1,\bbb]$,
$$
\omega_j(U_j)=\mathbb A_\delta
:=
\{z\in\C:\delta<|z|\le 1\},
\qquad
\omega_j(\partial_j\Sigma)=\mathbb T.
$$ The transition
maps are holomorphic in the interior, and $\omega_j\circ\zeta_j:\mathbb T\to\mathbb T$
is real analytic, i.e. extends holomorphically to a neighborhood of $\mathbb T$.

The above charts define a complex structure on $\Sigma$. A Riemannian metric $g$ is
compatible with this complex structure if, in every chart,
$(\omega_j^{-1})^*g=e^{\rho_j}|\d z|^2$
for some smooth real-valued function $\rho_j$.
The orientation of $\Sigma$ induces an orientation on each boundary component. We say
that $\partial_j\Sigma$ is outgoing if the orientation induced by
$\zeta_j(e^{\bi\theta})$ agrees with the boundary orientation, and incoming otherwise. We
set
$$
\varsigma_j=
\begin{cases}
-1, & \partial_j\Sigma \text{ is outgoing},\\
+1, & \partial_j\Sigma \text{ is incoming}.
\end{cases}
$$
By replacing a boundary chart by its inverse coordinate if necessary, we may assume that
the parametrizations are adapted to the signs:
$$
\zeta_j(e^{i\theta})
=
\begin{cases}
\omega_j^{-1}(e^{\bi\theta}), & \varsigma_j=-1,\\
\omega_j^{-1}(e^{-\bi\theta}), & \varsigma_j=+1.
\end{cases}
$$

An admissible metric is a compatible metric satisfying
$$
(\omega_j^{-1})^*g=\frac{|\d z|^2}{|z|^2}
\qquad\text{on }\omega_j(U_j),\qquad j=1,\ldots,\bbb.
$$
Then every boundary component is geodesic, has length $2\pi$, and $K_g=0$ near
$\partial\Sigma$. We denote by $\nu$ the inward-pointing unit normal vector field and
by $\d\ell_g$ the induced boundary measure.
\paragraph{\textbf{Gluing and cutting}}
We now recall the gluing convention. Suppose that two parametrized boundary components
$\partial_j\Sigma$ and $\partial_k\Sigma$ have opposite signs, one outgoing and one
incoming. The glued surface is obtained by identifying
$$
\zeta_j(e^{\bi\theta})\sim \zeta_k(e^{\bi\theta}),
\qquad \theta\in\R/2\pi\Z.
$$
Near the glued circle, the complex coordinate is obtained by using $\omega_j$ on one
side and $1/\omega_k$ on the other side. Hence the result is again a Riemann surface.
If the original metric is admissible, the two metrics glue to give a smooth compatible metric on the glued surface.
The remaining boundary components keep their original parametrizations.

Conversely, if $\mathcal C\subset\Sigma^\circ$ is an analytically embedded simple closed
curve and if a neighborhood of $\mathcal C$ is identified holomorphically with an annulus
$\{z:\delta<|z|<\delta^{-1}\}$, with $\mathcal C$ corresponding to $\mathbb T$, then
cutting $\Sigma$ along $\mathcal C$ produces a bordered Riemann surface
$\Sigma_{\mathcal C}:=\Sigma\setminus\mathcal C$
completed by two new analytic boundary components. One of them is outgoing and the other
is incoming. 

\subsection{Regularized determinant of the Laplacian and Green's function}\label{subsec:Laplacian and Green}

Let $(\Sigma,g)$ be a connected compact oriented Riemann surface, with or without boundary.
We denote by $\Delta_g=\d^*\d$ the non-negative Laplace--Beltrami operator. When
$\partial\Sigma=\emptyset$, we consider $\Delta_g$ acting on functions on $\Sigma$.
When $\partial\Sigma\neq\emptyset$, we consider the Dirichlet Laplacian, namely $\Delta_g$
acting on functions vanishing on $\partial\Sigma$.

The spectrum of $\Delta_g$ is discrete:
$$
\mathrm{Sp}(\Delta_g)=\{\lambda_j\}_{j\in\mathcal J},
\qquad
0\le \lambda_0\le \lambda_1\le \lambda_2\le\cdots,
\qquad
\lambda_j\to+\infty,
$$ where $\mathcal{J}=\N_0 $ (resp. $\mathcal{J}=\N$) if $\partial\Sigma=\emptyset$ (resp. if $\partial\Sigma\neq\emptyset$).
The regularized determinant of $\Delta_g$ is defined by zeta regularization. For $\Re(s)\gg1$, we set $\zeta_g(s):=\sum_{\lambda_j>0}\lambda_j^{-s}.$
This series converges for $\Re(s)\gg1$ and admits a meromorphic continuation to $\C$, which is holomorphic at $s=0$. We then define
$\det{}'(\Delta_g):=\exp\bigl(-\zeta'_g(0)\bigr).$
If $\partial\Sigma=\emptyset$ and $\hat g=e^\rho g$ for some $\rho\in C^\infty(\Sigma)$, then the determinant satisfies the
Polyakov formula \cite[eq. (1.31)]{OSGOOD1988148}
\begin{equation}\label{eq:Polyakov-formula}
\log \frac{\det{}'(\Delta_{\hat g})}{\v_{\hat g}(\Sigma)}
=
\log \frac{\det{}'(\Delta_g)}{\v_g(\Sigma)}
-\frac{1}{48\pi}\int_\Sigma\bigl(|d\rho|_g^2+2K_g\rho\bigr)\,\d\v_g.
\end{equation}

We now recall the Green function of the Laplacian. If $\partial\Sigma=\emptyset$, let $\Pi_0:L^2(\Sigma,\d\v_g)\to \ker\Delta_g$
be the orthogonal projection onto the constants. The resolvent operator
$R_g:L^2(\Sigma,\d\v_g)\to L^2(\Sigma,\d\v_g)$
is defined by
$$
\Delta_gR_g=2\pi(\mathrm{Id}-\Pi_0),
\qquad
R_g^*=R_g,
\qquad
R_g1=0.
$$
Its integral kernel is the Green function $G_g$, characterized by
$$
R_gf(x)=\int_\Sigma G_g(x,y)f(y)\,\d\v_g(y),
\qquad f\in L^2(\Sigma,\d\v_g).
$$
Equivalently, $G_g$ is the unique symmetric function on $\Sigma\times\Sigma\setminus \{(x,x):x\in\Sigma\}$
such that, for each fixed $x\in\Sigma$,
\begin{equation}\label{eq:Green-function-equations-closed}
-\Delta_g G_g(x,\cdot)=2\pi\left(\delta_x-\frac{1}{\v_g(\Sigma)}\right),
\qquad
\int_\Sigma G_g(x,y)\,\d\v_g(y)=0.
\end{equation}

If $\partial\Sigma\neq\emptyset$, we denote by $R_{g,D}$ the inverse of the Dirichlet Laplacian: $R_{g,D}:=(2\pi)\Delta_g^{-1}.$
Its integral kernel is the Dirichlet Green function $G_{g,D}$, characterized by
$$
R_{g,D}f(x)=\int_\Sigma G_{g,D}(x,y)f(y)\,\d\v_g(y),
\qquad f\in L^2(\Sigma,\d\v_g),
$$
or equivalently by the conditions that, for each fixed $x\in\Sigma$,
\begin{equation}\label{eq:Green-function-equations-boundary}
\begin{cases}
-\Delta_g G_{g,D}(x,\cdot)=2\pi\delta_x & \text{in }\Sigma,\\
G_{g,D}(x,\cdot)=0 & \text{on }\partial\Sigma.
\end{cases}
\end{equation}
The function $G_{g,\Sigma}$ is symmetric on $\Sigma\times\Sigma\setminus\{(x,x):x\in\Sigma\}$.

In both cases, the Green function has the same local logarithmic singularity along the diagonal:
$$
G(x,y)=\log \frac{1}{d_g(x,y)}+O(1)
\qquad\text{as }y\to x,
$$
where $G$ stands either for $G_g$ or for $G_{g,D}$, depending on the situation.

\subsection{Homology and cohomology}

\paragraph{\textbf{Closed surfaces}}
Consider a closed Riemann surface $(\Sigma,g)$ of genus $\bbg$.
Let $H_1(\Sigma)$ denote the first homology group of $\Sigma$ with value in $\Z$. The algebraic intersection number endows $H_1(\Sigma)$ with a symplectic structure. We will call a basis $([a_i],[b_i])_{i=1,\ldots\bbg}$ of $H_1(\Sigma)$ a \emph{geometric symplectic basis} if the basis $([a_i],[b_i])_{i=1,\ldots\bbg}$ is represented by simple closed curves $a_1,b_1,\ldots,a_{\bbg},b_{\bbg}$ such that $a_j$ intersects $b_j$ transversely in one point, while all other intersections vanish. 

We shall write $\bsigma=(a_1,b_1,\dots,a_\bbg,b_\bbg)$ for such a basis, and also use the notation
$\bsigma=(\sigma_1,\dots,\sigma_{2\bbg})$
when no distinction between $a$- and $b$-cycles is needed.

Let $H^1(\Sigma)$ denote the first de Rham cohomology space and $\harm^1(\Sigma)$ be the
space of harmonic $1$-forms.
The space $H^1(\Sigma)$ is dual to the space $H_1(\Sigma)$ via the pairing
$$
H^1(\Sigma)\times H_1(\Sigma)\rightarrow\R,
\qquad
(\omega,\sigma)\mapsto \la \omega,\sigma\ra=\int_\sigma\omega.
$$
We denote by $H^1(\Sigma;\a)$ the first de Rham cohomology space associated to $\a$-valued
$1$-forms and define the space $H^1_\Lambda(\Sigma;\a)$ of cohomology classes with periods
in $2\pi\Lambda$:
$$
H^1_\Lambda(\Sigma;\a):=
\left\{
\Omega\in H^1(\Sigma;\a):
\int_\sigma\Omega\in2\pi\Lambda
\quad
\forall\sigma\in H_1(\Sigma)
\right\}.
$$

\begin{lemma}\label{geometric symplectic basis}
Let $\bm{\sigma}=(\sigma_1,\ldots,\sigma_{2\bbg})$ be a basis of $H_1(\Sigma)$.
Then there exist $2\bbg$ independent closed smooth real-valued $1$-forms
$\eta_1,\ldots,\eta_{2\bbg}$ such that, for every $1\leq j,k\leq2\bbg,$
$$
\int_{\sigma_j}\eta_k=2\pi\delta_{jk}.
$$
Moreover, we can identify $\Lambda^{2\bbg}$ with $\H^1_\Lambda(\Sigma;\a)$ by sending
$(\lambda_1,\ldots,\lambda_{2\bbg})\in\Lambda^{2\bbg}$ to
$$
\sum^{2\bbg}_{k=1}\eta_k\otimes\lambda_k\in H^1_\Lambda(\Sigma;\a).
$$
\end{lemma}
\begin{proof}
The first claim is the pairing between $H_1(\Sigma)$ and $H^1(\Sigma)$, and the second claim follows directly from the definition of $H^1_\Lambda(\Sigma;\a)$.
\end{proof}

\paragraph{\textbf{Surfaces with boundary}}
We now turn to surfaces with boundary. Let $(\Sigma,[g])$ be a compact connected Riemann surface of genus $\bbg$, with non-empty analytic boundary $\partial\Sigma=\bigsqcup_{j=1}^{\bbb}\partial_j\Sigma$. Assume that the boundary components are oriented positively with respect to the orientation of $\Sigma$, and we write $c_j:=\partial_j\Sigma,\, j=1,\ldots,\bbb$ . As elements of $H_1(\Sigma)$, they satisfy $[c_1]+\cdots+[c_{\bbb}]=[\partial\Sigma]=0.$

The homology group $H_1(\Sigma)$ is represented by oriented closed curves in $\Sigma$ and is isomorphic to $\Z^{2\bbg+\bbb-1}$. We call a basis $$(a_1,b_1,\ldots,a_{\bbg},b_{\bbg},c_1,\ldots,c_{\bbb-1})$$  a \emph{canonical geometric basis} of $H_1(\Sigma)$ if $a_i,b_i$ are simple closed curves contained in $\Sigma^\circ$, chosen as in the closed case, and $c_1,\ldots,c_{\bbb-1}$ are boundary cycles, after a possible renumbering of the boundary components.

The relative homology group $H_1(\Sigma,\partial\Sigma)$ is represented by oriented closed curves in $\Sigma^\circ$, together with oriented curves whose endpoints lie on $\partial\Sigma$. It is also isomorphic to $\Z^{2\bbg+\bbb-1}$. We call a basis $$
(a_1,b_1,\ldots,a_{\bbg},b_{\bbg},d_1,\ldots,d_{\bbb-1})$$ a \emph{canonical geometric basis} of $H_1(\Sigma,\partial\Sigma)$ 
if $a_i,b_i$ are as above, and $d_1,\ldots,d_{\bbb-1}$ are pairwise disjoint oriented simple arcs with endpoints on $\partial\Sigma$, disjoint from the interior cycles $(a_i,b_i)$, such that the graph whose vertices are the boundary components and whose edges are the arcs $d_j$ is connected and has no cycle. 

We denote by $H^1(\Sigma)$ and $H^1(\Sigma,\partial\Sigma)$ respectively the absolute and relative de Rham cohomology spaces. Absolute cohomology is dual to absolute homology, and relative cohomology is dual to relative homology, through the pairings
\begin{equation}\label{eq:absolute-cohomology-duality}
H^1(\Sigma)\times H_1(\Sigma)\to\R,
\qquad
(\omega,\sigma)\mapsto\int_\sigma\omega,    
\end{equation}
and
\begin{equation}\label{eq:relative-cohomology-duality}
    H^1(\Sigma,\partial\Sigma)\times H_1(\Sigma,\partial\Sigma)\to\R,
\qquad
(\omega,\sigma)\mapsto\int_\sigma\omega.
\end{equation}
Equivalently, a relative cohomology class may be represented by a closed $1$-form whose pullback to $\partial\Sigma$ vanishes. In what follows, we often choose relative representatives compactly supported in $\Sigma^\circ$.

For $\a$-valued forms, we define 
$$
H^1_\Lambda(\Sigma;\a)
:=
\left\{
\Omega\in H^1(\Sigma;\a):
\int_\sigma\Omega\in2\pi\Lambda
\quad
\forall\sigma\in H_1(\Sigma)
\right\},
$$
and
$$
H^1_\Lambda(\Sigma,\partial\Sigma;\a)
:=
\left\{
\Omega\in H^1(\Sigma,\partial\Sigma;\a):
\int_\sigma\Omega\in2\pi\Lambda
\quad
\forall\sigma\in H_1(\Sigma,\partial\Sigma)
\right\}.
$$

We shall use the following $\a$-valued versions of the basic homological lemmas from \cite[Section 3.8]{Guillarmou:2023exh}. whose proofs can be obtained componentwise: after
fixing an orthonormal basis $(\varepsilon_1,\ldots,\varepsilon_r)$ of $\a$, one writes
each $\a$-valued form as a sum of scalar forms and applies the corresponding scalar
statement in \cite{Guillarmou:2023exh} to each component.
\begin{lemma}
\label{lem:a-valued-absolute-relative-exactness}
$\,$
\begin{enumerate}
\item Fix a basis of $H_1(\Sigma)$. Let
$\Omega_1,\Omega_2\in C^\infty_{\mathrm{abs}}(\Sigma,\Lambda^1\Sigma\otimes\a)$
be closed $\a$-valued absolute $1$-forms such that
$\int_\sigma \Omega_1=\int_\sigma \Omega_2$
for every element $\sigma$ in the basis of $H_1(\Sigma)$. Then there exists $f\in C^\infty(\Sigma;\a),
\,\partial_\nu f|_{\partial\Sigma}=0,$
such that
$$
\Omega_1=\Omega_2+\d f.
$$

\item Fix a basis of $H_1(\Sigma,\partial\Sigma)$. Let
$\Omega_1,\Omega_2\in C^\infty_{\mathrm{rel}}(\Sigma,\Lambda^1\Sigma\otimes\a)$
be closed $\a$-valued relative $1$-forms such that
$\int_\sigma \Omega_1=\int_\sigma \Omega_2$
for every element $\sigma$in the basis of $H_1(\Sigma,\partial\Sigma)$. Then
there exists
$f\in C^\infty(\Sigma;\a),
\,f|_{\partial\Sigma}=0,$
such that
$$
\Omega_1=\Omega_2+\d f.
$$
\end{enumerate}
\end{lemma}

\begin{lemma}
\label{lem:relative-period-basis}
Let $\Sigma$ be a compact connected Riemann surface of genus $\bbg$ with
$\bbb\ge1$ boundary components, and let $\bsigma=(\sigma_1,\ldots,\sigma_{2\bbg+\bbb-1})$
be a basis of $H_1(\Sigma,\partial\Sigma)$. Then, for every
$\blambda^c=(\lambda^c_1,\ldots,\lambda^c_{2\bbg+\bbb-1})
\in\Lambda^{2\bbg+\bbb-1},$
there exists a closed smooth $\a$-valued $1$-form
$\Omega^c_{\blambda^c},$
compactly supported in $\Sigma^\circ$, such that
$$
\int_{\sigma_k}\Omega^c_{\blambda^c}
=
2\pi\lambda^c_k,
\qquad
k=1,\ldots,2\bbg+\bbb-1.
$$
Moreover, the assignment $\blambda^c\mapsto [\Omega^c_{\blambda^c}]$
identifies $\Lambda^{2\bbg+\bbb-1}$ with $H^1_\Lambda(\Sigma,\partial\Sigma;\a)$

If, in addition, $\bsigma=(a_1,b_1,\ldots,a_{\bbg},b_{\bbg},d_1,\ldots,d_{\bbb-1})$
is a canonical geometric basis, then the representatives may be chosen so that, for every $(0,\ldots,0,\blambda^c_\partial)
\in \Lambda^{2\bbg}\times\Lambda^{\bbb-1},$
there exists $f_{\blambda^c_\partial}\in C^\infty(\Sigma;\a),$
locally constant near $\partial\Sigma$ such that
$\Omega^c_{(0,\ldots,0,\blambda^c_\partial)}=
\d f_{\blambda^c_\partial}.$

\end{lemma}

\begin{remark}[Glued homology bases]
\label{rem:glued-homology-bases}
For later use in the gluing arguments, we adopt the convention of
\cite[Lemmas 3.7 \& 3.8]{Guillarmou:2023exh} for gluing absolute and relative homology bases. More precisely, if two parametrized boundary components of
$\Sigma_1$ and $\Sigma_2$ are identified with opposite orientations and $\Sigma=\Sigma_1\#\Sigma_2,$
then, given canonical geometric bases $\bsigma_i^{\mathrm{abs}}\subset H_1(\Sigma_i),$ and $\bsigma_i^{\mathrm{rel}}\subset H_1(\Sigma_i,\partial\Sigma_i),$
for $i=1,2$,
we denote by $\bsigma^{\mathrm{abs}}=
\bsigma_1^{\mathrm{abs}}\#\bsigma_2^{\mathrm{abs}},$
$\bsigma^{\mathrm{rel}}=
\bsigma_1^{\mathrm{rel}}\#\bsigma_2^{\mathrm{rel}}$
the corresponding glued bases on $\Sigma$.
If $\Omega^{i,c}_{\blambda_i^c}
\in H^1_\Lambda(\Sigma_i,\partial\Sigma_i;\a),\, i=1,2,$
are compactly supported relative representatives, then after extending them by zero
outside $\Sigma_i$, we write $\Omega^c_{\blambda^c}
=
\Omega^{1,c}_{\blambda_1^c}
+
\Omega^{2,c}_{\blambda_2^c}$
for the corresponding representative on the glued surface $\Sigma$, with
$\blambda^c$ denoting the reindexed lattice variable determined by the glued basis. A similar discussion applies to self-gluing. We refer to \cite[Section 3.9]{Guillarmou:2023exh} for further details.
\end{remark}
\subsection{Equivariant \texorpdfstring{$\a$}{a}-valued functions and distributions}\label{sec:equivariant}
Let $\Sigma $ be a surface with or without boundary, and let $\beta_1=\dim H_1(\Sigma)$ denote its first Betti number.
Fix a base point $x_0\in \Sigma$, and let $\pi:\widetilde{\Sigma}_{x_0}\to \Sigma$
be the universal cover of $\Sigma$, with distinguished lift $\widetilde x_0$ of $x_0$.
We view $\widetilde{\Sigma}_{x_0}$ as the space of homotopy classes, with fixed endpoints,
of continuous paths starting at $x_0$. The fundamental group $\pi_1(\Sigma,x_0)$ acts on
$\widetilde{\Sigma}_{x_0}$ by deck transformations.

Let $\Omega$ be a smooth closed $\a$-valued $1$-form on $\Sigma$. For each
$\widetilde x\in \widetilde{\Sigma}_{x_0}$, define
\begin{equation}\label{eq:primitive-cover}
I_{x_0}(\Omega)(\widetilde x)
:=
\int_{\alpha_{x_0,x}}\Omega,
\qquad x=\pi(\widetilde x),
\end{equation}
where $\alpha_{x_0,x}$ is any $C^1$ path from $x_0$ to $x$ whose lift joins
$\widetilde x_0$ to $\widetilde x$. Since $\Omega$ is closed, the quantity
$I_{x_0}(\Omega)(\widetilde x)$ depends only on $\widetilde x$, and therefore defines
a smooth $\a$-valued function on $\widetilde{\Sigma}_{x_0}$. Moreover, $\d I_{x_0}(\Omega)=\pi^*\Omega.$

When $\Omega\in H^1_\Lambda(\Sigma;\a)$, the primitive $I_{x_0}(\Omega)$ is not
in general invariant under deck transformations, but its variation is controlled by the
period lattice $2\pi\Lambda$.

\begin{lemma}\label{lem:equivariant-primitive}
Let $\Omega\in H^1_\Lambda(\Sigma;\a)$. Then for every
$\gamma\in\pi_1(\Sigma,x_0)$ and every $\widetilde x\in \widetilde{\Sigma}_{x_0}$,
$$
I_{x_0}(\Omega)(\gamma\cdot\widetilde x)-I_{x_0}(\Omega)(\widetilde x)\in 2\pi\Lambda.
$$
Consequently, $I_{x_0}(\Omega)$ defines a smooth map from $\Sigma$ to the torus $\T=\a/(2\pi\Lambda).$

\end{lemma}

\begin{proof}
Let $\widetilde x\in \widetilde{\Sigma}_{x_0}$ and $h\in\pi_1(\Sigma,x_0)$.
Choose any piecewise $C^1$ path $\beta$ in $\widetilde{\Sigma}_{x_0}$ from
$\widetilde x$ to $h\cdot \widetilde x$. Then
$$
I_{x_0}(\Omega)(h\cdot\widetilde x)-I_{x_0}(\Omega)(\widetilde x)
=
\int_\beta \pi^*\Omega
=
\int_{\pi\circ\beta}\Omega.
$$
Since $\pi\circ\beta$ is a closed loop in $\Sigma$ and
$\Omega\in H^1_\Lambda(\Sigma;\a)$, the right-hand side belongs to $2\pi\Lambda$.

\end{proof}

The previous lemma shows that the primitive $I_{x_0}(\Omega)$ is not invariant under deck
transformations in general, but its increments are constrained to lie in $2\pi\Lambda$. This
naturally leads to considering $\a$-valued functions on $\widetilde\Sigma_{x_0}$ whose monodromy
is prescribed by the lattice $2\pi\Lambda$.

We let $\Gamma:=\pi_1(\Sigma,x_0)$ and define
$$
C^\infty_\Gamma(\widetilde\Sigma_{x_0};\a)
:=
\left\{
u\in C^\infty(\widetilde\Sigma_{x_0};\a):\forall h\in \Gamma,
h^*u-u\in 2\pi\Lambda
\right\}.
$$
We note that each $u\in C^\infty_\Gamma(\widetilde\Sigma_{x_0};\a)$ induces a group morphism $$\chi_u:\Gamma\rightarrow2\pi\Lambda,\quad \chi_u(h):=h^*u-u.$$ Moreover, each morphism corresponds to an element $\Omega_\chi$ in $H^1_\Lambda(\Sigma;\a)$, obtained from first fixing a basis $(\Omega_1,\ldots\Omega_{\beta_1})$ of $H^1_\Lambda(\Sigma;\a)$, then choosing $\blambda=(\lambda_1,\ldots,\lambda_{\beta_1})\in\Lambda^{\beta_1}$ so that $\chi(h)=\int_h\Omega_{\blambda}$ for all $h\in\Gamma$, and finally setting $\Omega_\chi:=\Omega_{\blambda}$. We write $\chi_{\blambda}$ for the morphism associated to $\Omega_{\blambda}.$

For each group morphism $\chi:\Gamma\to 2\pi\Lambda,$
we define
$$
C^\infty_\chi(\widetilde\Sigma_{x_0};\a)
:=
\left\{
u\in C^\infty_\Gamma(\widetilde\Sigma_{x_0};\a):\forall h\in \Gamma,\;
h^*u-u=\chi(h)
\right\}.
$$
Then for $u\in C^\infty_{\chi_{\blambda}}(\widetilde\Sigma_{x_0};\a)$, $u-I_{x_0}(\Omega_{\blambda})$ is $\Gamma$-invariant and so descends to a smooth function on $\Sigma.$ This implies that we can rewrite the space $C^\infty_{\chi_{\blambda}}(\widetilde\Sigma_{x_0};\a)$ as $$C^\infty_{\chi_{\blambda}}(\widetilde\Sigma_{x_0};\a)=\{\pi^* f+I_{x_0}(\Omega_{\blambda}):f\in C^\infty(\Sigma;\a)\}$$ 
and also the space $C^\infty_\Gamma(\widetilde\Sigma_{x_0};\a)$ as
$$C^\infty_\Gamma(\widetilde\Sigma_{x_0};\a)=\cup_{\blambda\in\Lambda^{\beta_1}}C^\infty_{\chi_{\blambda}}(\widetilde\Sigma_{x_0};\a).$$ That is, each $u\in C^\infty_\Gamma(\widetilde\Sigma_{x_0};\a)$ can be uniquely represented by $\pi^*f+I_{x_0}(\Omega_{\blambda})$ for some $f\in C^\infty(\Sigma;\a)$ and $\blambda\in\Lambda^{\beta_1}.$

Next we define the equivariant Sobolev spaces similarly. For $s\in\R$, we define the $\Z$-module $$H^s_{\Gamma}(\widetilde\Sigma_{x_0};\a):=\{f\in H^s_{\mathrm{loc}}(\widetilde\Sigma_{x_0};\a):\forall h\in\Gamma, h^*f-f\in 2\pi\Lambda\},$$ where $H^s_{\mathrm{loc}}(\widetilde\Sigma_{x_0};\a)$ denotes the space of $\a$-valued distributions
which belong to $H^s(U;\a)$ on every relatively compact open set
$U\subset \widetilde\Sigma_{x_0}$. Analogous to smooth $\a$-valued functions, each $u\in H^s_\Gamma(\widetilde\Sigma_{x_0};\a)$ can be uniquely represented by $u=\pi^* f+I_{x_0}(\Omega_{\blambda})$ for some $f\in H^s(\Sigma;\a)$ and some $\blambda\in \Lambda^{\beta_1}$.

\subsection{Magnetic backgrounds and boundary windings}\label{subsec:magnetic-backgrounds}
Let $\Sigma$ be a Riemann surface of genus $\bbg$ with or without boundary. If $\partial\Sigma\neq\emptyset$, we write
$\partial\Sigma=\bigsqcup_{\ell=1}^{\bbb}\partial_\ell\Sigma$
for its boundary components and assign to each boundary component a sign
$\varsigma_\ell\in\{\pm1\}$, with the convention that
$\varsigma_\ell=-1$ for an outgoing boundary component and
$\varsigma_\ell=+1$ for an incoming boundary component. Let $\z=(z_1,\ldots,z_{\nm})$ be $\nm$ pairwise distinct points in the interior of $\Sigma$, to which we attach \emph{magnetic charges} $\m=(m_1,\ldots,m_{\nm})\in \Lambda^{\nm}$ such that $\sum_{i=1}^{\nm}m_i=0.$
We denote by $U_1,\ldots,U_{\nm}$ neighborhoods of $z_1,\ldots,z_{\nm}$ and by $\psi_j:\D\rightarrow U_j$ biholomorphisms such that $\psi_j(0)=z_j$. In the unit disk $\D$, we use the coordinate $z=re^{i\theta}$ and denote by $\d\theta$ the closed and coclosed 1-form in the pointed disk $\D\setminus\{0\}$. If $\partial\Sigma\neq\emptyset$, we also choose collar neighborhoods $V_\ell$ of $\partial_\ell\Sigma$ and biholomorphisms $\psi_{\nm+\ell}:\{z\in\C:\delta<|z|\le 1\}\to V_\ell$ with $ \psi_{\nm+\ell}(\mathbb T)=\partial_\ell\Sigma$ for some $\delta<1$. Let $\beta_1:=\dim H_1(\Sigma)$ be the first Betti number.

We now construct the magnetic background associated with the marked-point charges and, in the boundary case, with the boundary winding data.
\begin{proposition}\label{prop:magnetic-background}
    Let $\bsigma=(\sigma_1,\ldots,\sigma_{\beta_1})\subset H_1(\Sigma)$ be a basis realized by closed curves avoiding the disks $U_j$ such that $\sigma_\l\cap\partial\Sigma=\emptyset$ for $\l=1,\ldots,2\bbg$ and such that $\sigma_{2\bbg+\l}=\partial\Sigma_\l$ for $\l=1,\ldots,\bbb-1$. Let $\blambda=(\lambda_1,\ldots,\lambda_{\bbb})\in\Lambda^{\bbb}$ and $\m=(m_1,\ldots, m_{n_{\m}})\in\Lambda^{n_{\m}}$ satisft $\sum_{\l=1}^{\bbb}\varsigma_\l\lambda_\l+\sum_{j=1}^{n_\m}m_j=0$. Then there exists a smooth $\a$-valued closed 1-form on $\Sigma_\z:=\Sigma\setminus\{\z\}$, denoted by $\nu_{\z,\m}$ if $\partial\Sigma=\emptyset$, resp. $\nu_{\z,\m,\blambda}$ if $\partial\Sigma\neq\emptyset$, such that $\iota_{\partial\Sigma}(i_\nu\nu_{\z,\m,\blambda})=0$ with $\nu$ the inward-pointing unit normal vector field along
$\partial\Sigma$, \begin{equation*}
    \begin{cases}
        \psi^*_j(\nu_{\z,\m}|_{U_j})=m_j\d\theta, &\partial\Sigma=\emptyset,\\
        \psi^*_j(\nu_{\z,\m,\blambda}|_{U_j})=m_j\d\theta,
        &\partial\Sigma\neq\emptyset
    \end{cases}
\end{equation*} and for $j=1,\ldots,2\bbg$ and $\l=1,\ldots,\bbb$, $$\begin{cases}
    \int_{\sigma_j}\nu_{\z,\m}=0, &\partial\Sigma=\emptyset,\\ \int_{\sigma_j}\nu_{\z,\m,\blambda}=0 \quad and \quad \int_{\partial\Sigma_\l}\nu_{\z,\m,\blambda}=2\pi\varsigma_\l\lambda_\l, 
    &\partial\Sigma\neq\emptyset.\\ 
\end{cases}$$ 
The form $\nu_{\z,\m,\blambda}$ can be chosen so that $\psi^*_{n+\l}\nu_{\z,\m,\blambda}=-\lambda_\l\d\theta$ near $\T.$ The form $\nu_{\z,\m}$ (resp. $\nu_{\z,\m,\blambda}$) satisfies $\d^*\nu_{\z,\m} \in C^\infty(\Sigma;\a)\cap C^\infty_{\mathrm{c}}(\Sigma_\z;\a)$ (resp. $\d^*\nu_{\z,\m,\blambda})\in C^\infty(\Sigma;\a)\cap C^\infty_{\mathrm{c}}(\Sigma_\z;\a)$). Moreover, in the distributional sense, 
\begin{equation*}
    \begin{cases}
        \d\nu_{\z,\m}=-2\pi\sum^{n_\m}_{j=1}m_j\delta_{z_j}, &\partial\Sigma=\emptyset,\\\d\nu_{\z,\m,\blambda}=-2\pi\sum^{n_\m}_{j=1}m_j\delta_{z_j}&\partial\Sigma\neq\emptyset,
    \end{cases}
\end{equation*} with $\delta_z$ the Dirac mass at $z.$ There is a unique closed and coclosed $\a$-valued $1$-form $\nu^{\mathrm{h}}_{\z,\m}$ if $\partial\Sigma=\emptyset$ and $\nu^{\mathrm{h}}_{\z,\m,\blambda}$ if $\partial\Sigma\neq\emptyset$, on $\Sigma_\z$ such that
\begin{equation}\label{eq:harmonic-magnetic-winding}
    \begin{cases}
        \nu^{\mathrm{h}}_{\z,\m}-\nu_{\z,\m}=\d f_\m,&\partial\Sigma=\emptyset,\\\nu^{\mathrm{h}}_{\z,\m,\blambda}-\nu_{\z,\m,\blambda}=\d f_{\m,\blambda}, &\partial\Sigma\neq\emptyset
    \end{cases}
\end{equation}
for some $f_\m\in C^\infty(\Sigma;\a)$ if $\partial\Sigma=\emptyset$ and $f_{\m,\blambda}\in C^\infty(\Sigma;\a)$ with $f_{\m,\blambda}|_{\partial\Sigma}=0$ if $\partial\Sigma\neq\emptyset.$

\end{proposition}

\begin{proof}
Writing the lattice elements in the coweight basis $$m_j=\gamma^{-1}\sum_{q=1}^r m_{j,q}\omega_q^\vee,
\quad\lambda_\ell=\gamma^{-1}\sum_{q=1}^r \lambda_{\ell,q}\omega_q^\vee,$$with $m_{j,q},\lambda_{\ell,q}\in\mathbb Z$, we apply \cite[Proposition 3.10]{Guillarmou:2023exh} with
$R=\gamma^{-1}$, scalar magnetic charges
$(m_{1,q},\ldots,m_{\nm,q})$, and scalar boundary windings
$(\lambda_{1,q},\ldots,\lambda_{\bbb,q})$. This gives scalar forms
$\nu^{(q)}_{\z,\m}$, respectively $\nu^{(q)}_{\z,\m,\blambda}$. Now setting $\nu_{\z,\m}:=
\sum_{q=1}^r \nu^{(q)}_{\z,\m}\,\omega_q^\vee,
\,\nu_{\z,\m,\blambda}
:=
\sum_{q=1}^r \nu^{(q)}_{\z,\m,\blambda}\,\omega_q^\vee $ proves all the assertions.
\end{proof}

We note that the form $\nu_{\z,\m}\in L^1(\Sigma)$ (resp. $\nu_{\z,\m,\blambda}\in L^1(\Sigma)$) does not belong to $L^2(\Sigma)$, but we can still define an $L^2$-norm via regularization: for $\omega=\sum_{\l=1}^r\omega_\l\eps_\l\in C^\infty(\Sigma\setminus\{\z\},T^*\Sigma\otimes\a)$ with $(\eps_\l)_{\l=1,\ldots,r}$ an orthonormal basis of $\a$, let $$\norm{\omega}^2_{g,\eps}:=\sum^r_{\l=1}\int_{\Sigma_{\z,\eps,g}}\omega_\l\wedge*\omega_\l,$$ where $\Sigma_{\z,\eps,g}:=\Sigma\setminus\bigcup_{j=1}^{\nm}B_g(z_j,\eps)$ with $B_g(z_j,\eps)$ the geodesic ball centered at $z_j$ with radius $\eps$ with respect to the metric $g.$
\begin{lemma}\label{lem:toda-reg-norm}
Let $\omega=\nu_{\z,\m}$ if $\partial\Sigma=\emptyset$ and $\omega=\nu_{\z,\m,\blambda}$ if $\partial\Sigma\neq\emptyset$. Then as $\eps\rightarrow0$, the following limit exists \begin{equation}\label{eq:magnetic-reg-norm}
    \|\omega\|_{g,0}^2
:=
\lim_{\eps\to0}
(\|\omega\|^2_{g,\eps}+
2\pi\log\eps\sum_{j=1}^{\nm}|m_j|^2
).
\end{equation}
Moreover, if $g'=e^\rho g$, with $\rho\in C^\infty(\Sigma)$, then
$$
\|\omega\|_{g',0}^2
=
\|\omega\|_{g,0}^2
+
\pi\sum_{j=1}^{\nm}|m_j|^2\rho(z_j).
$$
The same statements hold for the harmonic representatives
$\nu^{\mathrm h}_{\z,\m}$ and $\nu^{\mathrm h}_{\z,\m,\blambda}$.
\end{lemma}
\begin{proof}
Fix an orthonormal basis $(\varepsilon_1,\ldots,\varepsilon_r)$ of $\a$, and write
$\omega=\sum_{\l=1}^r \omega^{(\l)}\varepsilon_\l,$ and
$m_j=\sum_{\l=1}^r m_{j,\l}\varepsilon_\l .$
Although the coefficients $m_{j,\l}$ need not be integers, the proof of
\cite[Lemma 3.11]{Guillarmou:2023exh} applies verbatim to real coefficients, since it only uses the local expansion $\nu^{(\l)}=m_{j,\l}\d\theta+O(1)$ near $z_j$. Hence each scalar component admits a finite regularized norm and satisfies
$$
\|\nu^{(\l)}\|_{g',0}^2
=
\|\nu^{(\l)}\|_{g,0}^2
+
\pi\sum_{j=1}^{\nm}m_{j,\l}^2\rho(z_j).
$$
Summing over $\l$ and using orthonormality gives the claimed existence and conformal covariance formula.
\end{proof}

\subsection{Equivariant \texorpdfstring{$\a$}{a}-valued functions and distributions with marked points}

Let $(\Sigma,g)$ be a closed Riemann surface, let $\z=(z_1,\ldots,z_{\nm})$ be pairwise distinct marked points, and set $\Sigma_{\z}:=\Sigma\setminus\{z_1,\ldots,z_{\nm}\}.$
Fix a base point $x_0\in\Sigma_{\z}$, and let $\pi:\widetilde{\Sigma}_{\z,x_0}\to \Sigma_{\z}$ be the universal cover, with distinguished lift $\widetilde x_0$ of $x_0$.
We denote by $\Gamma_{\z}:=\pi_1(\Sigma_{\z},x_0)$ the first fundamental group.

We first introduce the equivariant $L^2$ space on the universal cover. We set
$$
L^2_{\Gamma_{\z}}(\widetilde{\Sigma}_{\z,x_0};\a)
:=
\left\{
u\in L^2_{\mathrm{loc}}(\widetilde{\Sigma}_{\z,x_0};\a):
\forall \gamma\in\Gamma_{\z},\;
\gamma^*u-u\in 2\pi\Lambda
\right\}.
$$
For each group morphism $\chi:\Gamma_{\z}\to 2\pi\Lambda,$
we also write
$$
L^2_{\chi}(\widetilde{\Sigma}_{\z,x_0};\a)
:=
\left\{
u\in L^2_{\mathrm{loc}}(\widetilde{\Sigma}_{\z,x_0};\a):
\forall \gamma\in\Gamma_{\z},\;
\gamma^*u-u=\chi(\gamma)
\right\}.
$$
Thus $L^2_{\Gamma_{\z}}(\widetilde{\Sigma}_{\z,x_0};\a)=\cup_{\chi}L^2_{\chi}(\widetilde{\Sigma}_{\z,x_0};\a).$

Now let $\m=(m_1,\ldots,m_{\nm})\in\Lambda^{\nm},$ with $\sum_{j=1}^{\nm}m_j=0,$
and let $\nu_{\z,\m}$ be the magnetic background $1$-form introduced in Proposition \ref{prop:magnetic-background}.
Since $\Omega_{\blambda}+\nu_{\z,\m}$, $\blambda\in\Lambda^{2\bbg}$ is closed on $\Sigma_{\z}$, we can define its primitive on
the universal cover by
$$
I_{x_0}(\Omega_{\blambda}+\nu_{\z,\m})(\widetilde x)
:=
\int_{\alpha_{x_0,x}}(\Omega_{\blambda}+\nu_{\z,\m}),
\qquad x=\pi(\widetilde x),
$$
where $\alpha_{x_0,x}$ is any piecewise $C^1$ path in $\Sigma_{\z}$ from $x_0$ to $x$
whose lift joins $\widetilde x_0$ to $\widetilde x$.

By the period condition, this primitive is equivariant modulo $2\pi\Lambda$, and it defines a group morphism
$$
\chi_{\blambda,\m}:\Gamma_{\z}\to 2\pi\Lambda,
\qquad
\chi_{\blambda,\m}(\gamma)
:=
\int_{\gamma}(\Omega_{\blambda}+\nu_{\z,\m}).
$$
Moreover, if $u\in L^2_{\chi_{\blambda,\m}}(\widetilde{\Sigma}_{\z,x_0};\a),$
then $u-I_{x_0}(\Omega_{\blambda}+\nu_{\z,\m})$
is $\Gamma_{\z}$-invariant, hence it descends to a function $f$ on $\Sigma_{\z}$. Since the primitive
$I_{x_0}(\Omega_{\blambda}+\nu_{\z,\m})$ is smooth on the universal cover, this gives $u=\pi^*f+I_{x_0}(\Omega_{\blambda}+\nu_{\z,\m}),$ for some $f\in L^2(\Sigma;\a)$.
Conversely, every function of this form belongs to $L^2_{\chi_{\blambda,\m}}(\widetilde{\Sigma}_{\z,x_0};\a).$
Therefore
$$
L^2_{\chi_{\blambda,\m}}(\widetilde{\Sigma}_{\z,x_0};\a)
=
\left\{
\pi^*f+I_{x_0}(\Omega_{\blambda}+\nu_{\z,\m}):\;
f\in L^2(\Sigma;\a)
\right\},
$$
and
$$
L^2_{\Gamma_{\z}}(\widetilde{\Sigma}_{\z,x_0};\a)
=
\bigcup_{\substack{\blambda\in\Lambda^{2\bbg}\\
\m\in\Lambda^{\nm},\ \sum_{j=1}^{\nm}m_j=0}}
\left\{
\pi^*f+I_{x_0}(\Omega_{\blambda}+\nu_{\z,\m}):\;
f\in L^2(\Sigma;\a)
\right\}.
$$

We now define the corresponding Sobolev space. For $s\in(-\frac{1}{2},0)$, we set
$$
H^s_{\Gamma_{\z}}(\widetilde{\Sigma}_{\z,x_0};\a)
:=
L^2_{\Gamma_{\z}}(\widetilde{\Sigma}_{\z,x_0};\a)+\pi^*H^s(\Sigma;\a).
$$
Equivalently,
$$
H^s_{\Gamma_{\z}}(\widetilde{\Sigma}_{\z,x_0};\a)
=
\bigcup_{\substack{\blambda\in\Lambda^{2\bbg}\\
\m\in\Lambda^{\nm},\ \sum_{j=1}^{\nm}m_j=0}}
\left\{
\pi^*f+I_{x_0}(\Omega_{\blambda}+\nu_{\z,\m}):\;
f\in H^s(\Sigma;\a)
\right\}.
$$

In other words, an element of $H^s_{\Gamma_{\z}}(\widetilde{\Sigma}_{\z,x_0};\a)$
is obtained by taking a Sobolev distribution $f\in H^s(\Sigma;\a)$ on the closed surface,
pulling it back to the universal cover, and then adding the fixed multivalued smooth background
$I_{x_0}(\Omega_{\blambda}+\nu_{\z,\m})$ carrying the topological and magnetic monodromy.

\section{Imaginary Gaussian multiplicative chaos}\label{sec:igmc}

In this section we recall the $\a$-valued Gaussian free field on a compact Riemann surface $(\Sigma,g)$,
with or without boundary, introduce its regularizations, and define the associated imaginary
Gaussian multiplicative chaos. This is the probabilistic object that will later be used to make
sense of the Toda interaction term.

We recall the notations from Section \ref{subsec:Laplacian and Green}.
\subsection{The \texorpdfstring{$\a$}{a}-valued Gaussian free field}

Assume first that $\partial\Sigma=\emptyset$. Let $(\varphi_j)_{j\in\N_0}$ be an orthonormal basis of
real-valued eigenfunctions of $\Delta_g$ in $L^2(\Sigma,\d\v_g)$, with eigenvalues
$$
0=\lambda_0<\lambda_1\le \lambda_2\le \cdots,
$$
and with
$\varphi_0=\v_g(\Sigma)^{-1/2}.$
Let $(\varepsilon_1,\ldots,\varepsilon_r)$ be an orthonormal basis of $\a$, and let
$(a_{j,\ell})_{j\ge1,\,1\le \ell\le r}$ be i.i.d.\ standard real Gaussian random variables.
We define the $\a$-valued Gaussian free field with vanishing mean in the metric $g$ by
\begin{equation}\label{eq:aGFF-series}
\x_g
:=
\sqrt{2\pi}\sum_{j\ge1}\sum_{\ell=1}^r
\frac{a_{j,\ell}}{\sqrt{\lambda_j}}\,\varphi_j\,\varepsilon_\ell.
\end{equation}
This series converges almost surely in $H^s(\Sigma;\a)$ for every $s<0$
(see \cite{sheffield2007gaussian,Berestycki_Powell_2025}), where
$$
H^s(\Sigma;\a)
=
\left\{
u=\sum_{j\ge0}u_j\varphi_j:\;
u_j\in \a,\quad
\|u\|_{H^s(\Sigma;\a)}^2
:=
|u_0|^2+\sum_{j\ge1}\lambda_j^s|u_j|^2<\infty
\right\}.
$$
Equivalently, for every $u,v\in \a$ and distinct points $x,y\in\Sigma$, one has
$$
\mathbb E\big[\langle u,\x_g(x)\rangle\langle v,\x_g(y)\rangle\big]
=
\langle u,v\rangle\,G_g(x,y).
$$

If now $\partial\Sigma\neq\emptyset$, one defines analogously the $\a$-valued Dirichlet Gaussian free
field, denoted by $\x_{g,D}$, by using the orthonormal basis of eigenfunctions of the Dirichlet
Laplacian on $\Sigma$. In that case, for every $u,v\in \a$ and distinct points $x,y\in\Sigma$,
$$
\mathbb E\big[\langle u,\x_{g,D}(x)\rangle\langle v,\x_{g,D}(y)\rangle\big]
=
\langle u,v\rangle\,G_{g,D}(x,y),
$$
where $G_{g,D}$ is the Dirichlet Green function introduced in
Section~\ref{subsec:Laplacian and Green}.

We recall the Markov property of GFF, which is an essential ingredient in the proof of Segal's gluing axioms (see Section \ref{sec:Segal's axioms}).
\begin{proposition}\label{prop:Markov-property}
    Let $(\Sigma,g)$ be a Riemann surface with smooth boundary $\partial\Sigma$. Let $\mathcal{C}$ be a union of smooth non-overlapping closed simple curves separating $\Sigma$ into two connected components $\Sigma_1$ and $\Sigma_2$.
    \begin{enumerate}
        \item If $\partial\Sigma\neq\emptyset,$ then the Dirichlet GFF $\x_{g,D}$ admits a decomposition in law as a sum of independent processes $$\x_{g,D}\overset{\mathrm{law}}{=}\x_1+\x_2+P$$ with $\x_i$ a Dirichlet GFF on $\Sigma_i$ for $i=1,\,2$ and $P$ the harmonic extension on $\Sigma\setminus\mathcal{C}$ of the restriction of $\x_{g,D}$ to $\mathcal{C}$ with zero boundary on $\partial\Sigma.$
        \item If $\partial\Sigma=\emptyset,$ then the GFF $\x$ admits a decomposition in law as a sum of independent processes $$\x_{g}\overset{\mathrm{law}}{=}\x_1+\x_2+P-c_g$$ with $\x_i$ a Dirichlet GFF on $\Sigma_i$ for $i=1,\,2$ and $P$ the harmonic extension on $\Sigma\setminus\mathcal{C}$ of the restriction of $\x_{g}$ to $\mathcal{C}$ and $c_g:=\frac{1}{\v_g(\Sigma)}\int_\Sigma(\x_1+\x_2+P)\d\v_g.$
    \end{enumerate}
\end{proposition}
\begin{proof}
    The claim follows from writing the GFF in the orthonormal basis $(\eps_j)_{j=1,\ldots,r}$ and applying the Markov property of scalar GFF \cite[Proposition C.1]{guillarmou2021segal}.
\end{proof}

\subsection{\texorpdfstring{$g$}{g}-regularization}

Since the Gaussian free field is only a random distribution, we need to regularize it to define the exponential interaction terms. We do this by averaging over small geodesic circles.

Let $h$ be a random distribution on $\Sigma$. For $x\in \Sigma\setminus\partial\Sigma$ and
$\eps>0$ small enough so that the geodesic circle is well defined and contained in the interior of
$\Sigma$, let $\mathcal C_g(x,\eps)$ be the geodesic circle of center $x$ and radius $\eps$, and let
$\mu_{x,\eps}$ denote the uniform probability measure on $\mathcal C_g(x,\eps)$ induced by the metric $g$.
Choose a sequence $(f^n_{x,\eps})_{n\ge1}\subset C^\infty(\Sigma)$ such that
$$
\|f^n_{x,\eps}\|_{L^1(\Sigma,d\v_g)}=1,
\qquad
f^n_{x,\eps}(y)=\theta_n\!\left(\frac{d_g(x,y)}{\eps}\right),
$$
where each $\theta_n\in C_c^\infty((0,2))$ is non-negative, supported near $1$, and such that
the measures $f^n_{x,\eps}\,d\v_g$ converge in $\mathcal D'(\Sigma)$ to $\mu_{x,\eps}$ as $n\to\infty$.

If the pairings $\langle h,f^n_{x,\eps}\rangle$ converge almost surely, as $n\to\infty$, to a random
variable $h_\eps(x)$, and if the resulting field $(x,\eps)\mapsto h_\eps(x)$ admits a continuous
modification, then we say that $h$ admits a $g$-regularization, and we call $h_\eps(x)$ the
geodesic circle average regularization of $h$. In particular, $\x_g$ and $\x_{g,D}$ admit such regularization.

For any compact set $K\subset \Sigma\setminus\partial\Sigma$, one has \cite[Lemma 3.2]{guillarmou2019polyakov}
$$
\mathbb E\big[\langle u,\x_{\eps}(x)\rangle^2\big]=|u|^2(\log \eps^{-1}+W(x))+o(1)
$$
uniformly in $x\in K$, where $\x=\x_g$ or $\x_{g,D}$ and 
$$W(x):=\lim_{y\to x}\left(G(x,y)-\log\frac{1}{d_g(x,y)}\right),
$$with the corresponding choice of the Green's function, is called the Robin mass and is smooth on $\Sigma$. 
\subsection{Imaginary Gaussian multiplicative chaos}

Let $\gamma\in\R$. If $h$ is a random distribution admitting a $g$-regularization
$(h_\eps)_\eps$, we define
\begin{equation}\label{eq:IGMC-eps}
M^{g,\eps}_{\gamma}(h,\d x)
:=
\eps^{-\gamma^2/2}e^{i\gamma h_\eps(x)}\,\d\v_g(x).
\end{equation}

In the Toda theory, the interaction is built from the simple-root projections of the field.
Thus, for each simple root $e_i$, $i=1,\ldots,r$, we define for $
\x=\x_g$ or $\x=\x_{g,D}$
\begin{equation}\label{eq:IGMC-root-eps}
M^{g,\eps}_{\gamma e_i}(\x,\d x)
:=
\eps^{-\frac{\gamma^2}{2}\langle e_i,e_i\rangle}
e^{\bi\gamma\langle e_i,\x_{\eps}(x)\rangle}\,\d\v_g(x).
\end{equation}
If $\gamma^2\langle e_i,e_i\rangle<2,$
 as $\eps\to0$, this measure converges in $L^2(\Omega)$, weakly in the space of
distributions, to a non-trivial random distribution of order $2$ \cite[Theorem 3.1]{lacoin2015complex}, denoted by $M^g_{\gamma e_i}(\x_g,\d x)$.
More precisely, for each $i=1,\ldots,r$, there exists
$D_{\Sigma,i}\in L^2(\Omega)$ such that for every $\varphi\in C^\infty(\Sigma)$
\begin{equation}\label{eq:order-two-bound}
\left|
\int_\Sigma \varphi(x)\,M^g_{\gamma e_i}(\d x)
\right|
\le
D_{\Sigma,i}\bigl(\|\varphi\|_\infty+\|\Delta_g\varphi\|_\infty\bigr)
\qquad\text{a.s.}
\end{equation}
Moreover, if $g'=e^\rho g$, then
\begin{equation}\label{eq:weyl IGMC}
M_{\gamma e_i}^{g'}(\x,\d x)
=
e^{(1-\gamma^2\langle e_i,e_i\rangle/4)\rho(x)}\,M_{\gamma e_i}^g(\x,\d x).
\end{equation}

We note that under our normalization, the longest root has norm $2$, and so the condition $\gamma^2<1$
implies that $\gamma^2\langle e_i,e_i\rangle<2, \,i=1,\ldots,r.$

\subsection{Exponential moments}
The key integrability input for the Toda interaction is that imaginary GMC has exponential moments. 

\begin{proposition}\label{prop:exp-moment-shifted} Assume that $\gamma^2<1$. Let $\Sigma$ be a compact Riemann surface with or without boundary. We set $D'=\Sigma$ if $\partial\Sigma\neq\emptyset$ and set $D'\subset \Sigma$ be open with closure $\overline{D'}\neq \Sigma$ if $\partial\Sigma=\emptyset.$ Let $\x_{g,D}$ be a Dirichlet $\a$-valued GFF on $D'$. Let $D\subset D'$ be an open subset and let $Z:D\to \a$ be an  $\a$-valued random variable. Finally let $f_1,\dots,f_r:D\to\C$ be measurable functions such that $$ V:=\sum_{i=1}^r\int_D |f_i(x)|\d\v_g(x)<\infty $$ and $$ U^2:= \sum_{i,j=1}^r \iint_{D^2} |f_i(x)|\,|f_j(y)|\, e^{\gamma^2\la e_i,e_j\ra G_{g,D}(x,y)} \d\v_g(x)\d\v_g(y) <\infty. $$ Then there exists a constant $C=C(\gamma,r)>0$ such that, for every $\mu\ge0$, \begin{align*} \mathbb E\Bigg[ \exp\Bigg( \mu\Bigg| \sum_{i=1}^r \mu_i \int_D f_i(x)e^{i\gamma\la e_i,Z(x)\ra}\,M^g_{\gamma e_i}(Z+\x_{g,D},\d x) \Bigg| \Bigg) \Bigg] \le e^{C\mu V}\left(1+C\mu U e^{C\mu^2U^2}\right). \end{align*} \end{proposition} \begin{proof} We observe that the field $Z$ gives rise to a factor of modulus 1, and so we can consider the case where $Z\equiv0.$ The claim then readily follows from \cite[Proposition 5.2]{Guillarmou:2023exh} and H\"older's inequailty. \end{proof}

\section{Path integrals and correlation functions}\label{sec:path integral}
We now construct the path integral on a closed Riemann surface. We begin with the
regularized curvature terms appearing in the action \eqref{eq:Toda equation}, then define
the path integral itself. We finally introduce correlation functions with electric and
magnetic insertions and prove their basic covariance properties.

\subsection{Curvature term}\label{subsec:curvature term}
Let $(\Sigma,g)$ be a closed oriented Riemann surface. Since $I_{x_0}(\Omega_{\blambda})$ is a multivalued function on $\Sigma$ and $I_{x_0}(\nu_{\z,\m})$ is a multivalued function on $\Sigma\setminus\{\z\}$, we need to make sense of the curvature terms in the Toda action \eqref{eq:Toda equation}, namely
$$
\int_\Sigma K_g I_{x_0}(\Omega_{\blambda})\d\v_g,
\qquad
\int_{\Sigma\setminus\{\z\}} K_g I_{x_0}(\nu_{\z,\m})\d\v_g.
$$
We regularize these terms by choosing branch cuts so that the corresponding primitives become single-valued on open dense subsets, and by adding correction terms supported on the cuts in order to recover the required invariance properties.
\subsubsection{Curvature term associated to $H^1_\Lambda(\Sigma;\a)$}
For $\Omega\in H^1_\Lambda(\Sigma;\a)$, we define the period morphism $\chi_\Omega: H_1(\Sigma)\rightarrow2\pi\Lambda$
by
$$
\chi_\Omega(\gamma):=\int_\gamma\Omega.
$$
If $\bsigma=(a_j,b_j)_{j=1,\ldots,\bbg}$ is a geometric symplectic basis of $H_1(\Sigma)$,
we set $\Sigma_{\bsigma}:=\Sigma\setminus\bigcup^{\bbg}_{j=1}(a_j\cup b_j).$
For $x_0\in\Sigma_{\bsigma}$ a fixed base point and any closed $\a$-valued $1$-form $\Omega$,
we define
\begin{equation}\label{primitive}
I^{\bsigma}_{x_0}(\Omega)(x):=\int_{\alpha_{x_0,x}}\Omega,\qquad x\in\Sigma_{\bsigma},
\end{equation}
where $\alpha_{x_0,x}\subset\Sigma_{\bsigma}$ is any smooth path from $x_0$ to $x$.
This defines a smooth $\a$-valued function on $\Sigma_{\bsigma}$ with
$\d I^{\bsigma}_{x_0}(\Omega)=\Omega.$

\begin{definition}
Let $\bsigma=(a_j,b_j)_{j=1,\ldots,\bbg}$ be a geometric symplectic basis of $H_1(\Sigma)$,
let $x_0\in\Sigma_{\bsigma}$, and let $\Omega\in H^1_\Lambda(\Sigma;\a)$.
We define the regularized integral
$$
\int^\reg_{\Sigma_{\bsigma}}K_gI^{\bsigma}_{x_0}(\Omega)\d\v_g
:=
\int_{\Sigma_{\bsigma}}K_gI^{\bsigma}_{x_0}(\Omega)\d\v_g
+
2\sum^{\bbg}_{j=1}\left(
\chi_\Omega(a_j)\int_{b_j}k_g\d \ell_g
-
\chi_\Omega(b_j)\int_{a_j}k_g\d \ell_g
\right),
$$
where $k_g$ is the geodesic curvature.
\end{definition}

\begin{lemma}[Invariance under change of geometric symplectic basis]\label{lem:local invariance}
Let $\bsigma'=(a'_j,b'_j)_{j=1,\ldots,\bbg}$ be another geometric symplectic basis
representing the same symplectic homology basis
$[\bsigma]=([a_j],[b_j])_{j=1,\ldots,\bbg}.$
Then for every $\Omega\in H^1_\Lambda(\Sigma;\a),$
$$
\int^\reg_{\Sigma_{\bsigma}}K_gI^{\bsigma}_{x_0}(\Omega)\d\v_g
-
\int^\reg_{\Sigma_{\bsigma'}}K_gI^{\bsigma'}_{x_0}(\Omega)\d\v_g\in8\pi^2\Lambda.$$

\end{lemma}

\begin{lemma}[Conformal change of metrics]\label{lem:conformal change of metrics}
Let $\rho\in C^\infty(\Sigma)$ and $\hat{g}=e^\rho g.$ Then for every
$\Omega\in H^1_\Lambda(\Sigma;\a),$
$$
\int^\reg_{\Sigma_{\bsigma}}K_{\hat{g}}I^{\bsigma}_{x_0}(\Omega)\d\v_{\hat{g}}
=
\int^\reg_{\Sigma_{\bsigma}}K_gI^{\bsigma}_{x_0}(\Omega)\d\v_g
+
\la\d\rho,\Omega\ra_{2}.
$$
\end{lemma}

\begin{lemma}[Invariance under diffeomorphisms]\label{lem:diffeomorphism invariance}
Let $\psi:\Sigma\rightarrow\Sigma$ be an orientation-preserving diffeomorphism and let
$\psi(\bsigma):=(\psi(a_j),\psi(b_j))_{j=1,\ldots,\bbg}.$
Then for every $\Omega\in H^1_\Lambda(\Sigma;\a),$
$$
\int^\reg_{\Sigma_{\bsigma}}K_gI^{\bsigma}_{x_0}(\Omega)\d\v_g
=
\int^\reg_{\Sigma_{\psi(\bsigma)}}K_{\psi_*g}
I^{\psi(\bsigma)}_{\psi(x_0)}(\psi_*\Omega)\d\v_{\psi_*g}.
$$
\end{lemma}

\begin{lemma}[Invariance under change of symplectic basis of $H_1(\Sigma)$]\label{lem:Invariance under change of symplectic basis}
Let $\bsigma$ and $\bsigma'$ be two geometric symplectic bases of $ H_1(\Sigma).$
Then for every $\Omega\in H^1_\Lambda(\Sigma;\a),$
$$
\int^\reg_{\Sigma_{\bsigma}}K_gI^{\bsigma}_{x_0}(\Omega)\d\v_g
-
\int^\reg_{\Sigma_{\bsigma'}}K_gI^{\bsigma'}_{x_0}(\Omega)\d\v_g
\in8\pi^2\Lambda.
$$
\end{lemma}
\begin{proof}[Proof of Lemmas \ref{lem:local invariance}--\ref{lem:Invariance under change of symplectic basis}]
For the conformal-change and diffeomorphism formulas, the assertions are linear in $\Omega$. Hence, after writing $\Omega=\sum_{\ell=1}^r\Omega_\ell\varepsilon_\ell$
in an orthonormal basis of $\a$, they follow componentwise from the corresponding scalar statements (\cite[Lemmas 4.3 \& 4.4]{Guillarmou:2023exh}).

For the two basis-independence statements, writing  
$\Omega=\sum_{\l=1}^r \Omega^{(\l)}\,\omega_\l^\vee$ and applying \cite[Lemmas 4.2 \& 4.5]{Guillarmou:2023exh} proves the assertion.
\end{proof}
\subsubsection{Curvature terms associated to magnetic points}
Let $\z=(z_1,\ldots,z_{n_{\mathfrak{m}}})\in\Sigma^{n_{\mathfrak{m}}}$ be pairwise distinct marked points and let $\m=(m_1,\ldots,m_{n_{\mathfrak{m}}})\in\Lambda^{n_{\mathfrak{m}}}$ be magnetic charges satisfying $\sum_{j=1}^{n_{\mathfrak{m}}} m_j=0$. Moreover, we let
$v_j\in T_{z_j}\Sigma,\, j=1,\ldots,n_{\mathfrak{m}},$
be unit tangent vectors and denote
$$\bv=((z_1,v_1),\ldots,(z_{n_{\mathfrak{m}}},v_{n_{\mathfrak{m}}}))\in (T\Sigma)^{n_{\mathfrak{m}}}.$$
We shall use a system of branch cuts joining the marked points to make the magnetic primitive single-valued. We endow the lattice $\Lambda$ with the lexicographic order induced by the ordered basis $(\omega^\vee_1,\ldots,\omega_r^\vee)$ and adapt \cite[Definition 4.6]{Guillarmou:2023exh} to our setting.

\begin{definition}[Defect graph]
We consider a collection of $n_{\mathfrak{m}}-1$ smooth arcs $\xi_p:[0,1]\rightarrow\Sigma$ with endpoints in $\z$ such that:
\begin{itemize}
    \item every arc is simple, and any two arcs do not intersect except possibly at their endpoints;
    \item if $\xi_p(0)=z_j$ and $\xi_p(1)=z_{j'}$, then
    $$
    \xi'_p(0)=\lambda_{p,j}v_j,\qquad \xi'_p(1)=\lambda_{p,j'}v_{j'}
    $$
    for some $\lambda_{p,j},\lambda_{p,j'}>0$;
    \item if $\xi_p(0)=z_j$ and $\xi_p(1)=z_{j'}$, then $m_j\leq m_{j'}$;
    \item the induced graph with vertices $\z$ is connected, acyclic, and when viewed as a subset of $\Sigma$, the set $$
\mathcal{D}_{\bv,\bxi}:=\bigcup_{p=1}^{n_{\mathfrak{m}}-1}\xi_p([0,1])
$$ is homotopic to a point.
\end{itemize}
We call $\mathcal{D}_{\bv,\bxi}$ the \emph{defect graph} associated to $\bv$ and $\bxi$.
\end{definition}

\begin{lemma}\label{exactness-defect-graph}
On $\Sigma\setminus\mathcal{D}_{\bv,\bxi}$, the $\a$-valued $1$-form $\nu_{\z,\m}$ is exact.
\end{lemma}
\begin{proof}
Write each $m_j$ in the lattice basis
$m_j=\gamma^{-1}\sum_{\l=1}^r m_{j,\l}\omega_\l^\vee,
\, m_{j,\l}\in\mathbb Z.$
Applying \cite[Lemma 4.7]{Guillarmou:2023exh} to each scalar charge family
$(m_{1,\l},\ldots,m_{n_{\mathfrak m},\l})$ gives exactness of each component of
$\nu_{\z,\m}$ on $\Sigma\setminus\mathcal D_{\bv,\bxi}$. Summing over $\l$ gives the
claim.
\end{proof}
Let $x_0\in \Sigma\setminus \mathcal D_{\bv,\bxi}$. By Lemma~\ref{exactness-defect-graph}, we may define on $\Sigma\setminus\mathcal D_{\bv,\bxi}$
$$
I^{\bxi}_{x_0}(\nu_{\z,\m})(x):=\int_{\alpha_{x_0,x}}\nu_{\z,\m},
$$
where $\alpha_{x_0,x}$ is any smooth curve in $\Sigma\setminus\mathcal D_{\bv,\bxi}$ from $x_0$ to $x$.

\begin{definition}[Regularized magnetic curvature term]\label{def:regmag}
Fix an edge $\xi_p$ of the defect graph and a point $x=\xi_p(t)$ with $t\in(0,1)$. Let $\tau_p$ be the positive unit tangent vector to $\xi_p$ at $x$, and let $\nu_p:=J\tau_p$
be the left unit normal, where $J$ denotes rotation by $+\pi/2$ in the oriented tangent bundle.
Choose a positively oriented smooth simple closed curve $\alpha_x$ such that:
\begin{itemize}
    \item $\alpha_x(0)=x$ and $\dot\alpha_x(0)=\nu_p$;
    \item $\alpha_x$ meets $\mathcal D_{\bv,\bxi}$ only at $x$;
    \item $\alpha_x$ bounds a topological disk.
\end{itemize}
We define
$$
\kappa(\xi_p):=\int_{\alpha_x}\nu_{\z,\m}.
$$
This $\a$-valued quantity is well defined, i.e. it does not depend on the choices of $x$ and $\alpha_x$. Moreover, it can be equivalently defined as
$$
\kappa(\xi_p)=2\pi\sum_{z_j\in D_{\alpha_x}}m_j,
$$
where $D_{\alpha_x}$ is the disk bounded by $\alpha_x$.

We then define the regularized magnetic curvature term by
$$
\int_\Sigma^{\reg} I^{\bxi}_{x_0}(\nu_{\z,\m})\,K_g\,\d\v_g
:=
\int_{\Sigma\setminus\mathcal D_{\bv,\bxi}} I^{\bxi}_{x_0}(\nu_{\z,\m})\,K_g\,\d\v_g
-
2\sum_{p=1}^{n_{\mathfrak m}-1}\kappa(\xi_p)\int_{\xi_p}k_g\,\d\ell_g.
$$
\end{definition}

\begin{lemma}[Invariance under change of defect graph]\label{lem:defect-graph-independence}
For any two defect graphs $\bxi,\bxi'$, we have 
$$
\int_\Sigma^{\reg} I^{\bxi}_{x_0}(\nu_{\z,\m})\,K_g\,\d\v_g-\int_\Sigma^{\reg} I^{\bxi'}_{x_0}(\nu_{\z,\m})\,K_g\,\d\v_g\in8\pi^2\Lambda.
$$
\end{lemma}
\begin{lemma}[Conformal change of metrics]\label{lem:conformal-change-magnetic-curvature}
Let $g'=e^\rho g$ be a conformal change of metric.
Then
$$
\int_{\Sigma}^{\reg} I^{\bxi}_{x_0}(\nu_{\z,\m})K_{g'}\,\d\v_{g'}
=
\int_{\Sigma}^{\reg} I^{\bxi}_{x_0}(\nu_{\z,\m})K_g\,\d\v_g
+
\langle \d\rho,\nu_{\z,\m}\rangle_2.
$$
In particular, if $\nu_{\z,\m}=\nu_{\z,\m}^{\mathrm h}$ is harmonic, then
$$
\int_{\Sigma}^{\reg} I^{\bxi}_{x_0}(\nu_{\z,\m}^{\mathrm h})K_{g'}\,\d\v_{g'}
=
\int_{\Sigma}^{\reg} I^{\bxi}_{x_0}(\nu_{\z,\m}^{\mathrm h})K_g\,\d\v_g.
$$
\end{lemma}

\begin{proof}[Proof of Lemma \ref{lem:defect-graph-independence}]
Writing $$m_j=\gamma^{-1}\sum_{\l=1}^r m_{j,\l}\omega_\l^\vee,
\, m_{j,\l}\in\mathbb Z,\qquad\nu_{\z,\m}=
\sum_{\l=1}^r \nu^{(\l)}_{\z,\m}\,\omega_\l^\vee,$$
the scalar proof of \cite[Lemma 4.9]{Guillarmou:2023exh} applies to each component
$\nu^{(\l)}_{\z,\m}$. For each elementary $S$-, $R$-, $A$-, or $D$-move, the
corresponding change of the scalar regularized curvature term belongs to
$8\pi^2\gamma^{-1}\mathbb Z$. Therefore, the $\a$-valued change belongs to $8\pi^2\Lambda.$ Since any two defect graphs are connected by a finite sequence of such moves, the result
follows.
\end{proof}

\begin{proof}[Proof of Lemma \ref{lem:conformal-change-magnetic-curvature}]
Fix an orthonormal basis $(\varepsilon_1,\ldots,\varepsilon_r)$ of $\a$, and write
$\nu_{\z,\m}=\sum_{\ell=1}^r \nu^{(\ell)}_{\z,\m}\varepsilon_\ell$ and
$m_j=\sum_{\ell=1}^r m_{j,\ell}\varepsilon_\ell .$
Although the coefficients $m_{j,\ell}$ are real numbers, not necessarily integers, this causes no
difficulty here: the proof of \cite[Lemma 4.10]{Guillarmou:2023exh} applies verbatim to
real magnetic coefficients, since it only uses the local form
$\nu^{(\ell)}_{\z,\m}=m_{j,\ell}\d\theta+O(1)$
near each marked point and the identity \eqref{eq:curvature-conformal-change} with the corresponding transformation of geodesic curvature along the cuts. Applying the
scalar argument to each component and summing over $\ell$ proves the first assertion. The second assertion then follows from the coclosedness of $\nu_{\z,\m}^{\mathrm{h}}$.
\end{proof}

\subsection{Path integrals}

Throughout the rest of this section, we work under the assumptions that $\gamma^2<1$ and that $Q\in\Lambda^*$ with $\gamma>0$ so that both the Toda interaction characters and the curvature phase are compatible with the compactification lattice $\Lambda$, where we recall that the background charge is given by $$Q=\gamma\rho-\frac{2}{\gamma}\rho^\vee.$$

We fix the following data:
\begin{itemize}
\item a geometric symplectic basis
$\bsigma=(\sigma_1,\ldots,\sigma_{2\bbg})$ of $H_1(\Sigma)$ and closed smooth real-valued $1$-forms $\eta_1,\ldots,\eta_{2\bbg}$ dual to $\bsigma$ as in Lemma \ref{geometric symplectic basis};
\item a base point $x_0\in \Sigma_{\bsigma}=\Sigma\setminus\cup^{2\bbg}_{j=1}\sigma_j$ and $I_{x_0}^{\bsigma}(\Omega)$ the smooth function on $\Sigma_{\bsigma}$ defined in \eqref{primitive}

\end{itemize}

We first introduce the space of test functions on which the path integral is defined. Let $(e_j)_{j\geq 0}$ be an orthonormal basis of $L^2(\Sigma,\d\v_g)$ consisting of eigenfunctions of
$\Delta_g$. We shall write elements of
$H^{s}(\Sigma;\a)$, $s<0$ in the form
$$
f=f_0+\sqrt{2\pi}\sum_{j\geq 1} f_j e_j,
$$ and equip the Sobolev space $H^{s}(\Sigma;\a)$ with the pushforward of the measure $\d c\otimes \P$ on $\T(\gamma)\times\Omega$ via the map $(c,\omega)\mapsto c+\x_g(\omega).$

We recall that from Section \ref{sec:equivariant},  any equivariant distribution $u\in H_\Gamma^{s}(\widetilde\Sigma;\a)$ can be uniquely written as $$u=\pi^* (f_0+\sqrt{2\pi}\sum_{j\geq 1} f_j e_j)+I_{x_0}(\Omega_{\blambda})$$
for some $f\in H^s(\Sigma;\a)$ and $\blambda\in \Lambda^{2\bbg}\simeq H_1(\Sigma)$. The space $\mathcal{E}_\Lambda(\Sigma;\a)$ of test functions consists of all functionals $F:H^s_\Gamma(\widetilde\Sigma;\a)\rightarrow\C$ of the form
\begin{equation}\label{eq:test-functions-elementary}
F(u)
=
\sum_{\l\in A}
e^{\bi\langle \ell,f_0\rangle}
P_{\ell}(f-f_0)\,
G_{\ell}
\left(
e^{\bi\langle q,I^{\bsigma}_{x_0}(\Omega_{\blambda})\rangle}
\right)
\end{equation} for some finite subset $A\subset\Lambda^*$ and $q\in\Lambda^*$ 
where each $P_{\ell}$, depending on $\blambda,$ is a polynomial in the non-zero mode
$f-f_0$, namely
$P_{\ell}(f-f_0)
=
\mathcal P_{\ell}
\left(
\langle f-f_0,h_1\rangle,\ldots,
\langle  f-f_0,h_{m_\l}\rangle
\right),$
for some polynomial $\mathcal P_{\ell,\blambda}$ and
$h_1,\ldots,h_m\in H^{s}_0(\Sigma;\a)$,
and each $G_{\ell}$, also depending on $\blambda$, is a bounded continuous function on
$C^0(\Sigma_{\bsigma};\S^1).$ One can observe that the space is independent of the choice of the representative of $[\Omega_{\blambda}]\in H^1(\Sigma;\a)$.

The condition $\ell,q \in\Lambda^*$ ensures that $F$ is
$2\pi\Lambda$-periodic in the compactified field. In particular,
$F(u+2\pi\lambda)=F(u),
\, \lambda\in\Lambda.$ The polynomial dependence on the non-zero mode is important: it is the class on which
one can later perform analytic continuation arguments for correlation functions. 

We now introduce the normed space obtained by completing these elementary test
functions. For $\blambda\in\Lambda^{2\bbg}$, let $\Omega_{\blambda}^{\mathrm h}
:=
\Pi_1\Omega_{\blambda}$
be the harmonic representative of the cohomology class $[\Omega_{\blambda}]$, and
write
$\Omega_{\blambda}
=
\Omega_{\blambda}^{\mathrm h}
+
\d f_{\blambda},$
where $f_{\blambda}$ is a smooth $\a$-valued function.
For $p\geq 1$, define
\begin{equation}\label{eq:L-infty-p-norm}
\|F\|_{\mathcal L^{\infty,p}_{\Lambda}}
:=
\sup_{\blambda\in\Lambda^{2\bbg}}
\left(
\int_{\a/(2\pi\Lambda)}
\E\left[
e^{-\frac{1}{2\pi}\langle \d\x_g,\Omega_{\blambda}\rangle_2
-\frac{1}{4\pi}\|\d f_{\blambda}\|_2^2}
\left|
F\left(c+\x_g+I^{\bsigma}_{x_0}(\Omega_{\blambda})\right)
\right|^p
\right]\d c
\right)^{1/p}.
\end{equation}
We denote by $\mathcal L^{\infty,p}_{\Lambda}(\Sigma;\a)$
the completion of $\mathcal E_{\Lambda}(\Sigma;\a)$ with respect to this norm, after
identifying two test functions whose distance is zero. 

\begin{lemma}\label{lem:L-infty-p-independent-representatives}
The norm $\|\cdot\|_{\mathcal L^{\infty,p}_{\Lambda}}$
does not depend on the choice of representatives $\Omega_{\blambda}\in [\Omega_{\blambda}]\in H^1_{\Lambda}(\Sigma;\a)$.
\end{lemma}

\begin{proof}
Let $\blambda\in\Lambda^{2\bbg}$. We write $\Omega_{\blambda}^{\mathrm h}:=\Pi_1\Omega_{\blambda}$
for the harmonic representative of the cohomology class of $\Omega_{\blambda}$ so that
$\Omega_{\blambda}
=
\Omega_{\blambda}^{\mathrm h}+\d f_{\blambda},$
where $f_{\blambda}$ is a smooth $\a$-valued function. By the Girsanov transform, we have
$$
\begin{aligned}
&\int_{\a/(2\pi\Lambda)}
\E\left[
e^{-\frac{1}{2\pi}\langle \d\x_g,\Omega_{\blambda}\rangle_2
-\frac{1}{4\pi}\|\d f_{\blambda}\|_2^2}
\left|
F\left(
c+\x_g+I^{\bsigma}_{x_0}(\Omega_{\blambda})
\right)
\right|^p
\right]\d c                                                \\
&\quad =
\int_{\a/(2\pi\Lambda)}
\E\left[
e^{-\frac{1}{2\pi}
\langle \d\x_g,\Omega_{\blambda}^{\mathrm h}\rangle_2}
\left|
F\left(
c+\x_g+I^{\bsigma}_{x_0}(\Omega_{\blambda}^{\mathrm h})
+m_g(f_{\blambda})-f_{\blambda}(x_0)
\right)
\right|^p
\right]\d c .
                                           \\
&\quad=
\int_{\a/(2\pi\Lambda)}
\E\left[
\left|
F\left(
c+\x_g+I^{\bsigma}_{x_0}(\Omega_{\blambda}^{\mathrm h})
\right)
\right|^p
\right]\d c 
\end{aligned}
$$
where $m_g(f_{\blambda})$ denotes the average with respect to $\d \v_g$, the associated shift is absorbed by translating the zero mode $c\in \a/(2\pi\Lambda)$, and we have used 
$\langle \d\x_g,\Omega_{\blambda}^{\mathrm h}\rangle_2
=\langle \x_g,\d^*\Omega_{\blambda}^{\mathrm h}\rangle_2
=0$ since $\Omega^{\mathrm{h}}_{\blambda}$ is harmonic. This proves the claim.
\end{proof}

On a closed surface $\Sigma$, and for a topological sector
$\blambda\in\Lambda^{2\bbg}$, we write the Toda field as
$$
\Phi^{\blambda}_g
:=
c+\x_g+I^{\bsigma}_{x_0}(\Omega_{\blambda}).
$$
Here $c\in\a/(2\pi\Lambda)$, $\x_g$ is the $\a$-valued GFF with zero average, and
$I^{\bsigma}_{x_0}(\Omega_{\blambda})$ is the primitive of the closed $1$-form
$\Omega_{\blambda}$ on the cut surface
$\Sigma_{\bsigma}$. Thus $\Phi^{\blambda}_g$ is first defined as an element of
$H^s(\Sigma_{\bsigma};\a)$, for $s<0$. However, because the jumps of
$I^{\bsigma}_{x_0}(\Omega_{\blambda})$ across the cuts belong to $2\pi\Lambda$, the
lift of $\Phi^{\blambda}_g$ to the universal cover extends uniquely to an equivariant
distribution on $\widetilde\Sigma$. In other words, $\Phi^{\blambda}_g\in H^s_\Gamma(\widetilde\Sigma;\a),$
with monodromy determined by $\blambda$.

In what follows we shall not distinguish notationally between the representative of
$\Phi^{\blambda}_g$ on the cut surface and its equivariant lift to
$\widetilde\Sigma$. Thus, whenever $F$ is a test function on
$H^s_\Gamma(\widetilde\Sigma;\a)$, the expression
$F(\Phi^{\blambda}_g)$
means that $F$ is evaluated on this equivariant extension. The field
$\Phi^{\blambda}_g$ depends on the zero mode $c$, on the GFF $\x_g$, on the
topological sector $\blambda$, and on the auxiliary choices
$(\bsigma,x_0)$. The dependence on these auxiliary choices will disappear from the
path integral after summing over $\blambda$ and using the regularized curvature term.

\begin{definition}\label{def:path integrals}
For all $F\in\mathcal{E}_\Lambda(\Sigma;\a)$, we define
\begin{equation}\label{path integral}
\begin{aligned}
\la F\ra_{\Sigma,g}
:=&
\left(\frac{\v_g(\Sigma)}{\det'(\Delta_g)}\right)^{r/2}
\sum_{\blambda\in\Lambda^{2\bbg}}
e^{-\frac{1}{4\pi}\norm{\Omega_{\blambda}}^2_2} \\
&\times
\int_{\a/(2\pi\Lambda)}
\E\left[
e^{-\frac{1}{2\pi}\la\d\x_g,\Omega_{\blambda}\ra_2}
F(\Phi_g^{\blambda})
e^{-\frac{i}{4\pi}\la QK_g,\Phi_g^{\blambda}\ra^\reg_g
-\sum^r_{i=1}\mu_iM^g_{\gamma e_i}(\Phi^{\blambda}_g,\Sigma)}
\right]\d c,
\end{aligned}
\end{equation}
where the regularized curvature term is
$$
\la QK_g,\Phi_g^{\blambda}\ra^\reg_g
:=
\int_\Sigma\la Q,c+\x_g\ra K_g\d\v_g
+
\int^\reg_{\Sigma_{\bsigma}}
\la Q, I^{\bsigma}_{x_0}(\Omega_{\blambda})\ra K_g\d\v_g.
$$
\end{definition}

\begin{proposition}\label{prop:path-integral-well-defined}The map
$F\longmapsto \la F\ra_{\Sigma,g},\, F\in\mathcal E_\Lambda(\Sigma;\a),$ is well-defined and finite and
extends uniquely to a continuous linear functional on
$\mathcal L^{\infty,p}_{\Lambda}(\Sigma;\a)$ for $p>1$. Moreover,
$\la F\ra_{\Sigma,g}$ does not depend on the base point $x_0\in\Sigma$, on the choice
of the homology basis $\bsigma$, or on the choice of the closed forms representing the
cohomology basis dual to $\bsigma$.
\end{proposition}

\begin{proof}
We first prove the estimate for
$F\in\mathcal E_\Lambda(\Sigma;\a)$. Fix
$\blambda\in\Lambda^{2\bbg}$, and write
$\Omega_{\blambda}
=
\Omega_{\blambda}^{\mathrm h}
+
\d f_{\blambda}$ as before. By the Girsanov transform, together with the translation invariance of the zero mode $c\in\a/(2\pi\Lambda)$, the $\blambda$-summand in
\eqref{path integral} can be rewritten as
$$
e^{-\frac{1}{4\pi}\norm{\Omega_{\blambda}^{\mathrm h}}_2^2}
\int_{\a/(2\pi\Lambda)}
\E\left[
F(\Phi_g^{\blambda,\mathrm h})
e^{-\frac{\bi}{4\pi}
\la QK_g,\Phi_g^{\blambda,\mathrm h}\ra_g^{\reg}
-\sum_{i=1}^r
\mu_i M^g_{\gamma e_i}
(\Phi_g^{\blambda,\mathrm h},\Sigma)}
\right]\d c,
$$
where $\Phi_g^{\blambda,\mathrm h}
:=
c+\x_g+I_{x_0}^{\bsigma}(\Omega_{\blambda}^{\mathrm h}).$

Since the curvature term is purely imaginary, its exponential has modulus one, and hence
H\"older's inequality with $\frac1p+\frac1q=1.$ gives
$$
\begin{aligned}
&\left|\int_{\a/(2\pi\Lambda)}\E\left[F(\Phi_g^{\blambda,\mathrm h})e^{-\frac{\bi}{4\pi}\la QK_g,\Phi_g^{\blambda,\mathrm h}\ra_g^{\reg}-\sum_{i=1}^r\mu_iM^g_{\gamma e_i}(\Phi_g^{\blambda,\mathrm h},\Sigma)}
\right]\d c\right|\\
&\qquad\leq\left(\int_{\a/(2\pi\Lambda)}\E\left[\left|F(\Phi_g^{\blambda,\mathrm h})\right|^p\right]\d c\right)^{1/p}             
\left(\int_{\a/(2\pi\Lambda)}\E\left[
e^{q\left|\sum_{i=1}^r\mu_iM^g_{\gamma e_i}(\Phi_g^{\blambda,\mathrm h},\Sigma)\right|}\right]\d c\right)^{1/q}.
\end{aligned}
$$
The second factor is bounded uniformly in $\blambda$ by the exponential moment
estimate for shifted imaginary GMC, namely Proposition
\ref{prop:exp-moment-shifted} (see the proof of \cite[Proposition 6.4]{Guillarmou:2023exh}). For the first factor, Lemma
\ref{lem:L-infty-p-independent-representatives} gives
$$
\left(
\int_{\a/(2\pi\Lambda)}
\E\left[
\left|F(\Phi_g^{\blambda,\mathrm h})\right|^p
\right]\d c
\right)^{1/p}
\leq
\|F\|_{\mathcal L^{\infty,p}_{\Lambda}}.
$$
Therefore there exists a constant $C>0$, independent of $\blambda$, such that
$$
|\la F\ra_{\Sigma,g}|
\leq
C
\|F\|_{\mathcal L^{\infty,p}_{\Lambda}}
\left(\frac{\v_g(\Sigma)}{\det'(\Delta_g)}\right)^{r/2}
\sum_{\blambda\in\Lambda^{2\bbg}}
e^{-\frac{1}{4\pi}
\norm{\Omega_{\blambda}^{\mathrm h}}_2^2}.
$$
The map $\blambda\longmapsto
\norm{\Omega_{\blambda}^{\mathrm h}}_2^2$
is a positive-definite quadratic form on the lattice $\Lambda^{2\bbg}$, and thus the above series converges. Thus $\la F\ra_{\Sigma,g}$ is well-defined for $F\in\mathcal E_\Lambda(\Sigma;\a)$ and the path integral extends uniquely to every
$F\in\mathcal L^{\infty,p}_{\Lambda}(\Sigma;\a)$ by the above estimate.

It remains to prove the independence statements. We first prove them for
$F\in\mathcal E_\Lambda(\Sigma;\a)$, and the general case follows immediately by the continuity. If the base point $x_0$ is changed, then the primitive
$I^{\bsigma}_{x_0}(\Omega_{\blambda})$ changes by an additive constant in $\a$.
This constant is absorbed by translating the zero mode $c\in\a/(2\pi\Lambda).$ Since the Haar measure $\d c$ is translation invariant and the test functions are
$2\pi\Lambda$-periodic in the compactified field, the value of the path integral is unchanged.

Next, suppose that the representative of the cohomology class
$[\Omega_{\blambda}]$ is changed, for example, replacing $\Omega_{\blambda}$ by $\Omega_{\blambda}+\d h_{\blambda}$ for some smooth $\a$-valued function $h_{\blambda}$. The corresponding primitive
changes by
$$
I^{\bsigma}_{x_0}(\Omega_{\blambda}+\d h_{\blambda})
=
I^{\bsigma}_{x_0}(\Omega_{\blambda})
+
h_{\blambda}-h_{\blambda}(x_0).
$$
The Girsanov transform for the shift $\x_g\longmapsto \x_g+h_{\blambda}$
shows that the change in the linear Gaussian factor $e^{-\frac{1}{2\pi}\la\d\x_g,\Omega_{\blambda}\ra_2}$
is exactly compensated by that in the quadratic factor $e^{-\frac{1}{4\pi}\norm{\Omega_{\blambda}}_2^2}.$
The remaining additive constant $h_{\blambda}(x_0)$ is again absorbed by translating
the zero mode $c$. Equivalently, both choices of representative reduce, by the same
Girsanov argument, to the expression involving the harmonic representative
$\Omega_{\blambda}^{\mathrm h}$. Therefore the path integral is independent of the
choice of closed representative.

Finally, changing the geometric symplectic basis $\bsigma$ changes the
parametrization of the lattice of topological sectors. Hence the sum over
$\blambda\in\Lambda^{2\bbg}$ is merely reindexed. By Lemma
\ref{lem:Invariance under change of symplectic basis}, the regularized curvature term changes by an element whose pairing with $Q$ gives a trivial phase under the standing assumption $Q\in\Lambda^*.$ Thus the value of $\la F\ra_{\Sigma,g}$ is independent of $\bsigma$.

This completes the proof.
\end{proof}
\subsection{Correlation functions}
\paragraph{\emph{Magnetic operators}}
We now extend the preceding construction to correlation functions with magnetic
operators. Let $\z=(z_1,\ldots,z_n)\in \Sigma^n$
be a family of distinct marked points, to which we associate magnetic charges $\m=(m_1,\ldots,m_n)\in\Lambda^n$ satisfying the neutrality condition $\sum_{j=1}^n m_j=0.$ We also fix tangent vectors $\bv=((z_1,v_1),\ldots,(z_n,v_n))\in (T\Sigma)^n.$

Let $\bm\xi$ be a defect graph associated with the data $(\bv,\m)$ and we recall from Proposition \ref{prop:magnetic-background} the harmonic representative of the magnetic $1$-form $\nu^{\mathrm{h}}_{\z,\m}$ with monodromy
$2\pi m_j$ around $z_j$. We denote by $I^{\bm\xi}_{x_0}(\nu_{\z,\m}^{\mathrm h})$
the corresponding primitive on the cut surface determined by the defect graph.

For $\blambda\in\Lambda^{2\bbg}$, we consider the Toda field
$$
\Phi^{\blambda,\m}_g
:=
c+\x_g
+I^{\bsigma}_{x_0}(\Omega_{\blambda})
+I^{\bm\xi}_{x_0}(\nu_{\z,\m}^{\mathrm h}).
$$
This field is defined on the surface cut along both the homology cuts $\bsigma$ and the
defect graph $\bxi$, but can be equivalently viewed as an equivariant distribution on the universal
cover of $\Sigma_{\z}=\Sigma\setminus\{z_1,\ldots,z_n\}$, with monodromy determined both by
$\blambda$ and by $\m$.

We introduce the corresponding elementary test functions. The space $\mathcal E_{\Lambda}^{\m}(\Sigma;\a)$ consists of all functionals $F:H^s_{\Gamma,\m}(\widetilde{\Sigma}_{\z};\a)\rightarrow \C$
which, in the sector $\blambda$, can be written as
\begin{equation}\label{eq:magnetic-test-functions-elementary}
F(u)
=
\sum_{\ell\in A}
e^{\bi\langle \ell,f_0\rangle}
P_{\ell}(f-f_0)\,
G_{\ell}
\left(
e^{\bi\langle q,
I^{\bsigma}_{x_0}(\Omega_{\blambda})
+
I^{\bm\xi}_{x_0}(\nu_{\z,\m}^{\mathrm h})
\rangle}
\right),
\end{equation}
where, as before,  $A\subset\Lambda^*$ is finite, $q\in\Lambda^*$, each
$P_{\ell}$, depending on $(\blambda,\m)$, is a polynomial in the non-zero mode $f-f_0$, and
each $G_{\ell}$, depending on $(\blambda,\m)$ is a bounded continuous function on $C^0(\Sigma_{\bsigma,\bm\xi};\S^1).$

The conditions $\ell,q\in\Lambda^*$ ensure that the functional is well-defined on the compactified field and is invariant under addition of $2\pi\Lambda$. Moreover, the space $\mathcal{E}_\Lambda^{\m}(\Sigma;\a)$ does not depend on the choice of the representative $\Omega_{\blambda}+\nu^{\mathrm{h}}_{\z,\m}$ of its cohomology class. 

As before, we define the corresponding $\mathcal L^{\infty,p}$-norm for $p\geq 1$ by 
\begin{equation}\label{eq:L-infty-p-magnetic-norm}
\begin{aligned}
\|F\|_{\mathcal L^{\infty,p}_{\Lambda,\m}}:=\sup_{\blambda\in\Lambda^{2\bbg}}
\Bigg(\int_{\a/(2\pi\Lambda)}\E\Big[
&e^{-\frac{1}{2\pi}\langle \d\x_g,\Omega_{\blambda}+\nu_{\z,\m}^{\mathrm h}\rangle_2
-\frac{1}{4\pi}\|\d f_{\blambda}\|_2^2}
\\&\times
\left|F\left(c+\x_g
+I^{\bsigma}_{x_0}(\Omega_{\blambda})
+I^{\bm\xi}_{x_0}(\nu_{\z,\m}^{\mathrm h})
\right)\right|^p\Big]\d c\Bigg)^{1/p}.
\end{aligned}
\end{equation}
Here, as in the previous subsection,
$\Omega_{\blambda}
=
\Omega_{\blambda}^{\mathrm h}
+
\d f_{\blambda}$ and we denote by $\mathcal L^{\infty,p}_{\Lambda,\m}(\Sigma;\a)$ the completion of $\mathcal E_{\Lambda,\m}(\Sigma;\a)$ with respect to this norm.

Following the same lines as in the proof of Lemma \ref{lem:L-infty-p-independent-representatives}, we have:
\begin{lemma}\label{L-infty-p-independent-representatives-magnetic}
    The norm $\norm{\cdot}_{\mathcal{L}^{\infty,p}_{\Lambda,\m}}$ does not depend on the choice of representatives $\Omega_{\blambda}\in [\Omega_{\blambda}]\in H^1_{\Lambda}(\Sigma;\a)$.
 
\end{lemma}

\begin{definition}\label{def:magnetic-correlation}
For $F\in\mathcal E_{\Lambda}^{\m}(\Sigma;\a)$, we define the path integral with
magnetic insertions by
\begin{equation}\label{eq:magnetic-path-integral}
\begin{aligned}
&\la F V^g_{(0,\m)}(\bv)\ra_{\Sigma,g}
\\&:=
\left(\frac{\v_g(\Sigma)}{\det'(\Delta_g)}\right)^{r/2}
\sum_{\blambda\in\Lambda^{2\bbg}}
e^{
-\frac{1}{4\pi}\norm{\Omega_{\blambda}}^2_2
-\frac{1}{4\pi}\norm{\nu_{\z,\m}^{\mathrm h}}^2_{g,0}
-\frac{1}{2\pi}
\la\Omega_{\blambda},\nu_{\z,\m}^{\mathrm h}\ra_2
}
\\&\quad
\times
\int_{\a/(2\pi\Lambda)}
\E\Bigg[
e^{
-\frac{1}{2\pi}
\la \d\x_g,\Omega_{\blambda}+\nu_{\z,\m}^{\mathrm h}\ra_2
}
F(\Phi_g^{\blambda,\m})
e^{
-\frac{\bi}{4\pi}
\la QK_g,\Phi_g^{\blambda,\m}\ra^\reg_g
-\sum_{i=1}^r
\mu_i M^g_{\gamma e_i}(\Phi^{\blambda,\m}_g,\Sigma)
}
\Bigg]\d c,
\end{aligned}
\end{equation}
where the regularized curvature term is
$$
\begin{aligned}
\la QK_g,\Phi_g^{\blambda,\m}\ra^\reg_g
:=
&
\int_\Sigma
\la Q,c+\x_g\ra K_g\d\v_g
+
\int^\reg_{\Sigma_{\bsigma}}
\la Q,I^{\bsigma}_{x_0}(\Omega_{\blambda})\ra
K_g\d\v_g
\\
&+
\int^\reg_{\Sigma}
\la Q,I^{\bm\xi}_{x_0}(\nu_{\z,\m}^{\mathrm h})\ra
K_g\d\v_g.
\end{aligned}
$$
\end{definition}

\begin{proposition}\label{prop:magnetic-path-integral-well-defined}
The map $F\mapsto
\la F V^g_{(0,\m)}(\bv)\ra_{\Sigma,g},
\,
F\in\mathcal E_{\Lambda}^{\m}(\Sigma;\a),$
is well-defined and finite and
extends uniquely to a continuous linear functional on
$ \mathcal L^{\infty,p}_{\Lambda,\m}(\Sigma;\a)$ for $p>1.$
Moreover, the extended functional does not depend on the base point
$x_0\in\Sigma$, on the choice of the closed representatives
$\Omega_{\blambda}\in[\Omega_{\blambda}]$, on the choice of the geometric
symplectic basis $\bsigma$, or on the auxiliary defect graph $\bm\xi$, once the
marked tangent vectors $\bv=((z_1,v_1),\ldots,(z_n,v_n))$
are fixed.
\end{proposition}
\begin{proof}
The proof follows the same lines as Proposition
\ref{prop:path-integral-well-defined}, and so we only indicate the changes caused by the magnetic background.

Write $\Omega_{\blambda}=\Omega_{\blambda}^{\mathrm h}+\d f_{\blambda}.$
By the Girsanov transform and by translating the zero mode
$c\in\a/(2\pi\Lambda)$, the $\blambda$-summand in
\eqref{eq:magnetic-path-integral} is equal to the same expression with
$\Omega_{\blambda}$ replaced by $\Omega_{\blambda}^{\mathrm h}$, namely it is bounded by
$$
C\,
\|F\|_{\mathcal L^{\infty,p}_{\Lambda,\m}}
e^{
-\frac{1}{4\pi}\|\Omega_{\blambda}^{\mathrm h}\|_2^2
-\frac{1}{2\pi}
\langle \Omega_{\blambda}^{\mathrm h},\nu_{\z,\m}^{\mathrm h}\rangle_2
-\frac{1}{4\pi}\|\nu_{\z,\m}^{\mathrm h}\|_{g,0}^2
},
$$
where $C$ is independent of $\blambda$ and $F$. 
The map $\blambda\longmapsto \|\Omega_{\blambda}^{\mathrm h}\|_2^2$ is a positive-definite quadratic form on the lattice $\Lambda^{2\bbg}$, while
$\blambda\mapsto\langle \Omega_{\blambda}^{\mathrm h},\nu_{\z,\m}^{\mathrm h}\rangle_2$
is linear. Hence the series converges after completing the square.

The independence of the auxiliary defect graph follows from Lemma \ref{lem:defect-graph-independence}.
\end{proof}
\begin{corollary}\label{cor:rotation-covariance-magnetic}
Let
$$
r_{\btheta}\bv
:=
\bigl((z_1,r_{\theta_1}v_1),\ldots,(z_{\nm},r_{\theta_{\nm}}v_{\nm})\bigr),
\qquad
\btheta=(\theta_1,\ldots,\theta_{\nm})\in\R^{\nm},
$$
where $r_{\theta_j}$ denotes rotation by angle $\theta_j$ in the oriented tangent space
$T_{z_j}\Sigma$.
Then, for every $F\in\mathcal{L}^{\infty,p}_{\Lambda,\m}(\Sigma;\a)$, we have
$$
\la FV^g_{(0,\m)}(r_{\btheta}\bv)\ra_{\Sigma,g}
=
e^{-i\sum_{j=1}^{\nm}\la Q,m_j\ra\,\theta_j}
\la FV^g_{(0,\m)}(\bv)\ra_{\Sigma,g}.
$$
\end{corollary}
\begin{proof}
It is enough to prove the claim when only one tangent vector is rotated and $F\in\mathcal{E}_{\Lambda}^{\m}(\Sigma;\a)$, since the
general case then follows by continuity and applying the one-point statement successively.

Up to relabelling the marked points, we may assume $m_1\le \cdots \le m_{\nm}.$ We choose the canonical defect graph $z_1\to z_2\to \cdots \to z_{\nm}.$ We first rotate the first tangent vector. Let
$$
\bv'=\bigl((z_1,r_{\theta_1}v_1),(z_2,v_2),\ldots,(z_{\nm},v_{\nm})\bigr).
$$
We shall show that
\begin{equation}\label{eq:one-point-rotation-canonical-proof}
\la FV^g_{(0,\m)}(\bv')\ra_{\Sigma,g}
=
e^{-i\la Q,m_1\ra\theta_1}\,
\la FV^g_{(0,\m)}(\bv)\ra_{\Sigma,g}.
\end{equation}

Let $\xi_1$ be the first edge of the canonical defect graph, joining $z_1$ to $z_2$.
Choose another smooth arc $\widetilde\xi$ from $z_1$ to $z_2$ such that $\widetilde\xi'(0)=\widetilde\lambda_1\,r_{\theta_1}v_1$ and $\widetilde\xi'(1)=\widetilde\lambda_2\,v_2$
for some $\widetilde\lambda_1,\widetilde\lambda_2>0$, and such that $\widetilde\xi$ does not
meet the other edges of the defect graph. Replacing $\xi_1$ by $\widetilde\xi$ gives a
defect graph adapted to $\bv'$.

Since the correlation functions do not depend on the defect graph (Lemma \ref{lem:defect-graph-independence}), it suffices to compare
the corresponding regularized magnetic curvature terms. Let $D$ be the domain bounded
by $\xi_1$ and $\widetilde\xi$. We orient $\partial D$ so that $\widetilde\xi$ is positively
oriented and $\xi_1$ negatively oriented.

On the domain
$D$, these two primitives differ by the jump across the first defect line. For the
canonical defect graph, this jump is exactly $2\pi m_1$, because $z_1$ is a leaf of the
tree. Thus
$$
I_{x_0}^{\widetilde\xi}(\nu^{\mathrm h}_{\z,\m})
=
I_{x_0}^{\xi}(\nu^{\mathrm h}_{\z,\m})-2\pi m_1
\qquad\text{on }D.
$$
Therefore, by definition of the regularized magnetic curvature term,
\begin{align*}
&\int_\Sigma^{\reg}
\la Q,I_{x_0}^{\widetilde\xi}(\nu^{\mathrm h}_{\z,\m})\ra K_g\,\d v_g
-
\int_\Sigma^{\reg}
\la Q,I_{x_0}^{\xi}(\nu^{\mathrm h}_{\z,\m})\ra K_g\,\d v_g \\
&=
\int_D \la Q,-2\pi m_1\ra K_g\,\d v_g
+
4\pi\la Q,m_1\ra
\left(
\int_{\xi_1}k_g\,\d\ell_g-\int_{\widetilde\xi}k_g\,\d\ell_g
\right).
\end{align*}

Applying the Gauss--Bonnet theorem, we obtain
$$
-\int_D K_g\,\d v_g
=
2\int_{\widetilde\xi}k_g\,\d\ell_g
-
2\int_{\xi_1}k_g\,\d\ell_g
+
2\theta_1.
$$
Multiplying by $2\pi\la Q,m_1\ra$ and substituting this into the previous identity, the geodesic-curvature terms cancel and we get
$$
\int_\Sigma^{\reg}
\la Q,I_{x_0}^{\widetilde\xi}(\nu^{\mathrm h}_{\z,\m})\ra K_g\,\d v_g
-
\int_\Sigma^{\reg}
\la Q,I_{x_0}^{\xi}(\nu^{\mathrm h}_{\z,\m})\ra K_g\,\d v_g
=
4\pi\la Q,m_1\ra\theta_1.
$$

This proves \eqref{eq:one-point-rotation-canonical-proof}. The same argument applies when rotating any one tangent vector, and rotating the tangent vectors one after another and applying the one-point formula successively, we obtain
$$
\la FV^g_{(0,\m)}(r_{\btheta}\bv)\ra_{\Sigma,g}
=
e^{-i\sum_{j=1}^{\nm}\la Q,m_j\ra\theta_j}
\la FV^g_{(0,\m)}(\bv)\ra_{\Sigma,g}.
$$
This proves the corollary.
\end{proof}
\paragraph{\emph{Electric operators}}

We now construct electric operators in the presence of the magnetic background
introduced above. The purely electric case is obtained by taking
$\m=0$.

Let $x\in\Sigma_{\z}:=\Sigma\setminus\{z_1,\ldots,z_{\nm}\}$ and let
$\alpha\in\a$. For an equivariant field
$u\in H^s_{\Gamma,\m}(\widetilde\Sigma_{\z};\a)$, we define the regularized electric
operator by
$$
V^g_{\alpha,\epsilon}(u,x)
:=
\epsilon^{-\frac{|\alpha|^2}{2}}
e^{\bi\left\langle \alpha,u_{g,\epsilon}(x)\right\rangle},
$$
where $u_{g,\epsilon}$ denotes a $g$-regularization of $u$ at scale $\epsilon$.
When $u=\Phi_g^{\blambda,\m}$, we simply write
$$
V^g_{\alpha,\epsilon}(x)
:=
V^g_{\alpha,\epsilon}(\Phi_g^{\blambda,\m},x).
$$

Let $\mathbf x=(x_1,\ldots,x_{\ne})\in \Sigma_{\z}^{\ne}$
be distinct points, also distinct from the magnetic insertion points $\z$, and let $\balpha=(\alpha_1,\ldots,\alpha_{\ne})\in(\Lambda^*)^{\ne}$
be electric charges. We set
$$
V^{g,\epsilon}_{(\balpha,0)}(u,\mathbf x)
:=
\prod_{j=1}^{\ne}
V^g_{\alpha_j,\epsilon}(u,x_j),
$$
and again write $V^{g,\epsilon}_{(\balpha,0)}(\mathbf x)$ when the field is
$\Phi_g^{\blambda,\m}$.
The condition $\alpha_j\in\Lambda^*$ ensures that the operator is well-defined on
the compactified field, namely it is invariant under the addition of $2\pi\Lambda$ to the zero mode. 

We introduce the electric shift
$$
u_{\mathbf x}(y):=\bi\sum_{j=1}^{\ne}\alpha_j G_g(y,x_j),\qquad y\in\Sigma,
$$
where $G_g$ is the Green function of $\Delta_g$ with zero average. Notice that
$u_{\mathbf x}\in H^s(\Sigma;\a_{\C})$ for every $s<1$. This will be the shift that appears
after applying the Girsanov transform to the regularized electric insertions.

We now define the corresponding electric-magnetic norm. For $p\geq1$, set
\begin{equation}\label{eq:L-infty-p-electric-magnetic-norm}
\begin{aligned}
\|F\|_{\mathcal L^{\infty,p}_{\Lambda,\balpha,\m}}
:=
\sup_{\blambda\in\Lambda^{2\bbg}}
\Bigg(
\int_{\a/(2\pi\Lambda)}
\E\Big[
&
e^{-\frac{1}{2\pi}
\langle \d\x_g,\Omega_{\blambda}+\nu_{\z,\m}^{\mathrm h}\rangle_2
-\frac{1}{4\pi}\|\d f_{\blambda}\|_2^2}
\\
&\times
\left|
F\left(
c+\x_g
+u_{\mathbf x}
+I^{\bsigma}_{x_0}(\Omega_{\blambda})
+I^{\bm\xi}_{x_0}(\nu_{\z,\m}^{\mathrm h})
\right)
\right|^p
\Big]
\d c
\Bigg)^{1/p}.
\end{aligned}
\end{equation}
Here
$\Omega_{\blambda}
=
\Omega_{\blambda}^{\mathrm h}
+
\d f_{\blambda}$
as before.
We denote by $\mathcal L^{\infty,p}_{\Lambda,\balpha,\m}(\Sigma;\a)$
the completion of $\mathcal E^\m_\Lambda(\Sigma;\a)$ with respect to this norm.

As in Lemma \ref{L-infty-p-independent-representatives-magnetic}, the norm is
independent of the choice of representative of the cohomology class of
$\Omega_{\blambda}$:
\begin{lemma}\label{lem:L-infty-p-electric-magnetic-independent-representatives}
The norm $\|\cdot\|_{\mathcal L^{\infty,p}_{\Lambda,\balpha,\m}}$ does not depend on the choice of representatives $\Omega_{\blambda}\in[\Omega_{\blambda}]\in H^1_\Lambda(\Sigma;\a).$

\end{lemma}

For $F\in\mathcal E^\m_\Lambda(\Sigma;\a)$, we define the regularized electric-magnetic
correlation by
$$
\left\langle
F V^{g,\epsilon}_{(\balpha,0)}(\mathbf x)
V^g_{(0,\m)}(\bv)
\right\rangle_{\Sigma,g}
:=
\left\langle
F\,V^{g,\epsilon}_{(\balpha,0)}(\Phi_g^{\blambda,\m},\mathbf x)
V^g_{(0,\m)}(\bv)
\right\rangle_{\Sigma,g},
$$
where the right-hand side is understood through Definition
\ref{def:magnetic-correlation}, with the regularized electric factor inserted in the
expectation.

The path integral with electric and magnetic insertions is then defined, if the limit
exists, by
\begin{equation}\label{eq:electric-magnetic-path-integral-limit}
\left\langle
F V^g_{(\balpha,0)}(\mathbf x)
V^g_{(0,\m)}(\bv)
\right\rangle_{\Sigma,g}
:=
\lim_{\epsilon\to0}
\left\langle
F V^{g,\epsilon}_{(\balpha,0)}(\mathbf x)
V^g_{(0,\m)}(\bv)
\right\rangle_{\Sigma,g}.
\end{equation}
\begin{theorem}\label{thm:mixed operators}
Let $\bx=(x_1,\ldots,x_{n_{\mathfrak e}})\in \Sigma^{n_{\mathfrak e}}$
and $\bv=((z_1,v_1),\ldots,(z_{n_{\mathfrak m}},v_{n_{\mathfrak m}}))
\in (T\Sigma)^{n_{\mathfrak m}},$ with all points $\bx,\z$
pairwise distinct. Let $\balpha=(\alpha_1,\ldots,\alpha_{n_{\mathfrak e}})
\in (\Lambda^*)^{n_{\mathfrak e}},
\,
\m=(m_1,\ldots,m_{n_{\mathfrak m}})
\in \Lambda^{n_{\mathfrak m}},$
and assume the magnetic neutrality condition
\begin{equation}\label{eq:mixed-neutrality}
\sum_{j=1}^{n_{\mathfrak m}}m_j=0,
\end{equation}
as well as the electric charge condition
\begin{equation}\label{eq:mixed-charge-cond}
\alpha_j-Q\in {\mathcal C}_+,
\qquad
j=1,\ldots,n_{\mathfrak e},
\end{equation}
where $\mathcal{C}_+ $ is the open positive Weyl chamber
\begin{equation}\label{eq:weyl chamber}
{\mathcal C}_+
=
\{u\in\mathfrak a:\langle u,e_i\rangle>0,\ i=1,\ldots,r\}.
\end{equation}

Then, for every $F\in\mathcal E^\m_\Lambda(\Sigma;\a),$
the limit
$$
\la\,
F V^g_{(\balpha,0)}(\bx)V^g_{(0,\m)}(\bv)
\ra_{\Sigma,g}
:=
\lim_{\epsilon\to0}
\la\,
F V^{g,\epsilon}_{(\balpha,0)}(\bx)V^g_{(0,\m)}(\bv)
\ra_{\Sigma,g}
$$
exists and is finite. Moreover, for every $p>1$, the map
$F\mapsto
\la\,
F V^g_{(\balpha,0)}(\bx)V^g_{(0,\m)}(\bv)
\ra_{\Sigma,g},
\,F\in\mathcal E^\m_\Lambda(\Sigma;\a),$
extends uniquely to a continuous linear functional on
$\mathcal L^{\infty,p}_{\Lambda,\balpha,\m}(\Sigma;\a).$
Moreover, for every $F\in\mathcal E_{\Lambda}^{\m}(\Sigma;\a),$
the following hold true:
\begin{itemize}
    \item[\textbf{(i)}] \textbf{Conformal anomaly.}
    If $g'=e^\rho g$, with $\rho\in C^\infty(\Sigma)$, then
    $$
    \frac{
    \la F V^{g'}_{(\balpha,0)}(\bx)V^{g'}_{(0,\m)}(\bv)\ra_{\Sigma,g'}
    }{
    \la F(\cdot-\frac{\bi}{2}Q\rho) V^{g}_{(\balpha,0)}(\bx)V^{g}_{(0,\m)}(\bv)\ra_{\Sigma,g}
    }
    =
    e^{
    \frac{c}{96\pi}\int_\Sigma \bigl(|d\rho|_g^2+2K_g\rho\bigr)\,\d v_g
    -\sum_{j=1}^{n_{\mathfrak e}}\Delta_{(\alpha_j,0)}\rho(x_j)
    -\sum_{j=1}^{n_{\mathfrak m}}\Delta_{(0,m_j)}\rho(z_j)
    },
    $$
    where
    $$
    c=\mathrm{rank}(\mathfrak g)-6\langle Q,Q\rangle,
    \qquad
    \Delta_{(\alpha,m)} =
    \langle\frac{\alpha}{2},\frac{\alpha}{2}-Q\rangle
    +\frac14 |m|^2.
    $$

    \item[\textbf{(ii)}] \textbf{Diffeomorphism invariance.}
    If $\psi:\Sigma'\to\Sigma$ is an orientation-preserving diffeomorphism, then
    $$
    \la
    F(\Phi_{\psi^*g}) V^{\psi^*g}_{(\balpha,0)}(\bx)
    V^{\psi^*g}_{(0,\m)}(\bv)
    \ra_{\Sigma',\psi^*g}
    =
    \la
    F(\Phi_g\circ\psi)
    V^g_{(\balpha,0)}(\psi(\bx))
    V^g_{(0,\m)}(\psi_*\bv)
    \ra_{\Sigma,g},
    $$
    where
    $$
    \psi(\bx):=(\psi(x_1),\ldots,\psi(x_{n_{\mathfrak e}})),
\quad    \psi_*\bv
    :=
    ((\psi(z_1),\d\psi_{z_1}v_1),\ldots,
    (\psi(z_{n_{\mathfrak m}}),\d\psi_{z_{n_{\mathfrak m}}}v_{n_{\mathfrak m}})).
    $$

    \item[\textbf{(iii)}] \textbf{Spin covariance.}
    Let
    $$r_{\btheta}\bv:=((z_1,r_{\theta_1}v_1),\ldots,(z_{n_{\mathfrak m}},r_{\theta_{n_{\mathfrak m}}}v_{n_{\mathfrak m}})),\qquad\btheta=(\theta_1,\ldots,\theta_{n_{\mathfrak m}}).
    $$
    Then
    $$
    \la F V^{g}_{(\balpha,0)}(\bx)V^{g}_{(0,\m)}(r_{\btheta}\bv)
    \ra_{\Sigma,g}= e^{-i\sum_{j=1}^{n_{\mathfrak m}}\langle Q,m_j\rangle\theta_j}\la F V^{g}_{(\balpha,0)}(\bx)V^{g}_{(0,\m)}(\bv)\ra_{\Sigma,g}.
    $$
\end{itemize}
\end{theorem}

\begin{proof}
$\,$
\paragraph{\textbf{Step 1: Existence}} We first note that the functional $u\mapsto F(u)V^{g,\varepsilon}_{(\balpha,0)}(u,\bx)$, $F\in\mathcal{E}_\Lambda^{\m}(\Sigma;\a)$  lies in $\mathcal{E}_\Lambda^{\m}(\Sigma;\a)$, and so Proposition \ref{prop:magnetic-path-integral-well-defined} ensures the existence of $\la \,FV^{g,\epsilon}_{(\balpha,0)}(\bx)V^g_{(0,\m)}(\bv)\ra_{\Sigma,g}$. As argued in the proof of \cite[Theorem 6.11]{Guillarmou:2023exh}, the Girsanov transform only works for real-valued Gaussians, and so one has to go through an analytic continuation argument to apply the transform.

For $\w=(w_1,\ldots,w_{n_{\mathfrak e}})\in\mathbb C^{n_{\mathfrak e}}$,  we define a map $\w\in \C^{\ne}\mapsto A(\w)$ by
\begin{align*}
 A(\w)
:=
\mathbb E\Big[
e^{-\frac1{2\pi}\langle d\x_g,\Omega_{\blambda}\rangle_2}
F(\Phi_g)W^{g,\epsilon}_{\w}(\Phi_g,\bx)
e^{-\frac{i}{4\pi}\langle QK_g,\Phi_g\rangle_g^{\reg}}
e^{-\sum_{i=1}^r\mu_iM^g_{\gamma e_i}(\Phi_g,\Sigma)}
\Big],
\end{align*}
where $$W^{g,\varepsilon}_{\w}(u,\bx):=\prod_{j=1}^{\ne}\varepsilon^{w_j^2|\alpha_j|^2/2}e^{w_j\la \alpha_j,u_{g,\varepsilon}(x_j)\ra},$$
and $\Phi_g=c+\x_g+I^{\bsigma}_{x_0}(\Omega_{\blambda})+I^{\bxi}_{x_0}(\nu^{\mathrm{h}}_{\z,\m})$ is the Toda field. 
We observe that $A(i,\ldots,i)$ recovers the integrand of $\la F V^{g,\epsilon}_{(\balpha,0)}V^{g}_{(0,\m)}(\bv)\ra$ and that the function $A$ is holomorphic on $\C^{\ne}$. For $\w\in\R^{\ne}$, the Cameron--Martin theorem implies that \begin{align}\label{eq:expection RHS}
    A(\w)=&e^{-\frac{1}{2}\E\left[\left(\sum_{j=1}^{\ne}w_j\la \alpha_j,\x_{g,\epsilon}(x_j)\ra\right)^2\right]}\prod_{j=1}^{\ne}\eps^{w_j^2|\alpha_j|^2/2}e^{\sum^{\ne}_{j=1}w_j\la \alpha_j,c\ra}\\&\notag\times\E\Big[
e^{-\frac1{2\pi}\langle d\x_g+\d u_{0,\epsilon},\Omega_{\blambda}\rangle_2}F(\Phi_g+u_{0,\eps})
e^{-\frac{\bi}{4\pi}\langle QK_g,\Phi_g+u_{0,\epsilon}\rangle_g^{\reg}}
e^{-\sum_{i=1}^r\mu_iM^g_{\gamma e_i}(\Phi_g+u_{0,\epsilon},\Sigma)}
\Big],
\end{align}
where $G_{\epsilon,\epsilon'}(x,x'):=\E[X_{g,\epsilon}(x)X_{g,\epsilon'}(x')]$ (with the convention that $X_{g,0}=X_g$) and $u_{\epsilon,\epsilon'}(x):=\sum^{\ne}_{j=1}w_j\alpha_jG_{\epsilon,\epsilon'}(x,x_j)$. 

We claim that the right-hand side of the above identity is holomorphic in $\w$. The difficulty is that the map 
$\w \mapsto M_{\gamma e_i}^g(\Phi_g+u_{0,\epsilon},\Sigma)$ is not clearly a.s. holomorphic in $\w$. We therefore regularize the singular shift. For fixed $\eps>0$, since $x\mapsto u_{0,\eps}(x)$ is continuous on $\Sigma$ and  holomorphic in
$\w$, we can choose a family $(u_{0,\eps,\delta})_{\delta>0}\subset C^\infty(\Sigma;\a)$ such that

\begin{itemize}
\item for each $\delta>0$, the map $(x,\w)\mapsto u_{0,\eps,\delta}(x)$ is smooth in $x$ and
holomorphic in $\w$,
\item for every compact $K\subset\C^{\ne}$,
$$
\sup_{\w\in K}\sup_{x\in\Sigma}|u_{0,\eps,\delta}(x)-u_{0,\eps}(x)|\rightarrow 0
\qquad\text{as }\delta\to0.
$$
\end{itemize}

We first show that the expectation, for fixed $\delta>0,$ \begin{equation}\label{eq:expectation delta}\E\Big[
e^{-\frac1{2\pi}\langle d\x_g+\d u_{0,\epsilon},\Omega_{\blambda}\rangle_2}
e^{-\frac{i}{4\pi}\langle QK_g,\Phi_g+u_{0,\epsilon}\rangle_g^{\reg}}
e^{-\sum_{i=1}^r\mu_iM^g_{\gamma e_i}(\Phi_g+u_{0,\epsilon,\delta},\Sigma)}
\Big]\end{equation} is holomorphic in $\w$.

\begin{lemma}[{\cite[Lemma 6.12]{Guillarmou:2023exh}}]\label{lem:toda-order-2-distribution}The random variable $M_{\gamma e_i}^g(\Phi_g
=
X_g+I_{x_0}^{\bsigma}(\Omega_{\blambda})
+I_{x_0}^{\bxi}(\nu_{\z,\m}^{\mathrm h}),\d x)$ is almost surely a random
distribution of order $2$ on $\Sigma$. More precisely, there exists an $L^2$ random variable
$D_{\Sigma,i}(\blambda,\m)$ such that for every $f\in C^\infty(\Sigma)$,
$$
\left|
\int_\Sigma f(x)\,M_{\gamma e_i}^g(\Phi_g,dx)
\right|
\le
D_{\Sigma,i}(\blambda,\m)\,
\bigl(\|f\|_\infty+\|\Delta_g f\|_\infty\bigr)
$$
almost surely.
\end{lemma}

Lemma \ref{lem:toda-order-2-distribution} implies that the integrand 
$$e^{-\frac1{2\pi}\la d\x_g+du_{0,\eps},\Omega_{\blambda}\ra_2}
F(\Phi_g+u_{0,\eps})
e^{-\frac{\bi}{4\pi}\la QK_g,\Phi_g+u_{0,\eps}\ra_g^{\reg}}
e^{-\sum_{i=1}^r\mu_iM_{\gamma e_i}^g(\Phi_g+u_{0,\eps,\delta},\Sigma)}
$$
is holomorphic in $\w$, and hence it remains to show that, for any compact subset $K\subset\C^{\ne}$, 
$$
\sup_{\w\in K}\E\left[\left|e^{-\frac1{2\pi}\la d\x_g+du_{0,\eps},\Omega_{\blambda}\ra_2}F(\Phi_g+u_{0,\eps})
e^{-\frac{\bi}{4\pi}\la QK_g,\Phi_g+u_{0,\eps}\ra_g^{\reg}}
e^{-\sum_{i=1}^r\mu_iM_{\gamma e_i}^g(\Phi_g+u_{0,\eps,\delta},\Sigma)}\right|\right]<\infty.
$$ This readily follows from H\"older's inequality and Proposition \ref{prop:exp-moment-shifted}.

It suffices to show that the integral \eqref{eq:expectation delta} converges locally uniformly with respect to $\w$ to the right-hand side of \eqref{eq:expection RHS}. Using the local uniform convergence of $u_{0,\eps,\delta}(x)$ and H\"older's inequality, we only need to control the convergence of the GMC term, i.e., show that locally uniformly in $\w$, $$\E\left[\left|e^{\sum_{i=1}^r\mu_i(M^g_{\gamma e_i}(\Phi_g+u_{0,\eps,\delta},\Sigma)-M^g_{\gamma e_i}(\Phi_g+u_{0,\eps},\Sigma))}-1\right|^2\right]\longrightarrow0,$$ as $\delta\rightarrow0$.

By Proposition \ref{prop:exp-moment-shifted} with $$f_{i,\delta,\w}(x)
:=
e^{i\gamma\la e_i,u_{0,\eps,\delta}(x)\ra}
-
e^{i\gamma\la e_i,u_{0,\eps}(x)\ra},$$ we obtain that 
for every
$\alpha\ge0$,
$$
\E\left[
e^{\alpha
\left|
\sum_{i=1}^r \mu_i\Big(
M^g_{\gamma e_i}(\Phi_g+u_{0,\eps,\delta},\Sigma)
-
M^g_{\gamma e_i}(\Phi_g+u_{0,\eps},\Sigma)
\Big)
\right|
}
\right]
\le
\exp\!\Big(
C\alpha V_{\delta,\w}\bigl(1+C\alpha U_{\delta,\w}e^{C\alpha^2U_{\delta,\w}^2}\bigr)
\Big),
$$
where
$$
V_{\delta,\w}:=\sum_{i=1}^r\int_\Sigma |f_{i,\delta,\w}(x)|\,d\v_g(x)
$$
and
$$
U_{\delta,\w}^2
:=
\sum_{i,j=1}^r
\iint_{\Sigma^2}
|f_{i,\delta,\w}(x)|\,|f_{j,\delta,\w}(y)|\,
e^{\gamma^2\la e_i,e_j\ra G_g(x,y)}
\,d\v_g(x)\,d\v_g(y)
$$ locally converge to 0.
Therefore
$$
\sup_{\w\in K}
\E\left[
e^{
\left|
\sum_{i=1}^r \mu_i\left(
M^g_{\gamma e_i}(\Phi_g+u_{0,\eps,\delta},\Sigma)
-
M^g_{\gamma e_i}(\Phi_g+u_{0,\eps},\Sigma)
\right)
\right|
}
\right]\xrightarrow[\delta\rightarrow0]{}1.
$$The inequality $$|e^z-1|^2\leq (e^{|z|}-1)^2\leq e^{2|z|}-1,\quad z\in\C,$$ together with the above estimate proves the claim.
Hence the right-hand side of \eqref{eq:expection RHS} is holomorphic on $\C^{\ne}$, and since
it coincides with the holomorphic function $A(\w)$ on $\R^{\ne}$, the identity theorem implies
that \eqref{eq:expection RHS} remains valid for all $\w\in\C^{\ne}$. In particular, we may evaluate
it at $\w=(\bi,\ldots,\bi).$

For every $\eps>0$, this amounts to
\begin{align}
&\la FV^{g,\eps}_{(\balpha,0)}(\bx)V^g_{(0,\m)}(\bv)\ra_{\Sigma,g}\notag\\
&=
\left(\frac{\v_g(\Sigma)}{\det'(\Delta_g)}\right)^{r/2}
\sum_{\blambda\in\Lambda^{2\bbg}}
e^{-\frac{1}{4\pi}\|\Omega_{\blambda}\|_2^2-\frac{1}{4\pi}\|\nu^{\mathrm h}_{\z,\m}\|_{g,0}^2-\frac{1}{2\pi}\la\Omega_{\blambda},\nu^{\mathrm h}_{\z,\m}\ra_2}
\int_{\a/(2\pi\Lambda)}
\mathcal I_{\eps}(\blambda,c)\,dc,
\label{eq:eps-representation-mixed}
\end{align}
where
\begin{align*}
\mathcal I_{\eps}(\blambda,c)
:={}&
e^{
-\frac12\E\Big[\Big(\sum_{j=1}^{\ne}\bi\la\alpha_j,\x_{g,\eps}(x_j)\ra\Big)^2\Big]}
\prod_{j=1}^{\ne}\eps^{-|\alpha_j|^2/2}
e^{\bi\sum_{j=1}^{\ne}\la\alpha_j,c\ra}
\\
&\times
\E\Big[
e^{-\frac1{2\pi}\la d\x_g+du_{0,\eps},\Omega_{\blambda}\ra_2}
F(\Phi_g+u_{0,\eps})
e^{-\frac{\bi}{4\pi}\la QK_g,\Phi_g^{\blambda}+u_{0,\eps}\ra_g^{\reg}}
e^{-\sum_{i=1}^r\mu_iM^g_{\gamma e_i}(\Phi_g^{\blambda}+u_{0,\eps},\Sigma)}
\Big],
\end{align*}
and $u_{0,\eps}(x):=\bi\sum_{j=1}^{\ne}\alpha_j\,G_{\eps,0}(x,x_j).$

We use the shorthand $u_0:=u_{0,0}$. The deterministic prefactor in $\mathcal I_\eps(\blambda,c)$
converges, by the defining properties of a $g$-regularization, to the usual Wick-renormalized factor
$$
\mathcal{W}_g(\balpha,\bx):=e^{-\frac{1}{2}\sum_{j=1}^{\ne}|\alpha_j|^2W_g(x_j)-\sum_{1\le j<k\le \ne}\la\alpha_j,\alpha_k\ra G_g(x_j,x_k)}.
$$
Thus the only nontrivial point is the convergence of the Toda potential term. For $i=1,\ldots,r$, the candidate limit 
$M^g_{\gamma e_i}(\Phi_g^{\blambda}+u_0,\Sigma)$ has second moment given by an integral of the form
\begin{align}
\iint_{\Sigma^2}
&e^{
i\gamma\la e_i,u_0(x)-u_0(y)\ra
+\gamma^2\la e_i,e_i\ra G_g(x,y)-\frac{\gamma^2\la e_i,e_i\ra}{2}(W_g(x)+W_g(y))
}
\nonumber\\
&\qquad \times
e^{
i\gamma\la e_i,
I_{x_0}^{\bsigma}(\Omega_{\blambda})(x)-I_{x_0}^{\bsigma}(\Omega_{\blambda})(y)
+I_{x_0}^{\bxi}(\nu_{\z,\m}^{\mathrm h})(x)-I_{x_0}^{\bxi}(\nu_{\z,\m}^{\mathrm h})(y)
\ra
}\,d\v_g(x)d\v_g(y).
\label{eq:l2-moment-toda}
\end{align}
where the integrability is determined by the singularities
near the electric insertions $x_j$. In local coordinates centered at $x_j$, using
$
G(x,x_j)=-\log|x|+O(1)
$
and
$
\gamma\la e_i,Q\ra=\frac{\gamma^2}{2}|e_i|^2-2,
$
we reduce to the finiteness of
\begin{equation}
\iint_{|x|,|y|\le 1}
\frac{|x|^{\gamma\la e_i,\alpha_j\ra}\,|y|^{\gamma\la e_i,\alpha_j\ra}}
{|x-y|^{\gamma^2|e_i|^2}}
\,dx\,dy.
\label{eq:local-toda-integrability}
\end{equation}
This integral is finite as soon as
$$
\gamma\la e_i,\alpha_j\ra>-2+\frac{\gamma^2}{2}|e_i|^2
=\gamma\la e_i,Q\ra,
$$
that is, $\la e_i,\alpha_j-Q\ra>0.$ Because $\alpha_j-Q\in\mathcal C_+$ by assumption, this holds for every simple root $e_i$ and every
$j$. Moreover, by Proposition \ref{prop:exp-moment-shifted}, there exists a constant $C>0$,
independent of $\eps$, $\blambda$, and $c$, such that
\begin{equation}\label{eq:uniform-exp-bound-mixed}
\E\Bigg[
\exp\Bigg(
\Bigg|
\sum_{i=1}^r \mu_i M^g_{\gamma e_i}(\Phi_g^{\blambda}+u_{0,\eps},\Sigma)
\Bigg|
\Bigg)
\Bigg]
\le C,
\end{equation}
 and the family $\Big(M^g_{\gamma e_i}(\Phi_g^{\blambda}+u_{0,\eps},\Sigma)\Big)_{\eps>0}$ is Cauchy in $L^2$. This shows that $M^g_{\gamma e_i}(\Phi_g^{\blambda}+u_{0,\eps},\Sigma)
$ converges in $L^2$ to $M^g_{\gamma e_i}(\Phi_g^{\blambda}+u_0,\Sigma)$. 

Combining this convergence with the uniform exponential bound
\eqref{eq:uniform-exp-bound-mixed}, we obtain that, as $\epsilon\rightarrow0,$
$$
\E\Big[
e^{-\frac1{2\pi}\la d\x_g+du_{0,\eps},\Omega_{\blambda}\ra_2}F(\Phi_g+u_{0,\eps})
e^{-\frac{i}{4\pi}\la QK_g,\Phi_g^{\blambda}+u_{0,\eps}\ra_g^{\reg}}
e^{-\sum_{i=1}^r\mu_iM^g_{\gamma e_i}(\Phi_g^{\blambda}+u_{0,\eps},\Sigma)}
\Big]
$$
$$
\longrightarrow
\E\Big[
e^{-\frac1{2\pi}\la d\x_g+du_0,\Omega_{\blambda}\ra_2}F(\Phi_g+u_{0})
e^{-\frac{i}{4\pi}\la QK_g,\Phi_g^{\blambda}+u_0\ra_g^{\reg}}
e^{-\sum_{i=1}^r\mu_iM^g_{\gamma e_i}(\Phi_g^{\blambda}+u_0,\Sigma)}
\Big].
$$
Consequently, for every fixed $(\blambda,c)$, $\mathcal I_\eps(\blambda,c)$ converges to $ \mathcal I(\blambda,c),$
where
\begin{align*}
\mathcal I(\blambda,c)
:={}&
e^{
-\frac12\sum_{j=1}^{\ne}|\alpha_j|^2W_g(x_j)
-\sum_{1\le j<k\le \ne}\la\alpha_j,\alpha_k\ra G_g(x_j,x_k)
}
e^{\bi\sum_{j=1}^{\ne}\la\alpha_j,c\ra}
\\
&\times
\E\Big[
e^{-\frac1{2\pi}\la d\x_g+du_0,\Omega_{\blambda}\ra_2}F(\Phi_g+u_0)
e^{-\frac{\bi}{4\pi}\la QK_g,\Phi_g^{\blambda}+u_0\ra_g^{\reg}}
e^{-\sum_{i=1}^r\mu_iM^g_{\gamma e_i}(\Phi_g^{\blambda}+u_0,\Sigma)}
\Big].
\end{align*}

The justification of the passage to the limit in \eqref{eq:eps-representation-mixed} can be done similarly to the proof of \cite[Theorem 6.11]{Guillarmou:2023exh}. This shows that the limit
$$\la F V^g_{(\balpha,0)}(\bx)V^g_{(0,\m)}(\bv)\ra_{\Sigma,g}=\lim_{\epsilon\to0}\la FV^{g,\epsilon}_{(\balpha,0)}(\bx)V^g_{(0,\m)}(\bv)\ra_{\Sigma,g}
$$
exists and is finite, with
\begin{equation}\label{eq:final-mixed-representation}
\begin{aligned}
&\la FV^g_{(\balpha,0)}(\bx)V^g_{(0,\m)}(\bv)\ra_{\Sigma,g}\\&
=\left(\frac{\v_g(\Sigma)}{\det'(\Delta_g)}\right)^{r/2}
\sum_{\blambda\in\Lambda^{2\bbg}}
e^{-\frac{1}{4\pi}\|\Omega_{\blambda}\|_2^2
-\frac{1}{4\pi}\|\nu^{\mathrm h}_{\z,\m}\|_{g,0}^2
-\frac{1}{2\pi}\la\Omega_{\blambda},\nu^{\mathrm h}_{\z,\m}\ra_2}
\int_{\a/(2\pi\Lambda)}
\mathcal I(\blambda,c)\,\d c.
\end{aligned}
\end{equation}
Following the same estimates, one also obtains
$$\left|\la F V^g_{(\balpha,0)}(\bx)V^g_{(0,\m)}(\bv)\ra_{\Sigma,g}\right|\le C\|F\|_{\mathcal L^{\infty,p}_{\Lambda,\balpha,\m}},$$ and thus the map $F\mapsto
\la
F V^g_{(\balpha,0)}(\bx)V^g_{(0,\m)}(\bv)
\ra_{\Sigma,g}$
extends uniquely to a continuous linear functional on
$\mathcal L^{\infty,p}_{\Lambda,\balpha,\m}(\Sigma;\a).$

\paragraph{\textbf{Step 2: Conformal anomaly}} We compare the representation \eqref{eq:final-mixed-representation} for the metrics
$g$ and $g'=e^\rho g$. First, the Gaussian normalization factor satisfies the Polyakov--Alvarez formula \eqref{eq:Polyakov-formula}
$$
\left(\frac{\v_{g'}(\Sigma)}{\det'(\Delta_{g'})}\right)^{r/2}
=
\left(\frac{\v_g(\Sigma)}{\det'(\Delta_g)}\right)^{r/2}
e^{
\frac{r}{96\pi}\int_\Sigma \bigl(|d\rho|_g^2+2K_g\rho\bigr)\d\v_g
}.
$$
Next, the electric Wick-renormalization factor transforms as
\begin{align*}
e^{
-\frac12\sum_{j=1}^{n_{\mathfrak e}}|\alpha_j|^2W_{g'}(x_j)
-\sum_{1\le j<k\le n_{\mathfrak e}}\la\alpha_j,\alpha_k\ra G_{g'}(x_j,x_k)
}
=&
e^{
-\frac12\sum_{j=1}^{n_{\mathfrak e}}|\alpha_j|^2W_{g}(x_j)
-\sum_{1\le j<k\le n_{\mathfrak e}}\la\alpha_j,\alpha_k\ra G_{g}(x_j,x_k)
}\\
&\times
e^{
-\sum_{j=1}^{n_{\mathfrak e}}
\la \frac{\alpha_j}{2},Q-\frac{\alpha_j}{2}\ra \rho(x_j)
}.
\end{align*}
For the magnetic deterministic factor, Lemma~\ref{lem:toda-reg-norm} gives
$$e^{-\frac1{4\pi}\|\nu^{\mathrm h}_{\z,\m}\|_{g',0}^2}
=
e^{-\frac1{4\pi}\|\nu^{\mathrm h}_{\z,\m}\|_{g,0}^2
-\sum_{j=1}^{n_{\mathfrak m}}\frac{|m_j|^2}{4}\rho(z_j)
}.$$
For the curvature term, the conformal-change formula (Lemmas \ref{lem:conformal change of metrics} and \ref{lem:conformal-change-magnetic-curvature}) gives
\begin{equation}\label{eq:curvature-background-weyl-proof}
\begin{aligned}
-\frac{\bi}{4\pi}
\langle QK_{g'},
\Phi_{g'}^{\blambda,\m}+u_{\bx}^{g'}\rangle_{g'}^{\reg}
 =&
-\frac{\bi}{4\pi}
\langle QK_g,
\Phi_g^{\blambda,\m}
+u_{\bx}^{g}
-\frac{\bi}{2}Q\rho
\rangle_{g}^{\reg}
\\&-
\frac{6\langle Q,Q\rangle}{96\pi}
\int_\Sigma
\left(|d\rho|_g^2+2K_g\rho\right)\d \v_g.
\end{aligned}
\end{equation}
Next, we turn to Toda interaction terms. For each simple root
$e_i$,
\begin{equation}\label{eq:toda-weyl-proof}
M^{g'}_{\gamma e_i}
\left(
\Phi_{g'}^{\blambda,\m}
+u_{\bx}^{g'},\Sigma
\right)
=
M^{g}_{\gamma e_i}
\left(
\Phi_g^{\blambda,\m}
+u_{\bx}^{g}
-\frac{\bi}{2}Q\rho,\Sigma
\right).
\end{equation}
Indeed, the Weyl change of the volume form, the Wick renormalization, and the shift
$-\frac{\bi}{2}Q\rho$ cancel exactly since $\gamma\langle e_i,Q\rangle
=
\frac{\gamma^2}{2}|e_i|^2-2.$
Combining the above identities, we obtain that, for each fixed
$(\blambda,c)$-summand:
$$
\begin{aligned}
\mathcal S_{g'}^{\blambda,c}
\left(F;\balpha,\bx,\m,\bv\right)=&
e^{
\frac{r-6\langle Q,Q\rangle}{96\pi}
\int_\Sigma
\left(|d\rho|_g^2+2K_g\rho\right)\d \v_g
-\sum_{j=1}^{n_{\mathfrak e}}
\langle
\frac{\alpha_j}{2},
\frac{\alpha_j}{2}-Q
\rangle
\rho(x_j)
-\sum_{j=1}^{n_{\mathfrak m}}
\frac{|m_j|^2}{4}\rho(z_j)}\\
&\times
\mathcal S_{g}^{\blambda,c}
(
F(\cdot-\frac{\bi}{2}Q\rho);
\balpha,\bx,\m,\bv
).
\end{aligned}
$$
Here $\mathcal S_g^{\blambda,c}$ denotes the $(\blambda,c)$-summand in the limiting
representation \eqref{eq:final-mixed-representation} in the metric $g$. The multiplicative factor is independent of
$\blambda$ and $c$, so it factors out of the zero-mode integral and the instanton sum. This proves the desired conformal anomaly formula.

\paragraph{\textbf{Step 3: Diffeomorphism invariance}}
Let $\psi:\Sigma'\to\Sigma$ be an orientation-preserving diffeomorphism. We choose on
$\Sigma'$ the transported geometric data: the symplectic basis
$\psi^{-1}(\bsigma)$, the dual cohomology basis $\psi^*\eta_1,\ldots,\psi^*\eta_{2\bbg}$,
the base point $\psi^{-1}(x_0)$, and the transported defect graph $\psi^{-1}(\bxi)$.

Then for each $\blambda\in\Lambda^{2\bbg}$,
$$
\psi^*\Omega_{\blambda}
=
\sum_{j=1}^{2\bbg}\lambda_j\,\psi^*\eta_j,
\qquad
I^{\psi^{-1}(\bsigma)}_{\psi^{-1}(x_0)}(\psi^*\Omega_{\blambda})
=
I^{\bsigma}_{x_0}(\Omega_{\blambda})\circ\psi,
$$
and similarly
$$
I^{\psi^{-1}(\bxi)}_{\psi^{-1}(x_0)}(\psi^*\nu^{\mathrm h}_{\psi(\z),\m})
=
I^{\bxi}_{x_0}(\nu^{\mathrm h}_{\psi(\z),\m})\circ\psi.
$$
Moreover,
$$
G_{\psi^*g}(x,y)=G_g(\psi(x),\psi(y)),
\qquad
W_{\psi^*g}(x)=W_g(\psi(x)),
\qquad
\x_{\psi^*g}\stackrel{\mathrm{law}}{=}\x_g\circ\psi.
$$
The quantities $\|\Omega_{\blambda}\|_2^2$, $\|\nu^{\mathrm h}_{\z,\m}\|_{g,0}^2$, and
$\la\Omega_{\blambda},\nu^{\mathrm h}_{\z,\m}\ra_2$ are preserved under pullback by $\psi$,
and the regularized curvature terms are preserved by the diffeomorphism (Lemma \ref{lem:diffeomorphism invariance}). This concludes the diffeomorphism invariance.

\paragraph{\textbf{Step 4: Spin covariance}}
The only
dependence on $\bv$ comes from the magnetic sector,
and Corollary~\ref{cor:rotation-covariance-magnetic} applies verbatim and gives
$$
\la FV^g_{(\balpha,0)}(\bx)V^g_{(0,\m)}(r_\theta\bv)\ra_{\Sigma,g}
=
e^{-\bi\sum_{j=1}^{\nm}\la Q,m_j\ra\theta_j}
\la FV^g_{(\balpha,0)}(\bx)V^g_{(0,\m)}(\bv)\ra_{\Sigma,g}.
$$

This completes the proof.

\end{proof}

\paragraph{\emph{Electro-magnetic operators}} We set $n_{\mathfrak{e}}=n_{\mathfrak{m}}$ and define the path integral with electro-magnetic operators by $$\la FV^g_{(\balpha,\m)}(\bv)\ra_{\Sigma,g}:=\lim_{t\rightarrow1}\la FV^g_{(\balpha,0)}(\bx(t))V^g_{(0,\m)}(\bv)\ra_{\Sigma,g}$$ for $F\in\mathcal{E}_{\Lambda}^{\m}(\Sigma;\a)$, where $\bx(t)=(x_1(t),\ldots,x_{n_{\mathfrak{m}}}(t))$ with $x_j:[0,1]\rightarrow \Sigma$ any $C^1$ curve such that $x_j(1)=z_j$ and $\dot{x}_j(t)=v_j.$
\begin{theorem}\label{thm:electro-magnetic operators}
    Under the same assumptions as in Theorem \ref{thm:mixed operators}, the map $F\in\mathcal{E}^{\m}_{\Lambda}(\Sigma;\a)\mapsto\la FV^g_{(\balpha,\m)}(\bv)\ra_{\Sigma,g}$ satisfies the following properties:
    \begin{itemize}
        \item{\textbf{Existence:}} It is well-defined and extends to $F\in\mathcal{L}^{\infty,p}_{\Lambda,\balpha,\m}$.
        \item{\textbf{Conformal anomaly:}} If $g'=e^\rho g$ for some $\rho\in C^\infty(\Sigma),$ then  
        $$\frac{\la FV^{g'}_{(\balpha,\m)}(\bv)\ra_{\Sigma,g'}}{
    \la F(\cdot-\frac{\bi}{2}Q\rho)V^{g}_{(\balpha,\m)}(\bv)\ra_{\Sigma,g}}=
e^{\frac{c}{96\pi}\int_\Sigma (|\d\rho|_g^2+2K_g\rho)\d \v_g
    -\sum_{j=1}^{n_{\mathfrak m}}{\Delta_{(\alpha_j,m_j)}}\rho(z_j)}.$$
   
         \item{\textbf{Diffeomorphism invariance:}}  If $\psi:\Sigma'\to\Sigma$ is an orientation-preserving diffeomorphism, then
    $$
    \la F(\Phi_{\psi^* g})V^{\psi^*g}_{(\balpha,\m)}(\bv)\ra_{\Sigma',\psi^*g}
    =
    \la F(\Phi_g\circ\psi )V^g_{(\balpha,\m)}(\psi_*\bv)\ra_{\Sigma,g}.
    $$
        \item{\textbf{Spins:}} Let $r_\theta\bv:=\bigl((z_1,r_{\theta_1}v_1),\ldots,(z_{n_{\mathfrak m}},r_{\theta_{n_{\mathfrak m}}}v_{n_{\mathfrak m}})\bigr)$.
    Then
    $$
    \la FV^{g}_{(\balpha,\m)}(r_\theta\bv)\ra_{\Sigma,g}
    =
    e^{i\sum_{j=1}^{n_{\mathfrak m}}\la \alpha_j-Q,m_j\ra\,\theta_j}
    \la FV^{g}_{(\balpha,\m)}(\bv)\ra_{\Sigma,g}.
    $$
    \end{itemize}
\end{theorem}
\begin{proof}
For $t\in[0,1)$,
the quantity $\la FV^g_{(\balpha,0)}(\bx(t))V^g_{(0,\m)}(\bv)\ra_{\Sigma,g}$ is well-defined by Theorem \ref{thm:mixed operators}. It therefore remains to prove the existence of the limit, as the conformal anomaly, diffeomorphism invariance, and spin covariance follow immediately from
Theorem~\ref{thm:mixed operators} by passing to the limit $t\to1$.

\medskip

We observe that the prefactors in \eqref{eq:final-mixed-representation} converge in the limit $t\rightarrow1.$ since the harmonic magnetic form
$\nu^{\mathrm h}_{\z,\m}$ is of the form $m_j\,\d\theta$ in local polar coordinates near
$z_j$, the quantity $e^{\bi\langle \alpha_j, I^{\bxi}_{x_0}(\nu^{\mathrm h}_{\z,\m})(x_j(t))\rangle}$ has a limit as $x_j(t)\to z_j$ with prescribed tangent $v_j$. It remains to focus on the expectation term. Set
$M_t:=\sum_{i=1}^r \mu_i M^g_{\gamma e_i}(\Phi_g^{\blambda}+u_t,\Sigma),$
with
$u_t(x)=\bi\sum_{j=1}^{n_{\mathfrak m}}\alpha_j G_g(x,x_j(t)),$
Using H\"older's inequality it is enough to prove that, as $t\to 1,$
$$
\sup_{\blambda}\sup_{c\in \a/(2\pi\Lambda)}
\E\Big[\big|e^{-M_t}-e^{-M_1}\big|^q\Big]\longrightarrow 0
$$
for some $q>1$.
Using H\"older's inequality once again, together with Proposition~\ref{prop:exp-moment-shifted},
it is enough to show that for every $A>0$, as $t\to1$
$$
\sup_{\blambda}\sup_{c\in \a/(2\pi\Lambda)}
\E\Big[e^{A|M_t-M_1|}-1\Big]\longrightarrow 0.
$$
Applying Proposition~\ref{prop:exp-moment-shifted}, this follows provided the associated quantities
$$
V_t:=\sum_{i=1}^r \int_\Sigma
\left|
e^{i\gamma\langle e_i,u_t(x)\rangle}
-
e^{i\gamma\langle e_i,u_0(x)\rangle}
\right|\,\d\v_g(x)
$$
and
$$\begin{aligned}
U_t^2:=\sum_{i,j=1}^r \iint_{\Sigma^2}
\left|
e^{i\gamma\langle e_i,u_t(x)\rangle}
-
e^{i\gamma\langle e_i,u_0(x)\rangle}
\right|
&\left|
e^{i\gamma\langle e_j,u_t(y)\rangle}
-
e^{i\gamma\langle e_j,u_0(y)\rangle}
\right|
\\&\times e^{\gamma^2\langle e_i,e_j\rangle G_g(x,y)}
\,\d\v_g(x)\d\v_g(y)
\end{aligned}$$
tend to $0$ as $t\to1$.

This is exactly the local integrability statement already used in the proof of
Theorem~\ref{thm:mixed operators}: near each point $z_j$, the singularity is controlled by
$$
\frac{|x-a|^{\gamma\langle e_i,\alpha_j\rangle}
      |y-a|^{\gamma\langle e_i,\alpha_j\rangle}}
{|x-y|^{\gamma^2|e_i|^2}},
$$
which is integrable under the assumption \eqref{eq:mixed-charge-cond}.
This gives $V_t\to0$ and $U_t\to0$ as $t\to1$, and we may pass to the limit $t\to1$ in \eqref{eq:final-mixed-representation},
which proves that
$$
\lim_{t\to1}
\la V^g_{(\balpha,0)}(\bx(t))V^g_{(0,\m)}(\bv)\ra_{\Sigma,g}
$$
exists and is finite. This proves the existence of
$\la V^g_{(\balpha,\m)}(\bv)\ra_{\Sigma,g}$ and completes the proof.

\end{proof}

\section{Correlation functions on the Riemann sphere}\label{sec:riemann sphere}
In this section, we apply Theorem \ref{thm:electro-magnetic operators} to write $n$-point correlation functions on the Riemann sphere $\hat{\C}$ as a Dotsenko--Fateev type integral. We consider the metric
$$
g_0=|z|_+^{-4}|\d z|^2,
\qquad |z|_+=\max\{|z|,1\}.
$$
Let $\z=(z_1,\ldots,z_n)\in\hat{\C}$ be $n$ pairwise distinct points and $v_1,\ldots,v_n$ be the associated unit tangent vectors. Let $\m=(m_1,\ldots,m_n)\in \Lambda$ be magnetic charges associated to $\z$ satisfying $\sum_{j=1}^nm_j=0.$

As in the rank one case, the zero mode integral shows that the correlation functions vanish unless
\begin{equation}\label{eq:sphere-neutrality}
\sum_{j=1}^n\alpha_j-2Q\in -\gamma\sum_{i=1}^r \N e_i.
\end{equation}
In particular, if $\alpha_j-Q\in\mathcal C_+$ for all $j$, then non-trivial correlation functions on the sphere can only start at $n\ge 3$. Whenever \eqref{eq:sphere-neutrality} holds, we write
\begin{equation}\label{eq:sphere-screening-vector-n}
2Q-\sum_{j=1}^n\alpha_j=\gamma\sum_{i=1}^r s_i e_i,
\qquad
\bs=(s_1,\ldots,s_r)\in\N^r.
\end{equation}

Fix a defect graph $\bxi$ associated with $(\bv,\m)$, and let
$I^{\bxi}_0(\nu^{\mathrm h}_{\z,\m})$
be the corresponding magnetic primitive, normalized to vanish near infinity.
Using the formula
$$
G_{g_0}(x,y)=\log\frac1{|x-y|}+\log|x|_+ +\log|y|_+,
$$
performing the series expansion of the exponential in the expectation, and computing
the Fourier coefficients, we obtain
\begin{equation}\label{eq:sphere-n-point-screened}
\begin{aligned}
&\la V^{g_0}_{(\alpha_1,m_1)}(z_1)\cdots V^{g_0}_{(\alpha_n,m_n)}(z_n)\ra_{\hat{\C},g_0}
\\
&=
\mathrm{Vol}(\T(\gamma))
\left(\frac{\v_{g_0}(\hat{\C})}{\det'(\Delta_{g_0})}\right)^{r/2}
\prod_{i=1}^r \frac{(-\mu_i)^{s_i}}{s_i!}
\prod_{1\le j<k\le n}|z_j-z_k|^{\la\alpha_j,\alpha_k\ra}
\prod_{j=1}^n |z_j|_+^{4\Delta_{(\alpha_j,0)}}
\\
&\quad\times
e^{-\frac{1}{4\pi}\|\nu^{\mathrm h}_{\z,\m}\|_{g_0,0}^2
+\bi\sum_{j=1}^n\la\alpha_j,I^{\bxi}_0(\nu^{\mathrm h}_{\z,\m})(z_j)\ra}
\mathbb E\Bigg[
\prod_{i=1}^r
\left(
\int_{\C} F_i(x,\z)\,M^{g_0}_{\gamma e_i}(\x_{g_0},\d x)
\right)^{s_i}
\Bigg].
\end{aligned}
\end{equation}
where
\begin{equation}\label{eq:Fi-general-n}
F_i(x,\z)
:=
\prod_{j=1}^n
\left(
\frac{|x-z_j|}{|x|_+}
\right)^{\gamma\la e_i,\alpha_j\ra}
e^{\bi\gamma\la e_i,I^{\bxi}_0(\nu^{\mathrm h}_{\z,\m})(x)\ra}.
\end{equation}

We now rewrite the expectation term in \eqref{eq:sphere-n-point-screened} as a multiple integral.
\begin{lemma}\label{lem:n-point-DF}
Assume that
$$
2Q-\sum_{j=1}^n \alpha_j=\gamma\sum_{i=1}^r s_i e_i,
\qquad
\bs=(s_1,\ldots,s_r)\in\N^r.
$$
Then
\begin{equation}\label{eq:GMC-to-integral}
\begin{aligned}
&\mathbb E\Bigg[
\prod_{i=1}^r
\left(
\int_{\C} F_i(x,\z)\,M^{g_0}_{\gamma e_i}(\x_{g_0},\d x)
\right)^{s_i}
\Bigg]
\\
&=
\int_{\C^N}
\prod_{i=1}^r\prod_{a=1}^{s_i}
\Bigg(
\prod_{j=1}^n
|x_a^{(i)}-z_j|^{\gamma\la e_i,\alpha_j\ra}
\,e^{\bi\gamma\la e_i,I^{\bxi}_0(\nu^{\mathrm h}_{\z,\m})(x_a^{(i)})\ra}
\Bigg)
\\
&\qquad\qquad\times
\prod_{i=1}^r\prod_{1\le a<b\le s_i}
|x_a^{(i)}-x_b^{(i)}|^{\gamma^2\la e_i,e_i\ra}
\prod_{1\le i<j\le r}\prod_{a=1}^{s_i}\prod_{b=1}^{s_j}
|x_a^{(i)}-x_b^{(j)}|^{\gamma^2\la e_i,e_j\ra}
\prod_{i=1}^r\prod_{a=1}^{s_i}\d x_a^{(i)},
\end{aligned}
\end{equation}
where $N:=\sum_{i=1}^r s_i$.
\end{lemma}

\begin{proof}
For each $i\in\{1,\ldots,r\}$, set
$$
Y_i:=\int_{\C}F_i(x,\z)\,M^{g_0}_{\gamma e_i}(\x_{g_0},\d x).
$$
Since $s_i\in\N$, we may write
$$
Y_i^{s_i}
=
\int_{\C^{s_i}}
\prod_{a=1}^{s_i}F_i(x_a^{(i)},\z)
\prod_{a=1}^{s_i}M^{g_0}_{\gamma e_i}(\x_{g_0},\d x_a^{(i)}).
$$
Hence
\begin{equation}\label{eq:moment-expand-proof}
\begin{aligned}
\mathbb E\Bigg[\prod_{i=1}^rY_i^{s_i}\Bigg]
=
\int_{\C^N}
\Bigg(\prod_{i=1}^r\prod_{a=1}^{s_i}F_i(x_a^{(i)},\z)\Bigg)
\mathbb E\Bigg[
\prod_{i=1}^r\prod_{a=1}^{s_i}
M^{g_0}_{\gamma e_i}(\x_{g_0},\d x_a^{(i)})
\Bigg].
\end{aligned}
\end{equation}

We now compute the joint moment measure. Fix a $g_0$-regularization
$(\x_{g_0,\eps})_{\eps>0}$. By definition,
$$
M^{g_0}_{\gamma e_i}(\x_{g_0},\d x)
=
\lim_{\eps\to0}
\eps^{-\frac{\gamma^2}{2}\la e_i,e_i\ra}
e^{\bi\gamma\la e_i,\x_{g_0,\eps}(x)\ra}\,\d\v_{g_0}(x),
$$
and
$$
\d\v_{g_0}(x)=|x|_+^{-4}\,\d x.
$$
Therefore, for pairwise distinct screening variables,
\begin{equation}\label{eq:joint-moment-measure-proof}
\begin{aligned}
&\mathbb E\Bigg[
\prod_{i=1}^r\prod_{a=1}^{s_i}
M^{g_0}_{\gamma e_i}(\x_{g_0},\d x_a^{(i)})
\Bigg]
\\
&=
\prod_{i=1}^r\prod_{1\le a<b\le s_i}
e^{-\gamma^2\la e_i,e_i\ra G_{g_0}(x_a^{(i)},x_b^{(i)})}
\prod_{1\le i<j\le r}\prod_{a=1}^{s_i}\prod_{b=1}^{s_j}
e^{-\gamma^2\la e_i,e_j\ra G_{g_0}(x_a^{(i)},x_b^{(j)})}
\prod_{i=1}^r\prod_{a=1}^{s_i}|x_a^{(i)}|_+^{-4}\,\d x_a^{(i)}.
\end{aligned}
\end{equation}

Since
$$
G_{g_0}(x,y)=\log\frac{1}{|x-y|}+\log|x|_+ +\log|y|_+,
$$
we get
$$
e^{-\gamma^2\la e_i,e_j\ra G_{g_0}(x,y)}
=
|x-y|^{\gamma^2\la e_i,e_j\ra}
|x|_+^{-\gamma^2\la e_i,e_j\ra}
|y|_+^{-\gamma^2\la e_i,e_j\ra}.
$$
Substituting this into \eqref{eq:joint-moment-measure-proof}, we obtain
\begin{equation}\label{eq:joint-moment-measure-2-proof}
\begin{aligned}
\mathbb E\Bigg[
\prod_{i=1}^r\prod_{a=1}^{s_i}
M^{g_0}_{\gamma e_i}(\x_{g_0},\d x_a^{(i)})
\Bigg]
=
&\prod_{i=1}^r\prod_{1\le a<b\le s_i}
|x_a^{(i)}-x_b^{(i)}|^{\gamma^2\la e_i,e_i\ra}
\\
&\times
\prod_{1\le i<j\le r}\prod_{a=1}^{s_i}\prod_{b=1}^{s_j}
|x_a^{(i)}-x_b^{(j)}|^{\gamma^2\la e_i,e_j\ra}
\\
&\times
\prod_{i=1}^r\prod_{a=1}^{s_i}
|x_a^{(i)}|_+^{-4-\gamma^2\la e_i,\sum_{k=1}^rs_ke_k-e_i\ra}\,\d x_a^{(i)}.
\end{aligned}
\end{equation}

On the other hand, by \eqref{eq:Fi-general-n},
$$
\prod_{i=1}^r\prod_{a=1}^{s_i}F_i(x_a^{(i)},\z)
=
\prod_{i=1}^r\prod_{a=1}^{s_i}
\Bigg(
\prod_{j=1}^n
|x_a^{(i)}-z_j|^{\gamma\la e_i,\alpha_j\ra}
\Bigg)
|x_a^{(i)}|_+^{-\gamma\la e_i,\sum_{j=1}^n\alpha_j\ra}
e^{\bi\gamma\la e_i,I^{\bxi}_0(\nu^{\mathrm h}_{\z,\m})(x_a^{(i)})\ra}.
$$
Hence the total exponent of $|x_a^{(i)}|_+$ is
$$
-4-\gamma^2\la e_i,\sum_{k=1}^rs_ke_k-e_i\ra
-\gamma\la e_i,\sum_{j=1}^n\alpha_j\ra.
$$
Using the neutrality condition
$$
\sum_{j=1}^n\alpha_j+\gamma\sum_{k=1}^rs_ke_k=2Q,
$$
this becomes
$$
-4+\gamma^2|e_i|^2-\gamma\la e_i,2Q\ra.
$$
Since
$$
Q=\gamma\rho-\frac{2}{\gamma}\rho^\vee,
\qquad
\la e_i,\rho\ra=\frac{|e_i|^2}{2},
\qquad
\la e_i,\rho^\vee\ra=1,
$$
we get
$$
\gamma\la e_i,2Q\ra
=
2\gamma^2\la e_i,\rho\ra-4\la e_i,\rho^\vee\ra
=
\gamma^2|e_i|^2-4.
$$
Therefore the total exponent of $|x_a^{(i)}|_+$ is equal to $0$, and all the
$|x|_+$-factors cancel. Substituting this cancellation and \eqref{eq:joint-moment-measure-2-proof}
into \eqref{eq:moment-expand-proof} yields \eqref{eq:GMC-to-integral}.
\end{proof}

Applying Theorem \ref{thm:electro-magnetic operators}, we obtain that the correlation functions are conformally covariant: if $g=e^\rho g_0=g(z)|\d z|^2$ is a conformal metric, if $z_1,\ldots,z_n$ are distinct points on $\hat{\C}$ with associated unit tangent vectors $v_1,\ldots,v_n$, and if $\psi(z)=\frac{az+b}{cz+d}$ is a M\"obius map with $ad-bc=1$, then
\begin{equation}\label{eq:sphere-mobius-cov}
\la V^g_{(\balpha,\m)}(\psi_*\bv)\ra_{\hat{\C},g}
=
\prod_{j=1}^n
\left(
\frac{|\psi'(z_j)|^2g(\psi(z_j))}{g(z_j)}
\right)^{-\Delta_{(\alpha_j,m_j)}}
\la V^g_{(\balpha,\m)}(\bv)\ra_{\hat{\C},g}.
\end{equation}

In particular, if $n=3$ and $m_1\leq m_2\leq m_3,$ then the M\"obius covariance implies that the three-point function is determined up to a constant, denoted by
$
C_{\gamma,\bmu}(\balpha,\m)
$ and
called the \emph{structure constant}. More precisely,
\begin{equation}\label{eq:sphere-three-point-structure}
\begin{aligned}
\la V^g_{(\balpha,\m)}(\bv)\ra_{\hat{\C},g}
=&e^{\frac{c}{96\pi}\int_{\hat{\C}} (|\d\rho|_{g_0}^2+2K_{g_0}\rho)\,\d \v_{g_0}}
e^{-\bi\sum_{j=1}^3\la Q,m_j\ra\arg(v_j)}
\\&\times P_{(\balpha,\m)}(\z)
\prod_{j=1}^3 g(z_j)^{-\Delta_{(\alpha_j,m_j)}}
\,C_{\gamma,\bmu}(\balpha,\m),
\end{aligned}
\end{equation}
where
\begin{align*}
P_{(\balpha,\m)}(\z)
:=&
|z_1-z_3|^{2(\Delta_{(\alpha_2,m_2)}-\Delta_{(\alpha_1,m_1)}-\Delta_{(\alpha_3,m_3)})}
|z_2-z_3|^{2(\Delta_{(\alpha_1,m_1)}-\Delta_{(\alpha_2,m_2)}-\Delta_{(\alpha_3,m_3)})}\\&\times
|z_1-z_2|^{2(\Delta_{(\alpha_3,m_3)}-\Delta_{(\alpha_1,m_1)}-\Delta_{(\alpha_2,m_2)})}.
\end{align*}
In particular, if $e_1=\partial_x$ in the coordinate $z=x+\bi y\in\C$, then \begin{equation*}\begin{aligned}
C_{\gamma,\bmu}(\balpha,\m)=&\la V_{(\alpha_1,m_1)}(0,e_1)V_{(\alpha_2,m_2)}(1,e_1)V_{(\alpha_3,m_3)}(\infty,e_1) \ra_{\hat{\C},g_0}.\end{aligned}\end{equation*}

We now compute the three-point function in the metric $g_0$. Assume that $z_1,z_2,z_3\in \R$ with $z_1<z_2<z_3$ and that $v_1=v_2=v_3=e_1$. We choose the defect graph $z_1\to z_2\to z_3$, with branch cut $[z_1,z_3]$. We then take

\begin{equation*}\begin{aligned}\nu^{\mathrm{h}}_{\z,\m}(z)=&\Im\left(\frac{m_1(z_1-z_2)\d z}{(z-z_1)(z-z_2)}+\frac{m_3(z_3-z_2)\d z}{(z-z_2)(z-z_3)}\right)\\
=&\Im\left((\frac{m_1}{z-z_1}+\frac{m_2}{z-z_2}+\frac{m_3}{z-z_3})\d z\right)\end{aligned}\end{equation*}
so that
$
I^{\bxi}_0(\nu^{\mathrm h}_{\z,\m})(x)
=
-m_1\,\arg\frac{z_2-x}{z_1-x}
+
m_3\,\arg\frac{z_3-x}{z_2-x},
$
which vanishes on $]z_3,+\infty[$.

Set $\bs=(s_1,\ldots,s_r)\in\N^r$ by
\begin{equation*}
2Q-\sum_{j=1}^3\alpha_j=\gamma\sum_{i=1}^r s_i e_i.
\end{equation*} Combining \eqref{eq:sphere-n-point-screened}, \eqref{eq:GMC-to-integral} and the observation \cite[Lemma 7.1]{Guillarmou:2023exh} that 
\begin{equation*}
e^{-\frac{1}{4\pi}\|\nu^{\mathrm h}_{\z,\m}\|_{g_0,0}^2
+\bi\sum_{j=1}^3\la\alpha_j,I^{\bxi}_0(\nu^{\mathrm h}_{\z,\m})(z_j)\ra}
=
\prod_{1\le j<k\le 3}|z_j-z_k|^{\la m_j,m_k\ra}\prod_{j=1}^3|z_j|_+^{|m_j|^2}
e^{\bi\pi\la\alpha_2,m_1\ra-\bi\pi\la\alpha_3,m_3\ra},
\end{equation*}
we obtain

\begin{equation}\label{eq:sphere-three-point-DF}
\begin{aligned}
\la V^{g_0}_{\balpha,\m}(\z)\ra_{\hat{\C},g_0} 
=&\mathrm{Vol}(\T(\gamma))
\left(\frac{\v_{g_0}(\hat{\C})}{\det'(\Delta_{g_0})}\right)^{r/2}
\prod_{i=1}^r \frac{(-\mu_i)^{s_i}}{s_i!}
\prod_{1\le j<k\le 3}|z_j-z_k|^{\la\alpha_j,\alpha_k\ra+\la m_j,m_k\ra}
\\
&\times
\prod_{j=1}^3|z_j|_+^{4\Delta_{(\alpha_j,m_j)}}
e^{\bi\pi\la\alpha_2,m_1\ra-\bi\pi\la\alpha_3,m_3\ra}
\mathcal I_{\bs}(\balpha,\m;\z).
\end{aligned}
\end{equation}
where
\begin{equation}\label{eq:sphere-three-point-DF-integral}
\begin{aligned}
\mathcal I_{\bs}(\balpha,\m;\z)
:={}&
\int_{\C^N}
\prod_{i=1}^r\prod_{a=1}^{s_i}
(x_a^{(i)}-z_1)^{A_{i,1}}
(\overline{x_a^{(i)}-z_1})^{\bar A_{i,1}}
(x_a^{(i)}-z_2)^{A_{i,2}}
(\overline{x_a^{(i)}-z_2})^{\bar A_{i,2}}
\\
&\times
(x_a^{(i)}-z_3)^{A_{i,3}}
(\overline{x_a^{(i)}-z_3})^{\bar A_{i,3}}
\prod_{i=1}^r\prod_{1\le a<b\le s_i}
|x_a^{(i)}-x_b^{(i)}|^{\gamma^2\la e_i,e_i\ra}
\\
&\times
\prod_{1\le i<j\le r}\prod_{a=1}^{s_i}\prod_{b=1}^{s_j}
|x_a^{(i)}-x_b^{(j)}|^{\gamma^2\la e_i,e_j\ra}
\prod_{i=1}^r\prod_{a=1}^{s_i}\d x_a^{(i)},
\end{aligned}
\end{equation}
with
$$
A_{i,j}:=\frac{\gamma}{2}\la e_i,\alpha_j+m_j\ra,
\qquad
\bar A_{i,j}:=\frac{\gamma}{2}\la e_i,\alpha_j-m_j\ra,
\qquad j=1,2,3.
$$

The structure constant is then given by
\begin{equation}\label{eq:three-point-0-1-infty}
\begin{aligned}
C_{\gamma,
,\bmu}(\balpha,\m)
=&
\lim_{|z|\to\infty}
\la V^{g_0}_{(\alpha_1,m_1)}(0)V^{g_0}_{(\alpha_2,m_2)}(1)V^{g_0}_{(\alpha_3,m_3)}(z)\ra_{\hat{\C},g_0}\\
=&
\mathrm{Vol}(\T(\gamma))
\left(\frac{\v_{g_0}(\hat{\C})}{\det'(\Delta_{g_0})}\right)^{r/2}
\prod_{i=1}^r \frac{(-\mu_i)^{s_i}}{s_i!}\,
e^{i\pi\la\alpha_2,m_1\ra-i\pi\la\alpha_3,m_3\ra}\,
\mathcal I_{\bs}(\balpha,\m),
\end{aligned}
\end{equation}
where
\begin{equation}\label{eq:DF-0-1-infty}
\begin{aligned}
\mathcal I_{\bs}(\balpha,\m)
:={}&
\int_{\C^N}
\prod_{i=1}^r\prod_{a=1}^{s_i}
\bigl(x_a^{(i)}\bigr)^{\frac{\gamma}{2}\la e_i,\alpha_1+m_1\ra}
\bigl(\bar x_a^{(i)}\bigr)^{\frac{\gamma}{2}\la e_i,\alpha_1-m_1\ra}
\\
&\qquad\qquad\qquad\times
\bigl(1-x_a^{(i)}\bigr)^{\frac{\gamma}{2}\la e_i,\alpha_2+m_2\ra}
\bigl(1-\bar x_a^{(i)}\bigr)^{\frac{\gamma}{2}\la e_i,\alpha_2-m_2\ra}
\\
&\qquad\qquad\qquad\times
\prod_{i=1}^r\prod_{1\le a<b\le s_i}
|x_a^{(i)}-x_b^{(i)}|^{\gamma^2\la e_i,e_i\ra}
\\
&\qquad\qquad\qquad\times
\prod_{1\le i<j\le r}\prod_{a=1}^{s_i}\prod_{b=1}^{s_j}
|x_a^{(i)}-x_b^{(j)}|^{\gamma^2\la e_i,e_j\ra}
\prod_{i=1}^r\prod_{a=1}^{s_i}\d x_a^{(i)}.
\end{aligned}
\end{equation}

In \cite{fateev2007correlation}, when $\g=\mathfrak{sl}_{r+1}$, the integral \eqref{eq:DF-0-1-infty} was computed in the pure electric and semidegenerate case, i.e., $\m=0$ and $\alpha_1=\kappa\omega_\gamma$. Following the argument in \cite[Appendix A]{fateev2007correlation}, we extend this result to the case where $\alpha_1=\kappa\omega_r$ and $m_1=0$. 

To present the formula, we write $l(x)=\Gamma(x)/\Gamma(1-x)$ with $\Gamma$ the usual Gamma function and introduce the special function $\Upsilon,$ defined intially on $0<\Re(z)<q:=\gamma+\frac{2}{\gamma},$ where it admits the integral representation \begin{equation}\label{eq:Upsilon}
    \log\Upsilon(z)=\int^\infty_0\left(\left(\frac{q}{2}-z\right)^2\frac{e^{-t}}{2}-\frac{\left(\sinh\left(\left(\frac{q}{2}-z\right)\frac{t}{2\sqrt{2}}\right)\right)^2}{\sinh\left(\frac{t\gamma}{2\sqrt{2}}\right)\sinh\left(\frac{t}{\sqrt{2}\gamma}\right)}\right)\frac{\d t}{t},\end{equation} and then extended analytically to an entire function via the shift equations \begin{equation}\label{eq:shift}
    \Upsilon(z+\chi)=l\left(\frac{\chi}{2}z\right)\left(\frac{\chi}{\sqrt{2}}\right)^{1-\chi z}\Upsilon(z),\quad \chi\in\left\{\gamma,\frac{2}{\gamma}\right\}.
\end{equation} We note that if $\gamma^2\notin\mathbb Q$, the special function $\Upsilon$ has no poles but only simple zeros given by $(-\gamma\N-\frac{2}{\gamma}\N)\cup(q+\gamma\N+\frac{2}{\gamma}\N)$.

In the case where $\alpha_1=\kappa\omega_r$ and $\m=0$, the structure constant can be written as $$C_{\gamma,\bmu}(\kappa\omega_r,\alpha_2,\alpha_3,0,0,0)=\mathrm{Vol}(\T(\gamma))
\left(\frac{\v_{g_0}(\hat{\C})}{\det'(\Delta_{g_0})}\right)^{r/2}
\prod_{i=1}^r (-\pi\mu_i)^{s_i}C^{\fl}_\gamma(\kappa\omega_r,\alpha_2,\alpha_3),
$$
where $C^{\fl}_\gamma(\kappa\omega_r,\alpha_2,\alpha_3)$ is the imaginary Fateev--Litivinov formula \cite[eq. (3.10)]{Dupic_2019}:
\begin{equation}\label{eq:FL formula}
    \begin{aligned}
&C^{\fl}_\gamma(\kappa\omega_r,\alpha_2,\alpha_3)\\&=\left(l\left(\frac{\gamma^2}{2}\right)\left(\frac{\gamma}{\sqrt{2}}\right)^{2-\gamma^2}\right)^{\frac{\sqrt{2}}{\gamma}\la 2Q-\kappa\omega_r-\alpha_2-\alpha_3,\rho \ra}\frac{\Upsilon(\gamma)^r\Upsilon(\kappa)\prod_{e>0}\Upsilon\left(\left\la Q-\alpha_2,e\right\ra\right)\Upsilon\left(\left\la Q-\alpha_3,e\right\ra\right)}{\prod_{i,j=1}^{r+1}\Upsilon\left(\frac{\kappa}{r+1}+\la \alpha_2-Q,h_i \ra+\la \alpha_3-Q,h_j\ra\right)},
    \end{aligned}
\end{equation}with $\prod_{e>0}$ the product over all positive roots and  the weights of the first fundamental representation $h_1=\omega_1,$ $h_k=\omega_1-e_1-\cdots-e_{k-1}$ of highest weight $\omega_r.$

\begin{proposition}\label{prop:closed formula}Assume that $\alpha_1=\kappa\omega_r$ and that $m_1=0$. The structure constant \eqref{eq:three-point-0-1-infty} is given by \begin{equation}\label{eq:FL with magnetic charges}
    \begin{aligned}
        C_{\gamma,\bmu}(\kappa\omega_r,\alpha_2,\alpha_3,0,m_2,m_3)^2 =&\mathrm{Vol}(\T(\gamma))^2
\left(\frac{\v_{g_0}(\hat{\C})}{\det'(\Delta_{g_0})}\right)^{r}
\prod_{i=1}^r (\pi\mu_i)^{2s_i}\\&\times C^{\fl}_\gamma(\kappa\omega_r,\alpha_2+m_2,\alpha_3+m_3)C^{\fl}_\gamma(\kappa\omega_r,\alpha_2-m_2,\alpha_3-m_3)
    \end{aligned}
\end{equation}
\end{proposition}
We provide a proof in Appendix \ref{Appendix}.

\section{Segal's gluing axioms}\label{sec:Segal's axioms}

\subsection{Hilbert space of boundary states}\label{subsec:Hilbert-space}
We fix an orthonormal basis $(\eps_j)_{j=1,\ldots,r}$ of $\a$. Let $\mathbb T:=\{z\in\C:|z|=1\}$ denote the unit circle
A generic $\a$-valued Fourier series on $\mathbb T$ can be written as
$$
\widetilde\varphi(\theta)=c+\varphi(\theta),
\quad
\varphi(\theta)=\sum_{n\in\Z^*}\sum^r_{j=1}\varphi_{n,j} e^{\bi n\theta}\eps_j,
$$
with constant mode $c=\sum^r_{j=1}c_j\eps_j\in\a$ and Fourier coefficients $(\varphi_{n,j})_{n\in\Z^*,j=1,\ldots,r}$ parametrized by
$\varphi_{n,j}=\frac{x_{n,j}+\bi y_{n,j}}{2\sqrt n},\,\varphi_{-n,j}=\frac{x_{n,j}-\bi y_{n,j}}{2\sqrt n}$
for some $x_{n,j},y_{n,j}\in\mathbb R$.

For $s\in\mathbb R$, we define
$$
H^s(\mathbb T;\a)
:=
\left\{
\varphi=\sum_{n\in\mathbb Z}\varphi_n e^{\bi n\theta}:
\|\varphi\|_{H^s(\mathbb T;\a)}^2
:=
\sum_{n\in\mathbb Z}\sum^r_{j=1}
(1+|n|)^{2s}|\varphi_{n,j}|^2<\infty
\right\}.
$$
If the variables $x_{n,\ell},y_{n,\ell}$ are independent standard real Gaussian random
variables, then the corresponding random Fourier series belongs almost surely to
$H^s(\mathbb T;\a)$ for every $s<0$.

We next incorporate the compactification. Recall that the target torus is
$\T(\gamma)=\a/(2\pi\Lambda),$ with  $\Lambda=\gamma^{-1}\bigoplus_{i=1}^r\mathbb Z\omega_i^\vee.$ We define the equivariant Sobolev space
$$
H^s_\Lambda(\mathbb R;\a)
:=
\left\{
u\in H^s_{\mathrm{loc}}(\mathbb R;\a):
u(\theta+2\pi n)-u(\theta)\in 2\pi\Lambda
\quad\text{for all }n\in\mathbb Z
\right\}.
$$
Similarly to the discussion in Section \ref{sec:equivariant}, we have the identification $$\Lambda\times H^s(\T;\a)\rightarrow H^s_{\Lambda}(\R;\a),\quad (\lambda,\widetilde\varphi)\mapsto\pi^*\widetilde\varphi(\theta)+\lambda\theta.$$

Let $\Omega_{\mathbb T}:=(\R^{2r})^{\mathbb N^*}$, equipped with the cylinder sigma-algebra $\Sigma_{\T}$, and let $\mathbb P_{\mathbb T}$ be the probability measure
$$
\mathbb P_{\mathbb T}
:=
\bigotimes_{\substack{n\ge1,\\ j=1,\ldots,r}}
\left[
\frac{1}{2\pi}
e^{-\frac{1}{2}\left(x_{n,j}^2+y_{n,j}^2\right)}
\d x_{n,j}\d y_{n,j}
\right].
$$
The push-forward of $\mathbb P_{\mathbb T}$ by
$$
(x_{n,j},y_{n,j})_{\substack{n\ge1},j=1,\ldots,r}\longmapsto
\varphi(\theta)
=
\sum_{n>0}\sum_{j=1}^r
\frac{x_{n,j}+\bi y_{n,j}}{2\sqrt n}e^{\bi n\theta}\eps_j
+
\frac{x_{n,j}-\bi y_{n,j}}{2\sqrt n}e^{-\bi n\theta}\eps_j
$$
defines a measure on $H^s(\mathbb T;\a)$, still denoted
$\mathbb P_{\mathbb T}$, which is supported on $H^s(\T;\a)$ for any $s<0,$ i.e., $\P(\varphi\in H^s(\T;\a))=1$

We equip $\Lambda$ with the discrete sigma-algebra and the counting measure $\mu_\Lambda:=\sum_{\lambda\in\Lambda}\delta_\lambda,$ and $\T(\gamma)=\a/(2\pi\Lambda)$ with the Lebesgue measure induced by the Euclidean
structure on $\a$. We then define the boundary Hilbert space by
$$
\H:=\mathcal H_{\g,\gamma}
:=
L^2\!\left(
\T(\gamma)\times\Lambda\times\Omega_{\mathbb T},
\,
\d c\otimes\mu_\Lambda\otimes\mathbb P_{\mathbb T}
\right).
$$

Equivalently, the map
$$
(c,\lambda,\varphi)
\longmapsto
u_{\lambda,c,\varphi}(\theta)
:=
c+\lambda\theta+\varphi(\theta)
$$
pushes forward
$\d c\otimes\mu_\Lambda\otimes\mathbb P_{\mathbb T}^{\a}$
to a measure, denoted $\mu_0$, on $H^s_\Lambda(\mathbb R;\a)$. This measure is invariant
under global translations
$\tau_\ell:u\longmapsto u+2\pi\lambda,\, \lambda\in\Lambda,$
and therefore descends to the quotient $H^s_\Lambda(\mathbb R;\a)/(2\pi\Lambda).$
With this notation we may equivalently write
$$
\mathcal H\simeq L^2\left(H^s_\Lambda(\mathbb R;\a)/(2\pi\Lambda),\mu_0\right).
$$
Concretely, if $F:H^s_\Lambda(\mathbb R;\a)\to\mathbb C$ is measurable and invariant
under global translations by an element of the lattice $\Lambda$, then
\begin{equation}\label{eq:dmu0}
\int_{H^s_\Lambda(\mathbb R;\a)/(2\pi\Lambda)}
F(u)\,\d\mu_0(u)
=
\int_{\T(\gamma)}
\sum_{\lambda\in\Lambda}
\mathbb E
\left[
F\bigl(c+\lambda\theta+\varphi(\theta)\bigr)
\right]
\d c.
\end{equation}

\subsection{Dirichlet-to-Neumann map}\label{subsec:DN-map}

Let $(\Sigma,g)$ be a compact Riemann surface with analytic boundary
$\partial\Sigma=\bigsqcup_{j=1}^{\mathbbm b}\partial_j\Sigma$. We assume that each
$\partial_j\Sigma$ is equipped with an analytic parametrization
$\zeta_j:\mathbb T\longrightarrow \partial_j\Sigma,$
and that the metric $g$ is chosen so that each boundary component has length $2\pi$.
We denote by $\d\ell_g$ the Riemannian measure induced by $g$ on $\partial\Sigma$.

Boundary data will be written as $\widetilde{\bm\varphi}=(\widetilde\varphi_1,\ldots,\widetilde\varphi_{\bbb})
\in H^s(\mathbb T;\a)^{\bbb}.$
We use the following normalized distributional pairing between
$H^s(\mathbb T;\a)^{\mathfrak b}$ and
$H^{-s}(\mathbb T;\a)^{\mathfrak b}$:
$$
\langle \bm\varphi,\bm\psi\rangle_{\partial}
:=
\frac{1}{2\pi}
\sum_{j=1}^{\mathfrak b}
\int_0^{2\pi}
\la \varphi_j(e^{\bi\theta}),\psi_j(e^{\bi\theta})\ra
\,\d\theta
$$
whenever the integral expression makes sense.

For smooth boundary data
$\widetilde{\bm\varphi}\in C^\infty(\mathbb T;\a)^{\mathfrak b}$, we denote by
$P_\Sigma\widetilde{\bm\varphi}$ its harmonic extension to $\Sigma$, namely the unique smooth
$\a$-valued function on $\Sigma$ satisfying
$$
\Delta_g P_\Sigma\widetilde{\bm\varphi}=0
\quad\text{in }\Sigma,
\qquad
P_\Sigma\widetilde{\bm\varphi}|_{\partial_j\Sigma}
=
\widetilde\varphi_j\circ \zeta_j^{-1},
\quad j=1,\ldots,\bbb.
$$
For general $\widetilde{\bm\varphi}\in H^s(\mathbb T;\a)^{\bbb}$, the harmonic extension $P_\Sigma\widetilde{\bm\varphi}$
is understood weakly. More precisely, for every $u\in C^\infty(\mathbb T;\a)$,
$$
\lim_{r\to1^-}
\int_0^{2\pi}
\la P_\Sigma\bm\varphi(\zeta_j(re^{\bi\theta})),u(e^{\bi\theta})\ra
\,\d\theta
=
2\pi\langle \varphi_j,u\rangle .
$$

For later use, we define the Dirichlet-to-Neumann operator
$\mathbf D_\Sigma:
C^\infty(\mathbb T;\a)^{\bbb}
\rightarrow
C^\infty(\mathbb T;\a)^{\bbb}$
by
$$
\mathbf D_\Sigma\widetilde{\bm\varphi}
:=
\left(
-\partial_\nu P_\Sigma\widetilde{\bm\varphi}|_{\partial_j\Sigma}\circ\zeta_j
\right)_{j=1,\ldots,\bbb},
$$
where $\nu$ denotes the inward-pointing unit normal vector field along
$\partial\Sigma$.
Equivalently, if $\widetilde{\bm\varphi}=\sum_{\ell=1}^r\widetilde{\bm\varphi}^{(\ell)}\eps_\ell,$
then
$\mathbf D_\Sigma\widetilde{\bm\varphi}
=
\sum_{\ell=1}^r
\mathbf D_\Sigma^{\mathrm{scal}}\widetilde{\bm\varphi}^{(\ell)}\,\eps_\ell,$
where $\mathbf D_\Sigma^{\mathrm{scal}}$ is the usual scalar Dirichlet-to-Neumann map.
In particular, $\mathbf D_\Sigma$ is symmetric and non-negative with kernel given by $\ker \mathbf D_\Sigma
=
\left\{
(c,\ldots,c):c\in\a
\right\}.$
Indeed, Green's formula gives
\begin{equation}\label{eq:DN-energy-identity}
\int_\Sigma
\la \d P_\Sigma\widetilde{\bm\varphi},\d P_\Sigma\widetilde{\bm\varphi}\ra_g
\,\d\v_g
=
2\pi
\langle \widetilde{\bm\varphi},\mathbf D_\Sigma\widetilde{\bm\varphi}\rangle.
\end{equation}

We shall also need a Dirichlet-to-Neumann operator associated with cutting curves in
the interior of $\Sigma$. Let $\mathcal C'=\bigsqcup_{j=1}^{\bbb'}\mathcal C'_j$
be a finite collection of pairwise disjoint analytic simple closed curves contained in
the interior of $\Sigma$, and let $\zeta'_j:\mathbb T\to \mathcal C'_j$ be analytic parametrizations. For $\widetilde{\bm\varphi}=(\widetilde\varphi_1,\ldots,\widetilde\varphi_{\bbb'})
\in C^\infty(\mathbb T;\a)^{\bbb'},$
we denote by $P_{\mathcal C'}\widetilde{\bm\varphi}$ the harmonic function on
$\Sigma\setminus\mathcal C'$ with zero boundary condition on $\partial\Sigma$
and with trace
$$
P_{\mathcal C'}\widetilde{\bm\varphi}|_{\mathcal C'_j}
=
\widetilde\varphi_j\circ(\zeta'_j)^{-1},
\quad j=1,\ldots,\mathbbm b'.
$$

Let $\nu_-$ and $\nu_+$ denote the two unit normal vector fields along
$\mathcal C'_j$, pointing into the two sides of
$\Sigma\setminus\mathcal C'$. We define the interior Dirichlet-to-Neumann operator $\mathbf D_{\Sigma,\mathcal C'}:
C^\infty(\mathbb T;\a)^{\bbb'}
\longrightarrow
C^\infty(\mathbb T;\a)^{\bbb'}$
by the formula
\begin{equation}\label{eq:interior-DN-map}
\mathbf D_{\Sigma,\mathcal C'}\widetilde{\bm\varphi}
:=
-\left(
\left(\partial_{\nu_-}P_{\mathcal C'}\widetilde{\bm\varphi}|_{\mathcal C'_j}
+
\partial_{\nu_+}P_{\mathcal C'}\widetilde{\bm\varphi}|_{\mathcal C'_j}
\right)\circ \zeta'_j \right)_{j=1,\ldots,\bbb'}.
\end{equation}
The operator $\mathbf D_{\Sigma,\mathcal C'}$ is invertible and its inverse has
Schwartz kernel
\begin{equation}\label{eq:DN-inverse-kernel}
\mathbf D_{\Sigma,\mathcal C'}^{-1}(y,y')
=
\frac{1}{2\pi}G_{g,D}(y,y')\,\id_{\a},
\qquad
y\neq y'\in\mathcal C',
\end{equation}
where $\id_{\a}$ denotes the identity operator on $\a$.

For $k=\bbb$ or $k=\bbb'$, we introduce another operator $\mathbf D:
H^1(\mathbb T;\a)^k
\longrightarrow
L^2(\mathbb T;\a)^k$
by declaring that, for $\widetilde{\bm\varphi}
=
(\widetilde\varphi_1,\ldots,\widetilde\varphi_k)
\in C^\infty(\mathbb T;\a)^k,$
with
$$
\widetilde\varphi_j(\theta)
=
c_j\eps_\l+
\sum_{n\neq0}\sum_{\ell=1}^r
\varphi_{j,n,\ell}e^{\bi n\theta}\eps_\ell,
$$
one has
$$
(\mathbf D\widetilde{\bm\varphi})_j(\theta)
=
\sum_{n\neq0}\sum_{\ell=1}^r
|n|\varphi_{j,n,\ell}e^{\bi n\theta}\eps_\ell .
$$
Equivalently, we have
\begin{equation}\label{eq:model-DN-quadratic-form}
\left\langle
\mathbf D\widetilde{\bm\varphi},
\widetilde{\bm\varphi}
\right\rangle_2
=
2\sum_{j=1}^k\sum_{n>0}\sum_{\ell=1}^r
n|\varphi_{j,n,\ell}|^2
=
\frac12
\sum_{j=1}^k\sum_{n>0}\sum_{\ell=1}^r
\left(
x_{j,n,\ell}^2+y_{j,n,\ell}^2
\right).
\end{equation}
We then define the operators on $C^\infty(\T;\a)^\bbb$ and $C^\infty(\T;\a)^{\bbb'}$ 
\begin{equation}\label{eq:DN-smoothing-remainders}
\widetilde{\mathbf D}_{\Sigma}:=
\mathbf D_{\Sigma}-\mathbf D
,\qquad
\widetilde{\mathbf D}_{\Sigma,\mathcal C'}:=
\mathbf D_{\Sigma,\mathcal C'}-2\mathbf D,
\end{equation}
which are smoothing in the sense that their Schwartz kernels are smooth, and for every
$s,s'\in\mathbb R$, they extend to bounded operators
$$
\widetilde{\mathbf D}_{\Sigma}:
H^s(\mathbb T;\a)^{\bbb}
\longrightarrow
H^{s'}(\mathbb T;\a)^{\bbb},\qquad
\widetilde{\mathbf D}_{\Sigma,\mathcal C'}:
H^s(\mathbb T;\a)^{\bbb'}
\longrightarrow
H^{s'}(\mathbb T;\a)^{\bbb'},
$$respectively.

\subsection{Curvature terms on surfaces with boundary}
\label{subsec:boundary-curvature}

We extend the regularized curvature terms to compact Riemann
surfaces with analytic boundary. The construction is the $\a$-valued analogue of \cite[Section 8.3]{Guillarmou:2023exh}. Since all
forms and primitives are $\a$-valued, most statements follow componentwise after
passing to the fixed orthonormal basis $(\eps_\ell)_{\ell=1,\ldots,r}$ of $\a$.

Let $(\Sigma,g)$ be a compact connected Riemann surface with non-empty analytic boundary $\partial\Sigma=\bigsqcup_{j=1}^{\bbb}\partial_j\Sigma,$
and let $\zeta_j:\mathbb T\to\partial_j\Sigma, j=1,\ldots,\bbb,$
be analytic parametrizations. We assume that $g$ is admissible near the boundary, in
particular that each boundary component is geodesic and has length $2\pi$. We denote
by $p_j:=\zeta_j(1)\in\partial_j\Sigma$ the distinguished boundary point on the $j$-th boundary component, associated with the winding charges $\blambda=(\lambda_1,\ldots,\lambda_{\bbb})\in\Lambda^{\bbb}$.  Let $\bv=((z_1,v_1),\ldots,(z_{\nm},v_{\nm}))$
be unit tangent vectors with pairwise distinct base points, to which we attach magnetic charges
$\m=(m_1,\ldots,m_{\nm})\in\Lambda^{\nm}$. We impose the boundary magnetic neutrality condition
\begin{equation}\label{eq:boundary-magnetic-neutrality}
\sum_{i=1}^{\nm}m_i+\sum_{j=1}^{\bbb}\varsigma_j\lambda_j=0.
\end{equation}

We fix a canonical geometric basis $\bsigma=(\sigma_1,\ldots,\sigma_{2\bbg+\bbb-1})$ of the relative homology $H_1(\Sigma,\partial\Sigma),$ consisting of $2\bbg$ interior cycles $(\sigma_1,\ldots,\sigma_{2\bbg})=(a_1,b_1,\ldots,a_{\bbg},b_{\bbg})$ chosen as usual,
and $\bbb-1$ arcs $(\sigma_{2\bbg+1},\ldots,\sigma_{2\bbg+\bbb-1})=(d_1,\ldots,d_{\bbb-1})$ with endpoints on the boundary and no intersection with $\cup_j(a_j\cup b_j)$. We require each $d_j$ to have endpoints $p_j,\,p_{j+1}$ and make an angle $\pi/2$ with $\partial_j\Sigma$ at its endpoints. We recall from Lemma \ref{lem:relative-period-basis} the identification between  $\blambda^c=(\lambda^c_1,\ldots,\lambda^c_{2\bbg+\bbb-1})\in\Lambda^{2\bbg+\bbb-1}$ and $\Omega^c_{\blambda^c},$ where $\Omega^c_{\blambda^c}$ are chosen to be compactly supported in $\Sigma^\circ$ with $\int_{\sigma_k}\Omega^c_{\blambda^c}=2\pi\lambda^c_k$.

We recall from Proposition \ref{prop:magnetic-background} the closed $\a$-valued $1$-form $\nu_{\z,\m,\blambda}$ with $$(a_1,b_1,\ldots,a_{\bbg},b_{\bbg},\partial_1\Sigma,\ldots,\partial_{\bbb-1}\Sigma)$$ the choice of basis of $H_1(\Sigma$, where $\nu_{\z,\m,\blambda}$ is parametrized by $\blambda=(\lambda_1,\ldots,\lambda_\bbb)\in\Lambda^{\bbb}$ and satisfies \begin{equation}\label{eq:nu-boundary-condition}
    \int_{\partial_j\Sigma}\nu_{\z,\m,\blambda}=2\pi\varsigma_j\lambda_j, \qquad\sum^{n_\m}_{j=1}m_j+\sum^{\bbb}_{j=1}\varsigma_j\lambda_j=0.
\end{equation}
By Lemma \ref{lem:a-valued-absolute-relative-exactness} and possibly adding an exact $\a$-valued $1$-form $\d f$ with $\partial_\nu f|_{\partial\Sigma}=0$, we may require the form $\nu_{\z,\m,\blambda}$ to satisfy, for all $j=1,\ldots,\bbb$, 
\begin{equation}\label{nu-relative-period}
    \int_{\sigma_{2\bbg+j}}\nu_{\z,\m,\blambda}\in 2\pi\Lambda
\end{equation}
For notational convenience, we set $z_{n_{\m}+j}=p_j$ and $m_{n_{\m}+j}=\varsigma_j\lambda_j$ and we associate to the marked points $\z=(z_1,\ldots,z_{n_{\m}+j})$ the magnetic charges 
$$
\m(\blambda)=(m_1(\blambda),\ldots,m_{n_{\m}+\bbb}(\blambda))
=
(m_1,\ldots,m_{\nm},\varsigma_1\lambda_1,\ldots,\varsigma_{\bbb}\lambda_{\bbb})
\in\Lambda^{\nm+\bbb}.
$$
For $j>\nm$, we choose $v_j\in T_{p_{j-\nm}}\Sigma$ to be the unit normal vector to
the boundary component at $p_{j-\nm}$, with the sign chosen consistently with the
orientation convention of the boundary component. We also recall the lexicographic order on $\Lambda=\gamma^{-1}\bigoplus_{i=1}^r\Z\omega_i^\vee$ induced by the ordered basis $(\omega^\vee_1,\ldots,\omega_r^\vee)$.
\begin{definition}[Defect graph]
We consider a collection of $n_{\mathfrak{m}+\bbb}-1$ smooth arcs $\xi_p:[0,1]\rightarrow\Sigma$ with endpoints in $\z$ such that:
\begin{itemize}
    \item every arc is simple, and any two arcs do not intersect except possibly at their endpoints;
    \item if $\xi_p(0)=z_j$ and $\xi_p(1)=z_{j'}$, then
    $$
    \xi'_p(0)=\lambda_{p,j}v_j,\qquad \xi'_p(1)=\lambda_{p,j'}v_{j'}
    $$
    for some $\lambda_{p,j},\lambda_{p,j'}>0$ if $\xi_p(1)\in\partial\Sigma$ and $\lambda_{p,j'}<0$ if $\xi_p(1)\in\partial\Sigma$;
    \item if $\xi_p(0)=z_j$ and $\xi_p(1)=z_{j'}$, then $m_j\leq m_{j'}$;
    \item the induced graph with vertices $\z$ is connected, acyclic, and when viewed as a subset of $\Sigma$, the set $$
\mathcal{D}_{\bv,\bxi}:=\bigcup_{j=1}^{n_{\mathfrak{m}}-1}\xi_j([0,1])
$$ is homotopic to a point.
\end{itemize}
We call $\mathcal{D}_{\bv,\bxi}$ the \emph{defect graph} associated to $\bv$ and $\bxi$.
\end{definition}

As in Lemma \ref{exactness-defect-graph}, the form $\nu_{\z,\m,\blambda}$ is exact on $\Sigma\setminus\mathcal D_{\bv,\bxi}.$
Consequently, for $x_0\in\Sigma\setminus\mathcal D_{\bv,\bxi}$, the primitive
$$
I^{\bxi}_{x_0}(\nu_{\z,\m,\blambda})(x)
:=
\int_{\alpha_{x_0,x}}\nu_{\z,\m,\blambda}
$$
is well defined on $\Sigma\setminus\mathcal D_{\bv,\bxi}$, where
$\alpha_{x_0,x}$ is any smooth path in this complement.

We now define the magnetic curvature counterterm. For an edge $\xi_p$, choose
$x\in\xi_p((0,1))$. Let $\alpha_x$ be a positively oriented $C^1$ simple loop based
at $x$, contained in a small tubular neighborhood of the defect graph, meeting
$\mathcal D_{\bv,\bxi}$ only at $x$, and starting from the left face of
$\xi_p$. We set
\begin{equation}\label{eq:boundary-kappa-def}
\kappa(\xi_p)
:=
\int_{\alpha_x}\nu_{\z,\m,\blambda}.
\end{equation}
Equivalently,
$$
\kappa(\xi_p)
=
2\pi\sum_{z_j\in D_{\alpha_x}}m_j(\blambda),
$$
where $D_{\alpha_x}$ is the disk bounded by $\alpha_x$, with the convention that
boundary marked points $p_j$ carry charge $\varsigma_j\lambda_j$.

\begin{definition}[Regularized magnetic curvature term with boundary]
\label{def:boundary-magnetic-curvature-term}
We define
\begin{equation}\label{eq:boundary-magnetic-curvature-term}
\int_\Sigma^{\reg}
I^{\bxi}_{x_0}(\nu_{\z,\m,\blambda})K_g\d\v_g
:=
\int_{\Sigma\setminus\mathcal D_{\bv,\bxi}^{\partial}}
I^{\bxi}_{x_0}(\nu_{\z,\m,\blambda})K_g\d\v_g
-
2\sum_{p=1}^{\nm+\bbb-1}
\kappa(\xi_p)\int_{\xi_p}k_g\d\ell_g,
\end{equation}
where $k_g$ is the signed geodesic curvature of the oriented arc $\xi_p$.
\end{definition}

We also need the regularized curvature term associated with relative cohomology. Let $\Sigma_{\bsigma}$ be the surface cut along the curves of the basis, on which we define
$$
I^{\bsigma}_{x_0}(\Omega^c_{\blambda^c})(x)
:=
\int_{\alpha_{x_0,x}}\Omega^c_{\blambda^c}.
$$

\begin{definition}[Regularized relative cohomology curvature term]
\label{def:boundary-relative-curvature-term}
We define
\begin{equation}\label{eq:boundary-relative-curvature-term}
\begin{aligned}
\int_{\Sigma_{\bsigma}}^{\reg}
I^{\bsigma}_{x_0}(\Omega^c_{\blambda^c})K_g\d\v_g
:={}&
\int_{\Sigma_{\bsigma}}
I^{\bsigma}_{x_0}(\Omega^c_{\blambda^c})K_g\d\v_g
+
2\sum_{j=1}^{\bbg}
\left(
\int_{a_j}\Omega^c_{\blambda^c}\int_{b_j}k_g\d\ell_g
-
\int_{b_j}\Omega^c_{\blambda^c}\int_{a_j}k_g\d\ell_g
\right).
\end{aligned}
\end{equation}
\end{definition}
We stress that there is no extra boundary term since the boundary components are geodesic for admissible metrics and $\Omega^c_{\blambda^c}$ has compact support in $\Sigma^\circ$.

We collect some properties of the curvature terms, the proofs of which follow from the corresponding lemmas in \cite[Section 8.3]{Guillarmou:2023exh}.
\begin{lemma}\label{lem:boundary-curvature-conformal-change}
Let $\hat g=e^\rho g$, where $\rho\in C^\infty(\Sigma)$ satisfies
$\partial_\nu\rho|_{\partial\Sigma}=0.$
If $\Omega\in H^1_{\Lambda}(\Sigma;\a)$ with compact support in $\Sigma^\circ$, then
$$
\int_{\Sigma_{\bsigma}}^{\reg}
I^{\bsigma}_{x_0}(\Omega)K_{\hat g}\d\v_{\hat g}
=
\int_{\Sigma_{\bsigma}}^{\reg}
I^{\bsigma}_{x_0}(\Omega)K_g\d\v_g
+
\langle \d\rho,\Omega\rangle_2.
$$
Moreover,
$$
\int_\Sigma^{\reg}
I^{\bxi}_{x_0}(\nu_{\z,\m,\blambda})K_{\hat g}\d\v_{\hat g}
=
\int_\Sigma^{\reg}
I^{\bxi}_{x_0}(\nu_{\z,\m,\blambda})K_g\d\v_g
+
\langle \d\rho,\nu_{\z,\m,\blambda}\rangle_2.
$$
\end{lemma}

\begin{lemma}
\label{lem:boundary-relative-curvature-independence}
Let $\bsigma$ and $\bsigma'$ be two canonical geometric bases of
$H_1(\Sigma,\partial\Sigma)$. For $\Omega\in H^1_{\Lambda}(\Sigma;\a)$ with compact support in $\Sigma^\circ$, we have
$$
\int_{\Sigma_{\bsigma}}^{\reg}
I^{\bsigma}_{x_0}(\Omega)K_g\d\v_g
-
\int_{\Sigma_{\bsigma'}}^{\reg}
I^{\bsigma'}_{x_0}(\Omega)K_g\d\v_g
\in 8\pi^2\Lambda.
$$
\end{lemma}

\begin{lemma}
\label{lem:boundary-relative-curvature-diffeomorphism}
Let $\psi:\Sigma\to\Sigma$ be an orientation-preserving diffeomorphism. Then, for $x_0\in\Sigma$ and $\Omega\in H^1_\Lambda(\Sigma;\a)$ with compact support in $\Sigma^\circ,$ we have
$$
\int_{\Sigma_{\bsigma}}^{\reg}
I^{\bsigma}_{x_0}(\Omega)K_g\,d\v_g
=
\int_{\Sigma_{\psi(\bsigma)}}^{\reg}
I^{\psi(\bsigma)}_{\psi(x_0)}(\psi_*\Omega)
K_{\psi_*g}\,d\v_{\psi_*g}.
$$
\end{lemma}

\begin{lemma}
\label{lem:boundary-defect-graph-independence}
For any two defect graphs $\bxi$ and $\bxi'$ associated with the same data
$(\bv,\m,\blambda)$, one has
$$
\int_\Sigma^{\reg}
I^{\bxi}_{x_0}(\nu_{\z,\m,\blambda})K_g\d\v_g
-
\int_\Sigma^{\reg}
I^{\bxi'}_{x_0}(\nu_{\z,\m,\blambda})K_g\d\v_g
\in 8\pi^2\Lambda.
$$
\end{lemma}

\begin{lemma}
\label{lem:boundary-magnetic-curvature-diffeomorphism}
Let $\psi:\Sigma'\to\Sigma$ be an orientation-preserving diffeomorphism. Then
$$
\int_{\Sigma'}^{\reg}
I^{\bxi}_{x_0}(\nu_{\z,\m,\blambda})K_g\d\v_g
=
\int_{\Sigma}^{\reg}
I^{\psi\circ\bxi}_{\psi(x_0)}(\psi_*\nu_{\z,\m,\blambda})
K_{\psi_*g}\d\v_{\psi_*g}.
$$
\end{lemma}

We next record the additivity properties of the regularized curvature terms under
gluing. These are the $\a$-valued analogues of \cite[Lemmas 8.7 and 8.8]{Guillarmou:2023exh}.
Let $(\Sigma_i,g_i)$, $i=1,2$, be compact Riemann surfaces with admissible
metrics, with $\bbb_i>0$ analytic boundary components and genera $\bbg_i$. We consider the geometric data as above, decorated with a superscript or subscript $i=1,2.$ 

We choose the two defect graphs in the following compatible way. There is exactly
one arc $\xi^1_{j_0}$
in $\bxi_1$ having $p_{1,\bbb_1}=\zeta_{1,\bbb_1}(1)$ as endpoint, and exactly one arc $\xi^2_{j_0'}$
in $\bxi_2$ having $p_{2,\bbb_2}=\zeta_{2,\bbb_2}(1)$ as endpoint. We assume that
$p_{1,\bbb_1}$ is glued to $p_{2,\bbb_2}$, and that the tangent vectors of
$\xi^1_{j_0}$ and $\xi^2_{j_0'}$ match under the gluing, namely one is inward normal
and the other is outward normal. We also impose the condition
$\lambda_{1,\bbb_1}=\lambda_{2,\bbb_2}.$

Let $\blambda_i^-:=(\lambda_{i,1},\ldots,\lambda_{i,\bbb_i-1})$ and $
\blambda:=(\blambda_1^-,\blambda_2^-).$
We obtain a defect graph $\bxi$ on the glued surface $\Sigma$ by keeping all arcs of
$\bxi_1$ except $\xi^1_{j_0}$, all arcs of $\bxi_2$ except $\xi^2_{j_0'}$, and adding the
arc $\xi$ obtained by gluing $\xi^1_{j_0}$ and $\xi^2_{j_0'}$, with the induced
orientation. We also glue the magnetic forms
$\nu^{1}_{\z_1,\m_1,\blambda_1},\,$ 
$\nu^{2}_{\z_2,\m_2,\blambda_2}$
to obtain the form
$\nu_{\z,\m,\blambda}$
on $\Sigma$, where
$\z=(\z_1,\z_2),\,
\m=(\m_1,\m_2).$

Finally, since the forms $\Omega^{1,c}_{\blambda_i^c}$ are compactly supported in the
interiors of $\Sigma_i$, they glue to a compactly supported relative cohomology
representative on $\Sigma$, denoted by
$\Omega^c_{\blambda^c}
=
\Omega^{1,c}_{\blambda_1^c}
+
\Omega^{2,c}_{\blambda_2^c}.$

\begin{lemma}[Additivity under gluing of two surfaces]
\label{lem:boundary-curvature-additivity-gluing}
With the above choices, one has
\begin{equation}\label{eq:boundary-magnetic-curvature-gluing}
\begin{aligned}
\int_{\Sigma}^{\reg}
I^{\bxi}_{x_0}(\nu_{\z,\m,\blambda})K_g\d\v_g
=
\int_{\Sigma_1}^{\reg}
I^{\bxi_1}_{x^1_0}(\nu^1_{\z_1,\m_1,\blambda_1})K_{g_1}\d\v_{g_1}
+\int_{\Sigma_2}^{\reg}
I^{\bxi_2}_{x^2_0}(\nu^2_{\z_2,\m_2,\blambda_2})K_{g_2}\d\v_{g_2}.
\end{aligned}
\end{equation}
Moreover,
\begin{equation}\label{eq:boundary-relative-curvature-gluing}
\begin{aligned}
\int_{\Sigma_{\bsigma}}^{\reg}
I^{\bsigma}_{x_0}(\Omega^c_{\blambda^c})K_g\,d\v_g
=
\int_{(\Sigma_1)_{\bsigma_1}}^{\reg}
I^{\bsigma_1}_{x^1_0}(\Omega^{1,c}_{\blambda^c_1})K_{g_1}\d\v_{g_1}
+
\int_{(\Sigma_2)_{\bsigma_2}}^{\reg}
I^{\bsigma_2}_{x^2_0}(\Omega^{2,c}_{\blambda^c_2})K_{g_2}\d\v_{g_2}.
\end{aligned}
\end{equation}
\end{lemma}

We also need the corresponding self-gluing statement. Let $(\Sigma,g,\bzeta)$ be a
compact Riemann surface with admissible metric, genus $\bbg$, and
$\bbb\ge2$ analytic boundary components. Assume that
$\partial_{\bbb-1}\Sigma$ is outgoing and that $\partial_{\bbb}\Sigma$ is incoming.
Let
$\Sigma^\#$
be the surface obtained by identifying $\partial_{\bbb-1}\Sigma$ and
$\partial_{\bbb}\Sigma$ through the boundary parametrizations.

We consider marked tangent vectors $\bv=((z_1,v_1),\ldots,(z_{\nm},v_{\nm}))\in(T\Sigma)^{n_{\m}},$
 magnetic charges $\m=(m_1,\ldots,m_{\nm})\in\Lambda^{\nm},$
boundary winding data
$\blambda=(\lambda_1,\ldots,\lambda_{\bbb})\in\Lambda^{\bbb},$
and the magnetic form $\nu_{\z,\m,\blambda}.$ The signs are again chosen so that outgoing boundary components have
$\varsigma=-1$ and incoming ones have $\varsigma=+1$.
We assume that $\lambda_{\bbb-1}=\lambda_{\bbb}$ and write
$\blambda=(\blambda^-,\lambda_{\bbb},\lambda_{\bbb}),
\,
\blambda^-=(\lambda_1,\ldots,\lambda_{\bbb-2}).$

We choose the defect graph $\bxi$ on $\Sigma$ as follows. There is exactly one arc $\xi_{j_0}$
having $p_{\bbb}=\zeta_{\bbb}(1)$ as endpoint, and the other endpoint of this arc is
$p_{\bbb-1}=\zeta_{\bbb-1}(1)$. This arc is oriented according to the tangent vectors
at $p_{\bbb}$ and $p_{\bbb-1}$, one of which is inward normal and the other outward
normal. If $\bbb+\nm>2$, we also require that there are exactly two arcs having
$p_{\bbb-1}$ as endpoint; one of them is $\xi_{j_0}$, and we denote the other one by $\xi_{j_0'}.$
After self-gluing, we obtain a defect graph $\bxi^\#$ on $\Sigma^\#$ by keeping all
arcs of $\bxi$ and removing the arcs having $p_{\bbb-1}$ or $p_{\bbb}$ as endpoint. Finally, one decomposes $\nu_{\z,\m\blambda}$ before gluing as
$$
\nu_{\z,\m,\blambda}
=
\nu_{\z,\m,(\blambda^-,0,0)}
+
\nu_{\z,0,(0,\lambda_{\bbb},\lambda_{\bbb})},
$$
where both terms make sense by \cite[Lemma 3.8]{Guillarmou:2023exh}.

\begin{lemma}[Additivity under self-gluing]
\label{lem:boundary-curvature-additivity-self-gluing}
With the above notation, one has
\begin{equation}\label{eq:boundary-magnetic-curvature-self-gluing}
\begin{aligned}
\int_{\Sigma}^{\reg}
I^{\bxi}_{x_0}(\nu_{\z,\m,\blambda})K_g\d\v_g
=
\int_{\Sigma^\#}^{\reg}
I^{\bxi^\#}_{x_0}(\nu_{\z,\m,(\blambda^-,0,0)})K_g\d\v_g
+
\int_{(\Sigma^\#)_{\bsigma^\#}}^{\reg}
I^{\bsigma^\#}_{x_0}(\nu_{\z,0,(0,\lambda_{\bbb},\lambda_{\bbb})})K_g\d\v_g.
\end{aligned}
\end{equation}
\end{lemma}

\subsection{Amplitudes: definition and gluing}\label{subsec:amplitudes}

We now attach to a compact Riemann surface $(\Sigma,g)$, with or without boundary,
an amplitude. In the closed case this is nothing but the correlation
functional with electro-magnetic operator defined in
Theorem~\ref{thm:electro-magnetic operators}. When $\partial\Sigma\neq\emptyset$,
the amplitude becomes a functional of the boundary field and will be viewed as an
element of the boundary Hilbert space $\mathcal H^{\otimes\bbb}$ introduced in
Section~\ref{subsec:Hilbert-space}.

Boundary data will be encoded by
$$
\widetilde{\bm\varphi}_{\blambda}
=
(\widetilde\varphi_{\lambda_1,1},\ldots,\widetilde\varphi_{\lambda_{\bbb},\bbb})
\in \bigl(H^s_\Lambda(\mathbb R;\a)\bigr)^{\bbb},
\qquad
\widetilde\varphi_{\lambda_j,j}(\theta)=\lambda_j\theta+\widetilde\varphi_j(\theta),
$$
where $\widetilde{\bm\varphi}=(\widetilde\varphi_1,\ldots,\widetilde\varphi_{\bbb})
\in H^s(\mathbb T;\a)^{\bbb}$ is the single-valued part.
Let $P\widetilde{\bm\varphi}$ be the harmonic extension of the single-valued boundary
field (Section~\ref{subsec:DN-map}).
The winding data $\blambda$ is inserted through the closed $\a$-valued $1$-form
$\nu_{\z,\m,\blambda}$ (defined in Proposition \ref{prop:magnetic-background}), together with a
choice of defect graph $\mathcal D_{\bv,\bxi}$ and primitive $I^{\bxi}_{x_0}(\nu_{\z,\m,\blambda})$. We recall that $\nu_{\z,\m,\blambda}$ does \emph{not} belong to
$L^2(\Sigma;T^*\Sigma\otimes\a)$. For this reason we systematically use its
\emph{regularised} $L^2$-norm $\|\nu_{\z,\m,\blambda}\|_{g,0}$. By the same convention, whenever $\Omega$ is a smooth $\a$-valued $1$-form on $\Sigma$, we extend the notation and set
$$
\|\nu_{\z,\m,\blambda}+\Omega\|_{g,0}^2
:=
\|\nu_{\z,\m,\blambda}\|_{g,0}^2
+2\langle \nu_{\z,\m,\blambda},\Omega\rangle_2
+\|\Omega\|_{g,2}^2,
$$
so that $\|\cdot\|_{g,0}$ behaves formally as an $L^2$-norm on the affine space
$\nu_{\z,\m,\blambda}+C^\infty(\Sigma;T^*\Sigma\otimes\a)$.

We also need a class of test functionals suited to GFF on a surface with
boundary and to the non-uniqueness of the topological decomposition of equivariant fields.
Fix $s<0$. If $u\in H^s_\Lambda(\widetilde\Sigma_{\z};\a)$ is an equivariant distribution
on the universal cover of $\Sigma_{\z}:=\Sigma\setminus\{z_1,\ldots,z_{\nm}\}$, its
monodromies are recorded by the winding vector $\blambda$ along the boundary components,
the magnetic charges $\m$ around the punctures, and an auxiliary choice of interior cycle
data $\blambda^c\in\Lambda^{2\bbg+\bbb-1}$. Choosing $\blambda^c$ (non-uniquely) and a
representative of the corresponding cohomology class, one can write $u$ in the form
\begin{equation}\label{eq:boundary-decomposition-u}
u
=
\pi^*f
+
I^{\bsigma}_{x_0}(\Omega_{\blambda^c})
+
I^{\bxi}_{x_0}(\nu_{\z,\m,\blambda}),
\end{equation}
where $f\in H^s(\Sigma_{\z};\a)$ is a single-valued Sobolev distribution and $\pi$ is the
covering projection. This decomposition is not canonical: changing the representative within
the same cohomology class amounts to replacing
$\Omega_{\blambda^c}+\nu_{\z,\m,\blambda}$ by
$\Omega_{\blambda^c}+\nu_{\z,\m,\blambda}+ \d h$,
for some smooth $\a$-valued function $h$ on $\Sigma$ which is locally constant near
$\partial\Sigma$.

We denote by $\mathcal E^{\m}_{\Lambda}(\Sigma)$ the space of functionals $F$ defined
on $H^s_\Lambda(\widetilde\Sigma_{\z};\a)$ with the following property.
Fix a connected component of $H^s_\Lambda(\widetilde\Sigma_{\z};\a)$, equivalently a fixed
(relative) cohomology class determined by the monodromy data $(\blambda^c,\blambda,\m)$.
Choose a decomposition \eqref{eq:boundary-decomposition-u} for $u$ on this component.
We require that $F$ can be written as
\begin{equation}\label{eq:def-test-functional-boundary}
F(u)
=
P\!\big(\langle f,g_1\rangle,\ldots,\langle f,g_M\rangle\big)
G\left(
e^{\bi\big\langle q,I^{\bsigma}_{x_0}(\Omega_{\blambda^c})+
I^{\bxi}_{x_0}(\nu_{\z,\m,\blambda})\big\rangle}
\right),
\end{equation}
where $q\in\Lambda^*,$ $g_1,\ldots,g_M\in H^{-s}(\Sigma;\a)$, $P$ is a polynomial (allowed to depend on
$(\blambda^c,\blambda,\m)$), and $G$ is bounded and continuous on $C^0(\Sigma;\S^1)$
(again allowed to depend on $(\blambda^c,\blambda,\m)$). 

We stress that $\mathcal E^{\m}_{\Lambda}(\Sigma)$ does \emph{not} depend on the auxiliary
choices entering \eqref{eq:boundary-decomposition-u}. Indeed, changing the representative of
the cohomology class amounts to replacing
$\Omega_{\blambda^c}+\nu_{\z,\m,\blambda}$ by $\Omega_{\blambda^c}+\nu_{\z,\m,\blambda}+\d h$,
where $h$ is a smooth $\a$-valued function on $\Sigma$, locally constant near
$\partial\Sigma$. Correspondingly, the same $u$ admits a decomposition with the shifted
single-valued part $f' = f-h$. Since $h$ is smooth, the pairings satisfy
$\langle f',g_j\rangle=\langle f,g_j\rangle-\langle h,g_j\rangle$, hence the polynomial factor
in \eqref{eq:def-test-functional-boundary} can be absorbed into a (possibly different) polynomial
$P'$. Moreover, the primitive changes by the addition of $h-h(x_0)$, so the $\S^1$-valued
function inside $G$ is multiplied by the continuous $\S^1$-valued function
$\exp(i\langle q,\,h-h(x_0)\rangle)$; composing $G$ with this multiplication yields another
bounded continuous functional $G'$. Therefore the existence of a representation of the form
\eqref{eq:def-test-functional-boundary} is invariant under changing representatives, and the
definition of $\mathcal E^{\m}_{\Lambda}(\Sigma)$ is well posed.
\medskip

For $p\ge 1$, we define the seminorm on $\mathcal{E}^{\m}_\Lambda(\Sigma;\a)$
\begin{equation}\label{eq:boundary-seminorm-test-functionals}
\begin{aligned}
\|F\|_{\mathcal L^{\infty,p}_{\balpha,\m}}
:={}&
\sup_{\blambda^c\in\Lambda^{2\bbg+\bbb-1}}
\left(
\mathbb E\!\Bigg[
e^{
-\frac{1}{2\pi}\big\langle \d X_{g,D},\,
\Omega_{\blambda^c}^c+\nu_{\z,\m,\blambda}\big\rangle_2
-\frac{1}{4\pi}\big\|(1-\Pi_1^c)\Omega_{\blambda^c}\big\|_2^2}
\right.\\
&\hspace{3cm}\left.\times
\Big|F\!\big(X_{g,D}+P_\Sigma\widetilde{\bm\varphi}+u_{\z}^0+
I^{\bsigma}_{x_0}(\Omega^c_{\blambda^c})+I^{\bxi}_{x_0}(\nu_{\z,\m,\blambda})\big)\Big|^p
\Bigg]
\right)^{\!1/p},
\end{aligned}
\end{equation}
where $u_{z}^0:=\sum^{n_{\m}}_{j=1}i\alpha_jG_{g,D}(\cdot,z_j)$, and $\Pi_1^c$ is the $L^2$-orthogonal projection onto the space of harmonic $1$-forms with relative boundary
condition so that $(1-\Pi_1^c)\Omega_{\blambda^c}^c=\d f_{\blambda^c}$.
\begin{definition}[Amplitudes]\label{def:amplitude}$\,$
\begin{itemize}
    \item Let $\partial\Sigma=\emptyset$. Under the assumptions of Theorem \ref{thm:electro-magnetic operators}, for $F\in\mathcal{E}_\Lambda^{\m}(\Sigma;\a),$ we define the amplitude as \begin{equation}\label{eq:def-closed-amplitude}\A_{\Sigma,g,\bv,\balpha,\m}(F):=\la FV^g_{(\balpha,\m)}(\bv)\ra_{\Sigma,g}.\end{equation}
    \item Assume $\partial\Sigma\neq\emptyset$. For $F\in\mathcal E^{\m}_\Lambda(\Sigma;\balpha)$ and
boundary field $\widetilde{\bm\varphi}^{\blambda}\in (H^s_\Lambda(\mathbb R;\a))^{\bbb}$,
we define the amplitude $\mathcal A_{\Sigma,g,\bv,\balpha,\m,\bzeta}(F,\widetilde{\bm\varphi}^{\blambda})$ as 
\begin{equation}\label{eq:def-boundary-amplitude}
    \begin{aligned}
        &\mathcal A_{\Sigma,g,\bv,\balpha,\m,\bzeta}(F,\widetilde{\bm\varphi}^{\blambda})\\
        &:=\delta_0(\sum_{j=1}^{n_{\m}+\bbb}m_j(\blambda))\lim_{t\rightarrow1}\lim_{\eps\rightarrow0}\sum_{\blambda^c\in\Lambda^{2\bbg+\bbb-1}}e^{-\frac{1}{4\pi}\norm{\nu_{\z,\m,\blambda}+\Omega^c_{\blambda^c}}^2_{g,0}}Z_{\Sigma,g}\A^0_{\Sigma,g}(\widetilde{\bm\varphi})\\
        &\times\E\left[e^{-\frac{\la d\x_{g,D}+\d P\widetilde{\bm\varphi},\nu_{\z,\m,\blambda}+\Omega^c_{\blambda^c}\ra}{2\pi}}F(\Phi_g)\prod^{n_{\m}}_{j=1}V^g_{\alpha,\eps}(x_j(t))e^{-\frac{\bi}{4\pi}\la QK_g,\Phi_g\ra^{\reg}_{g}-\sum^r_{i=1}\mu_iM^g_{\gamma e_i}(\Phi_g,\Sigma)}\right]
    \end{aligned}
\end{equation}
where $\delta_0$ is the Dirac mass at $0,$ $x_j\in C^1([0,1],\Sigma)$ such that as $t\rightarrow1,$ $(x_j(t),\dot{x}_j(t))\rightarrow(z_j,v_j,)$ the Toda field is $\Phi_g=\x_{g,D}+P\widetilde{\bm\varphi}+I^{\bxi}_{x_0}(\nu_{\z,\m,\blambda})+I^{\bsigma}_{x_0}(\Omega^c_{\blambda^c})$, the expectation is over the $\a$-valued Dirichlet GFF $\x_{g,D}$, and $Z_{\Sigma,g}$ is the normalizing constant \begin{equation}\label{eq:amplitude-normalizing}
    Z_{\Sigma,g}:=\det(\Delta_{g,D})^{-\frac{r}{2}}.
\end{equation}
The regularized curvature term is 
\begin{equation}\label{eq:amplitude-curvature}
\begin{aligned}
    \la QK_g,\Phi_g\ra^{\reg}_{g}:=&\int_{\Sigma}\la QK_g,\x_{g,D}+P\widetilde{\bm\varphi}\ra\d\v_g+\int^{\reg}_{\Sigma}\la QK_g,I_{x_0}^{\bsigma}(\Omega^c_{\blambda^c})\ra\d\v_g\\
    &+\int^{\reg}_{\Sigma}\la QK_g,I_{x_0}^{\bxi}(\nu_{\z,\m,\blambda})\ra\d\v_g,
\end{aligned}
\end{equation}
and $\A^0_{\Sigma,g}(\widetilde{\bm\varphi})$ is the free field amplitude \begin{equation}\label{eq:amplitude-free-field}
\A^0_{\Sigma,g}(\widetilde{\bm\varphi})=e^{-\frac{1}{2}\la \widetilde{\bm\varphi},(\mathbf{D}_\Sigma-\mathbf{D})\widetilde{\bm\varphi}\ra}.
\end{equation} If $F=1$, we write $\A_{\Sigma,g,\bv,\balpha,\m,\bzeta}(\widetilde{\bm{\varphi}}^{\blambda}).$

\end{itemize}
\end{definition}
\begin{remark}
    $\,$
\begin{itemize}
\item As soon as $\partial\Sigma\neq\emptyset$, the primitive $I^{\bxi}_{x_0}(\nu_{\z,\m,\blambda})$
(and similarly $I^{\bsigma}_{x_0}(\Omega_{\blambda^c})$) depends on the choice of the base point $x_0$;
when needed we indicate this by writing $\mathcal A^{x_0}_{\Sigma,g,\bv,\balpha,\m,\z}$.
\item The definition of the amplitude does not depend on the choice of orientation of the arcs
$d_1,\ldots,d_{\bbb-1}$ entering the canonical relative basis. Indeed, reversing the
orientation of one arc $d_j$ replaces the corresponding dual representative
$\Omega_j^c$ by $-\Omega_j^c$, and therefore replaces the corresponding summation
variable $\lambda_j^c$ by $-\lambda_j^c$. Since the amplitude contains the full sum over
$\blambda^c\in\Lambda^{2\bbg+\bbb-1}$, this is only a reindexing of the lattice sum.
\end{itemize}
\end{remark}

We claim some properties of the amplitudes, whose proof is postponed to Section \ref{prop:proof of gluing}.
\begin{proposition}
\label{prop:CITT-boundary-amplitudes-properties}
Assume that $\partial\Sigma$ has $\bbb>0$ connected components. Under the standing
assumptions $\gamma^2<1$, $Q\in\Lambda^*$, and
$\alpha_j-Q\in\mathcal C_+,\, j=1,\ldots,n_{\mathfrak m},$
the amplitudes satisfy the following properties:

\begin{itemize}
\item The limit defining
$$
\mathcal A_{\Sigma,g,\bv,\balpha,\m,\bzeta}
(F,\widetilde{\bm\varphi}^{\blambda})
$$
exists for every $F\in\mathcal E^\m_\Lambda(\Sigma;\a)$, for
$\mu_0^{\otimes\bbb}$-almost every boundary field
$\widetilde{\bm\varphi}^{\blambda}$, and the resulting function of the boundary field
belongs to $\mathcal H^{\otimes\bbb}$.

\item The amplitude does not depend on the choice of the canonical relative homology
basis $\bsigma$, nor on the choice of compactly supported relative cohomology
representatives $(\Omega^c_j)_j$ dual to $\bsigma$.

\item The amplitude does not depend on the choice of the representative
$\nu_{\z,\m,\blambda}$ in its absolute cohomology class, provided it satisfies the conditions of Proposition \ref{prop:magnetic-background} and \eqref{nu-relative-period}
It is also independent of the auxiliary defect graph $\bxi$.

\item \textbf{Conformal anomaly.}
Let $g'$ and $g$ be two conformal admissible metrics on $\Sigma$, with
$g'=e^\rho g$, where $\rho\in C^\infty(\Sigma)$ and $\rho|_{\partial\Sigma}=0$.
Then
\begin{equation}\label{eq:boundary-amplitude-conformal-anomaly}
\begin{aligned}
\mathcal A_{\Sigma,g',\bv,\balpha,\m,\bzeta}
(F,\widetilde{\bm\varphi}^{\blambda})
=&e^{
\frac{c}{96\pi}
\int_\Sigma
\left(|d\rho|_g^2+2K_g\rho\right)\d\v_g
-
\sum_{j=1}^{n_{\mathfrak m}}
\Delta_{(\alpha_j,m_j)}\rho(z_j)
}
\\
&\times
\mathcal A_{\Sigma,g,\bv,\balpha,\m,\bzeta}
(
F(\cdot-\frac{\bi}{2}Q\rho),
\widetilde{\bm\varphi}^{\blambda}
),
\end{aligned}
\end{equation}
where
$$
c=r-6\langle Q,Q\rangle,
\qquad
\Delta_{(\alpha,m)}
=
\langle\frac{\alpha}{2},\frac{\alpha}{2}-Q\rangle
+\frac14|m|^2.
$$

\item \textbf{Diffeomorphism invariance.}
Let $\psi:\Sigma'\to\Sigma$ be an orientation-preserving diffeomorphism. Then
\begin{equation}\label{eq:boundary-amplitude-diffeomorphism}
\mathcal A^{x_0}_{\Sigma',\psi^*g,\bv,\balpha,\m,\bzeta}
(F,\widetilde{\bm\varphi}^{\blambda})
=
\mathcal A^{\psi(x_0)}_{\Sigma,g,\psi_*\bv,\balpha,\m,\psi\circ\bzeta}
(F_\psi,\widetilde{\bm\varphi}^{\blambda}),
\end{equation}
where $F_\psi(\phi):=F(\phi\circ\psi).$

\item \textbf{Spins.}
For $r_{\btheta}\bv
:=
((z_1,r_{\theta_1}v_1),\ldots,
(z_{n_{\mathfrak m}},r_{\theta_{n_{\mathfrak m}}}v_{n_{\mathfrak m}})),$
one has
\begin{equation}\label{eq:boundary-amplitude-spin}
\mathcal A_{\Sigma,g,r_{\btheta}\bv,\balpha,\m,\bzeta}
(F,\widetilde{\bm\varphi}^{\blambda})
=
e^{
i\sum_{j=1}^{n_{\mathfrak m}}
\langle \alpha_j-Q,m_j\rangle\theta_j
}\mathcal A_{\Sigma,g,\bv,\balpha,\m,\bzeta}
(F,\widetilde{\bm\varphi}^{\blambda}).
\end{equation}
\end{itemize}
\end{proposition}

We now state the gluing identities for the amplitudes. We first recall the geometric
setup. Let $(\Sigma_i,g_i,\z_i,\bzeta_i),\,i=1,2,$ be two compact connected Riemann surfaces with admissible metrics and parametrized
boundary. We write $\partial\Sigma_i=\bigsqcup_{j=1}^{\bbb_i}\partial_j\Sigma_i,\,$
$\bzeta_i=(\zeta_{i,1},\ldots,\zeta_{i,\bbb_i}),$
where each $\zeta_{i,j}:\mathbb T\to\partial_j\Sigma_i$ is an analytic
parametrization. We assume that the last boundary component of $\Sigma_1$ and that of $\Sigma_2$ are glued together by the parametrizations, with
opposite orientations, that $\partial_{\bbb-1}\Sigma_1$ is outgoing, that $x_0^i\in\partial_{\bbb_i}\Sigma_i$, and that $x_0^1=x_0^2$ on the glued surface. The resulting surface is denoted by $\Sigma=\Sigma_1\#\Sigma_2,$ whose remaining boundary parametrization is
$\bzeta
=
(\zeta_{1,1},\ldots,\zeta_{1,\bbb_1-1},
\zeta_{2,1},\ldots,\zeta_{2,\bbb_2-1}).$

The metric $g$ on $\Sigma$ is the metric obtained by gluing $g_1$ and $g_2$, and is
again assumed admissible. The marked points, tangent vectors, electric charges, and magnetic
charges are the union of the data on the two pieces:
$$\z=(\z_1,\z_2),\quad
\bv=(\bv_1,\bv_2),\quad
\balpha=(\balpha_1,\balpha_2),\quad
\m=(\m_1,\m_2).
$$
We also take test functionals $F_i\in\mathcal E^{\m_i}_{\Lambda}(\Sigma_i;\a)$, and write
$F_1\otimes F_2(\Phi_g)=F(\Phi_g|_{\Sigma_1})F(\Phi_g|_{\Sigma_2})$ for the corresponding functional on the glued surface, evaluated on the
two restrictions of the field. We write the boundary data on the glued surface $\Sigma$ as $(\widetilde{\bm\varphi}^{\blambda_1}_1,\widetilde{\bm\varphi}^{\blambda_2}_2)\in (H^s_\Lambda(\R;\a))^{\bbb_1-1}\times(H^s_\Lambda(\R;\a))^{\bbb_2-1}$ 
so that boundary data on $\Sigma_i$ before gluing takes the form $(\widetilde{\bm\varphi}_{1}^{\blambda_1},\widetilde{\varphi}^\lambda)$.

We will also need the analogous self-gluing setup. Let
$(\Sigma,g,\bzeta)$ be a compact connected Riemann surface with admissible metric and
$\bbb\ge2$ parametrized boundary components. We assume that the last two boundary
components are glued together, with opposite orientations and compatible parametrizations,
and denote the resulting surface by $\Sigma^\#,$
whose remaining boundary parametrizations are denoted by $\bzeta_\#$. The boundary data on the glued surface $\Sigma^\sharp$ will be written as  $\widetilde{\bm\varphi}^{\blambda}
=
(\widetilde\varphi^{\lambda_1},\ldots,
\widetilde\varphi^{\lambda_{\bbb-2}})\in(H^s_{\Lambda}(\R;\a))^{\bbb-2}$, and that on the unglued surface $\Sigma$ will be written as $(\widetilde{\bm\varphi}^{\blambda},\widetilde{\varphi}^{\lambda_{\bbb-1}},\widetilde{\varphi}^{\lambda_{\bbb}})\in (H^s_{\Lambda}(\R;\a))^{\bbb-2}\times H^s_{\Lambda}(\R;\a)\times H^s_{\Lambda}(\R;\a)$, where $\widetilde{\varphi}^{\lambda_{\bbb-1}},\,\widetilde{\varphi}^{\lambda_{\bbb-1}}$ denote the boundary data on $\partial_{\bbb-1}\Sigma,\,\partial_{\bbb}\Sigma,$ respectively.

\begin{proposition}\label{prop:CITT-Segal-gluing}
With the above notation, the following hold true:
\begin{itemize}
\item \textbf{Gluing two surfaces.} 
For functionals $F_i\in\mathcal{E}^{\m_i}_{\Lambda}(\Sigma_i;\a),\,i=1,2$, we have
\begin{equation}\label{eq:CITT-two-surface-gluing}
\begin{aligned}
\mathcal A_{\Sigma,g,\bv,\balpha,\m,\bzeta}
\left(
F_1\otimes F_2,
\widetilde{\bm\varphi}_{1}^{\blambda_1},
\widetilde{\bm\varphi}_{2}^{\blambda_2}
\right)
=C_r
&\int_{H^s_\Lambda(\mathbb R;\a)/(2\pi\Lambda)}
\mathcal A_{\Sigma_1,g_1,\bv_1,\balpha_1,\m_1,\bzeta_1}
\left(
F_1,
\widetilde{\bm\varphi}_{1}^{\blambda_1},
\widetilde\varphi^\lambda
\right)\\&\times
\mathcal A_{\Sigma_2,g_2,\bv_2,\balpha_2,\m_2,\bzeta_2}
\left(
F_2,
\widetilde{\bm\varphi}_{2}^{\blambda_2},
\widetilde\varphi^\lambda
\right)
\,\d\mu_0(\widetilde\varphi^\lambda),
\end{aligned}
\end{equation}
with 
$$
C_r:=
\begin{cases}
(\sqrt 2\,\pi)^{-r}, & \partial\Sigma\neq\emptyset,\\
2^{r/2}, & \partial\Sigma=\emptyset.
\end{cases}
$$

\item \textbf{Self-gluing.} For a functional $F\in\mathcal{E}^{\m}_{\Lambda}(\Sigma;\a)$, we have
\begin{equation}\label{eq:CITT-self-gluing}
\begin{aligned}
\mathcal A_{\Sigma^\#,g^\#,\bv,\balpha,\m,\bzeta_\#}
\left(
F,
\widetilde{\bm\varphi}^{\blambda}_{\#}
\right)
=
C_r
\int_{H^s_\Lambda(\mathbb R;\a)/(2\pi\Lambda)}
\mathcal A_{\Sigma,g,\bv,\balpha,\m,\bzeta}
\left(
F,
\widetilde{\bm\varphi}^{\blambda}_{\#},
\widetilde\varphi^\lambda,
\widetilde\varphi^\lambda
\right)
\,\d\mu_0(\widetilde\varphi^\lambda),
\end{aligned}
\end{equation}
with $$
C_r:=
\begin{cases}
(\sqrt 2\,\pi)^{-r}, & \partial\Sigma^\#\neq\emptyset,\\
2^{r/2}, & \partial\Sigma^\#=\emptyset.
\end{cases}
$$

\end{itemize}
\end{proposition}

\subsection{Proof of Propositions \ref{prop:CITT-boundary-amplitudes-properties} and \ref{prop:CITT-Segal-gluing}}\label{prop:proof of gluing}

We start with the case of free compactified boson, namely we restrict to the case where $\balpha=0,\, \bmu=0,\, Q=0.$
In particular, there is no curvature phase and no Toda interaction. The only remaining nontrivial terms are the compactification, the topological sectors, and the magnetic
background $\nu_{\z,\m,\blambda}$. 

For a positive measurable functional $F$, or more generally for an integrable measurable
functional for which the expression below is finite, we set
\begin{equation}\label{eq:free-compactified-boson-amplitude}
\begin{aligned}
\mathcal A^0_{\Sigma,g,\z,\m,\bzeta}
\left(F,\widetilde{\bm\varphi}^{\blambda}\right)
:=&
\delta_0\!\left(\sum_{j=1}^{n_{\mathfrak m}+\bbb}m_j(\blambda)\right)
\sum_{\blambda^c\in\Lambda^{2\bbg+\bbb-1}}
e^{-\frac{1}{4\pi}
\left\|\nu_{\z,\m,\blambda}+\Omega^c_{\blambda^c}\right\|_{g,0}^2}
Z_{\Sigma,g}\,
\mathcal A^0_{\Sigma,g}(\widetilde{\bm\varphi})
\\
&\times
\mathbb E\left[
e^{
-\frac{1}{2\pi}
\left\langle
\d \x_{g,D}+\d P\widetilde{\bm\varphi},
\nu_{\z,\m,\blambda}+\Omega^c_{\blambda^c}
\right\rangle_2
}F(\Phi_g^0)
\right],
\end{aligned}
\end{equation}
where
$$
\Phi_g^0
:=
\x_{g,D}
+
P\widetilde{\bm\varphi}
+
I^{\bxi}_{x_0}(\nu_{\z,\m,\blambda})
+
I^{\bsigma}_{x_0}(\Omega^c_{\blambda^c}).
$$
Here $F$ is assumed to be $2\pi\Lambda$-periodic: $F(\Phi+2\pi\lambda)=F(\Phi),
\, \lambda\in\Lambda.$
The amplitude $\mathcal A^0_{\Sigma,g,\z,\m,\bzeta}$ has no dependence on the tangent
vectors $\bv$, except through the auxiliary choice of defect graph used to write the
primitive; by the defect-graph independence established above, the resulting amplitude is
independent of that choice.

We now record the free gluing identities. They are the compactified-boson analogues of
Proposition~\ref{prop:CITT-Segal-gluing}, and they are the basic input for the interacting
Toda gluing identities.

\begin{proposition}
\label{prop:CITT-free-boson-two-surface-gluing}
In the setup of Section~\ref{subsec:amplitudes}, the following hold true: 
\begin{itemize}
    \item \textbf{Gluing two surfaces.} For $F_i$ $2\pi\Lambda$-periodic positive measurable, or $2\pi\Lambda$-periodic integrable, functionals on $\Sigma_i$, $i=1,2$, we have
\begin{equation}\label{eq:CITT-free-boson-two-surface-gluing}
\begin{aligned}
\mathcal A^0_{\Sigma,g,\z,\m,\bzeta}
\left(
F_1\otimes F_2,
\widetilde{\bm\varphi}_{1}^{\blambda_1},
\widetilde{\bm\varphi}_{2}^{\blambda_2}
\right)
=&
C_r
\int_{H^s_\Lambda(\mathbb R;\a)/(2\pi\Lambda)}
\mathcal A^0_{\Sigma_1,g_1,\z_1,\m_1,\bzeta_1}
\left(
F_1,
\widetilde{\bm\varphi}_{1}^{\blambda_1},
\widetilde\varphi^\lambda
\right)
\\
&\times
\mathcal A^0_{\Sigma_2,g_2,\z_2,\m_2,\bzeta_2}
\left(
F_2,
\widetilde{\bm\varphi}_{2}^{\blambda_2},
\widetilde\varphi^\lambda
\right)
\,\d\mu_0(\widetilde\varphi^\lambda),
\end{aligned}
\end{equation}
with the constant
$$C_r=\begin{cases}
(\sqrt 2\,\pi)^{-r}, & \partial\Sigma\neq\emptyset,\\
2^{r/2}, & \partial\Sigma=\emptyset.
\end{cases}
$$
\item \textbf{Self-gluing.} For $F$ a
$2\pi\Lambda$-periodic positive measurable, or $2\pi\Lambda$-periodic integrable functional on $\Sigma$, we have
\begin{equation}\label{eq:CITT-free-boson-self-gluing}
\begin{aligned}
\mathcal A^0_{\Sigma^\#,g^\#,\z,\m,\bzeta_\#}
\left(F,\widetilde{\bm\varphi}^{\blambda}\right)=
C_r\int_{H^s_\Lambda(\mathbb R;\a)/(2\pi\Lambda)}
\mathcal A^0_{\Sigma,g,\z,\m,\bzeta}
\left(F,\widetilde{\bm\varphi}^{\blambda},\widetilde\varphi^\lambda,\widetilde\varphi^\lambda\right)\,\d\mu_0(\widetilde\varphi^\lambda),
\end{aligned}
\end{equation}
with the constant
$$
C_r=
\begin{cases}
(\sqrt 2\,\pi)^{-r}, & \partial\Sigma^\#\neq\emptyset,\\
2^{r/2}, & \partial\Sigma^\#=\emptyset.
\end{cases}
$$

\end{itemize}

\end{proposition}
\begin{proof}The proof is a higher-rank adaptation of the gluing argument in \cite{Guillarmou:2023exh}. We give the details for the gluing of two surfaces, indicating only the modifications required in the present setting. The self-gluing case follows from the argument of \cite[Proposition~8.14]{Guillarmou:2023exh} after making the corresponding higher-rank replacements.

We separate into two cases: $\partial\Sigma=\emptyset$ and $\partial\Sigma\neq\emptyset$.
\paragraph{\textbf{Case 1: $\partial\Sigma=\emptyset$.}}
Since the glued surface $\Sigma$ is closed, each
piece $\Sigma_i$ has a single boundary component, and if the amplitudes on $\Sigma_1$ and $\Sigma_2$ are
non-zero, the boundary magnetic neutrality conditions force
$\lambda_1=
\sum_{j=1}^{n_{\mathfrak m}^1}m_{1j}$ and
$\lambda_2=-\sum_{j=1}^{n_{\mathfrak m}^2}m_{2j}.$
Moreover, the magnetic charge neutrality on $\Sigma$ reads
$\sum_{j=1}^{n_{\mathfrak m,1}}m_{1,j}
+
\sum_{j=1}^{n_{\mathfrak m,2}}m_{2,j}
=0,$ and so $\lambda=\lambda_1=\lambda_2=
\sum_{j=1}^{n_{\mathfrak m}^1}m_{1j}.$

Next we consider the $1$-forms $\nu^1_{\z_1,\m_1,\lambda},
\,
\nu^2_{\z_2,\m_2,\lambda}$
on $\Sigma_1$ and $\Sigma_2$, respectively. By the construction of
Proposition~\ref{prop:magnetic-background}, in the chart
$\omega_i:U_i\to\mathbb A_\delta$ near the boundary component of $\Sigma_i$, these forms
may be chosen so that, close to the boundary,
$\nu^i_{\z_i,\m_i,\lambda}
=
-\varsigma_i\,\lambda\d\theta .$
Since the two boundary parametrizations are
identified with opposite orientations, the two expressions match under gluing.
Equivalently, $\nu^2_{\z_2,\m_2,\lambda}$ gives a smooth continuation of
$\nu^1_{\z_1,\m_1,\lambda}$ across the glued circle. Hence the two forms glue to a closed
$\a$-valued $1$-form $\nu_{\z,\m}$
on $\Sigma_{\z}$, with prescribed winding $m_{i,j}$ around each marked point
$z_{i,j}$, $i=1,2$. This is precisely the magnetic background form entering the
closed compactified-boson amplitude on $\Sigma$.

The defect graphs are chosen compatibly with this gluing (see Lemma \ref{lem:boundary-curvature-additivity-gluing}). Namely, on each $\Sigma_i$ we
choose a defect graph $\mathcal D_{\bv_i,\bxi_i}$ such that exactly one arc has the
distinguished boundary point $\zeta_i(1)$ as an endpoint. After identifying the two
boundary components, these two boundary-ending arcs concatenate to a single smooth arc
on $\Sigma$. Keeping all the other arcs unchanged gives a defect graph
$\mathcal D_{\bv,\bxi}
\subset \Sigma$ associated with the glued magnetic data $(\bv,\m)$.

For the topological sectors, we observe that $\bsigma:=\bsigma_1\cup\bsigma_2$ forms a basis of $H_1(\Sigma)$ since $\bbb_1=\bbb_2=1,$ and the glued curve $\mathcal{C}$ is homologically trivial. This shows that, after choosing
compactly supported representatives $\Omega^{1,c}_{\blambda_1^c}$ and
$\Omega^{2,c}_{\blambda_2^c}$ in the interiors of the two pieces, their sum defines a
smooth closed representative $\Omega^c_{\blambda^c}=\Omega^{1,c}_{\blambda_1^c}+\Omega^{2,c}_{\blambda_2^c}$ on $\Sigma$. 

On $\Sigma$, the compactified-boson amplitude can therefore be written, for every
positive or integrable $2\pi\Lambda$-periodic functional $F$, as
\begin{equation}\label{eq:closed-free-boson-amplitude-glued}
\begin{aligned}
\mathcal A^0_{\Sigma,g,\z,\m}(F)
:=
\left(
\frac{\operatorname{vol}_g(\Sigma)}
{\det{}'(\Delta_g)}
\right)^{\frac r2}
\sum_{\blambda^c\in\Lambda^{2\bbg}}
e^{-\frac{1}{4\pi}
\|\Omega^c_{\blambda^c}+\nu_{\z,\m}\|_{g,0}^2}
\int_{\T(\gamma)}
\mathbb E\left[
e^{-\frac{1}{2\pi}
\langle \d \x_g,\Omega^c_{\blambda^c}+\nu_{\z,\m}\rangle_2}
F(\Phi_g)
\right]\d c,
\end{aligned}
\end{equation}
where $\x_g$ is the $\a$-valued GFF with zero mean on the closed surface,
and
$$
\Phi_g
=
c+\x_g+
I^{\bsigma}_{x_0}(\Omega^c_{\blambda^c})
+
I^{\bxi}_{x_0}(\nu_{\z,\m}).
$$

We now express the closed GFF by cutting $\Sigma$ along the gluing curve. Let
$\mathcal C\subset\Sigma$ be the common boundary circle. The Markov decomposition (Proposition \ref{prop:Markov-property}) gives
$$
\x_g
\overset{\mathrm{law}}{=}
\x_{1}+\x_{2}+P\x_{\mathcal{C}}-c_g,
$$
where $\x_{i}$ are independent $\a$-valued Dirichlet GFFs on $\Sigma_i$,
$\x_{\mathcal{C}}$ is the restriction of $\x$ to the glued curve $\mathcal C$, and $c_g:=\frac{1}{\v_g(\Sigma)}\int_{\Sigma}(\x_1+\x_2+P\x_{\mathcal{C}})\d\v_g$. Since the zero mode $c$ is integrated over
$\T(\gamma)$, this correction may be absorbed into $c$. 
Therefore, plugging this decomposition into
\eqref{eq:closed-free-boson-amplitude-glued}, and shifting the zero mode by
$c_g$, we obtain
\begin{equation}\label{eq:CITT-closed-gluing-after-Markov}
\begin{aligned}
\mathcal A^0_{\Sigma,g,\z,\m}(F_1\otimes F_2)
={}&
\left(
\frac{\operatorname{vol}_g(\Sigma)}
{\det{}'(\Delta_g)}
\right)^{\frac r2}
\sum_{\blambda^c}
e^{-\frac{1}{4\pi}
\|\Omega^c_{\blambda^c}+\nu_{\z,\m}\|_{g,0}^2}
\\ &\times
\int_{\T(\gamma)}
\int\E\left[\mathcal{B}
_1(c,\x_{\mathcal{C}},\Omega^{1,c}_{\blambda^c_1}+\nu^1_{\z_1,\m_1,\lambda_1})\mathcal{B}
_2(c,\x_{\mathcal{C}},\Omega^{2,c}_{\blambda^c_2}+\nu^2_{\z_2,\m_2,\lambda_2})\right]\d c
\end{aligned}
\end{equation}
where $\mathcal{B} _i(c,\varphi,\Omega):=\E[e^{-\frac{1}{2\pi}\la\d\x_i+\d P\varphi,\Omega\ra}F_i(\Phi_i)]$ with $\Phi_i:=c+\x_i+P\varphi+I^i_{x_0}(\Omega)$, and the expectation is over the GFF $\x_i.$ Here $I^i(\Omega^{i,c}_{\blambda_i}+\nu^i_{\z_i,\m_i,\lambda_i}):=I^{\bsigma_i}_{x_0}(\Omega^{i,c}_{\blambda^c_i})+I_{x_0}^{\bxi_i}(\nu^i_{\z_i,\m_i,\lambda_i})$.

We now make a further shift in the zero mode $c$, subtracting the average of the boundary trace $m_{\mathcal C}(\x_{\mathcal{C}})
:=
\frac{1}{2\pi}\int_0^{2\pi}\x_{\mathcal{C}}(e^{\bi\theta})\,\d\theta
$
This allows us to replace the law $\P_{\x_{\mathcal{C}}}$ of the boundary trace
$\x_{\mathcal{C}}$ by the law $\P_{\x_{\mathcal{C}}-m_{\mathcal C}(\x_{\mathcal{C}})}$ of its recentered part $\x_{\mathcal{C}}-m_{\mathcal C}(\x_{\mathcal{C}}).$

By the $\a$-valued version of \cite[(5.14)]{guillarmou2021segal}, applied componentwise
in the orthonormal basis $(\eps_\ell)_{\ell=1}^r$, the desired amplitude becomes
\begin{equation}\label{eq:CITT-recentered-trace-law}
\begin{aligned}
\A^0_{\Sigma,g,\z,\m}(F)=&
2^{\frac r2}Z_{\Sigma_1,g_1}Z_{\Sigma_2,g_2}
\sum_{\blambda^c\in\Lambda^{2\bbg}}
e^{-\frac{1}{4\pi}
\|\Omega^c_{\blambda^c}+\nu_{\z,\m}\|_{g,0}^2}
\\&\times
\int_{\T(\gamma)}\int
\mathcal{B}
_1(c,\varphi,\Omega^{1,c}_{\blambda^c_1}
+\nu^1_{\z_1,\m_1,\lambda_1})
\mathcal{B}
_2(c,\varphi,\Omega^{2,c}_{\blambda^c_2}
+\nu^2_{\z_2,\m_2,\lambda_2})
\\
&\hspace{6cm}
\times
e^{-\frac12
\langle \varphi,
\widetilde{\mathbf D}_{\Sigma,\mathcal C}\varphi
\rangle_2}
\,\d\mathbb P_{\mathbb T}(\varphi)\d c.
\end{aligned}
\end{equation}
We observe that $e^{-\frac12
\langle \varphi,
\widetilde{\mathbf D}_{\Sigma,\mathcal C}\varphi
\rangle_2}
=
\mathcal A^0_{\Sigma_1,g_1}(c+\varphi)\,
\mathcal A^0_{\Sigma_2,g_2}(c+\varphi).$
Combining this, \eqref{eq:CITT-recentered-trace-law}, the additivity
$$
\|\nu_{\z,\m}+\Omega^c_{\blambda^c}\|_{g,0}^2
=
\|\nu^1_{\z_1,\m_1,\lambda}+\Omega^{1,c}_{\blambda_1^c}\|_{g_1,0}^2
+
\|\nu^2_{\z_2,\m_2,\lambda}+\Omega^{2,c}_{\blambda_2^c}\|_{g_2,0}^2,
$$ and the fact that  $\lambda=\sum^{n_{\m,1}}_{j=1}m_{1,j},$
we get
\begin{equation}\label{eq:CITT-closed-case-final-gluing}
\begin{aligned}
\mathcal A^0_{\Sigma,g,\z,\m}(F_1\otimes F_2)
=
2^{\frac r2}
\int_{H^s_\Lambda(\mathbb R;\a)/(2\pi\Lambda)}
&
\mathcal A^0_{\Sigma_1,g_1,\z_1,\m_1,\bzeta_1}
\left(
F_1,\widetilde\varphi^\lambda
\right)
\\
&\times
\mathcal A^0_{\Sigma_2,g_2,\z_2,\m_2,\bzeta_2}
\left(
F_2,\widetilde\varphi^\lambda
\right)
\,\d\mu_0(\widetilde\varphi^\lambda).
\end{aligned}
\end{equation}
This proves the desired gluing identity in the case where the glued surface
$\Sigma$ is closed. 
\medskip

\paragraph{\textbf{Case 2: $\partial\Sigma\neq\emptyset$.}}

We now consider the case where the glued surface has non-empty boundary, so that we may assume without loss of generality that $\bbb_1\geq 2$.
We observe that if $\sum_{j=1}^{n_{\mathfrak m,1}}m_{1,j}
+
\sum_{j=1}^{n_{\mathfrak m,2}}m_{2,j}
+
\sum_{\ell=1}^{\bbb_1-1}\varsigma_{1,\ell}\lambda_{1,\ell}
+
\sum_{\ell=1}^{\bbb_2-1}\varsigma_{2,\ell}\lambda_{2,\ell}
\neq0,
$ then both sides of \eqref{eq:CITT-free-boson-two-surface-gluing} vanish and the equality holds. If $\sum_{j=1}^{n_{\mathfrak m,1}}m_{1,j}
+
\sum_{j=1}^{n_{\mathfrak m,2}}m_{2,j}
+
\sum_{\ell=1}^{\bbb_1-1}\varsigma_{1,\ell}\lambda_{1,\ell}
+
\sum_{\ell=1}^{\bbb_2-1}\varsigma_{2,\ell}\lambda_{2,\ell}
=0,
$ the neutrality condition on $\Sigma_i,$ forces $$\lambda=\lambda_{1\bbb_1}=\lambda_{\bbb_2}=
\sum_{j=1}^{n_{\mathfrak m,1}}m_{1,j}+
\sum_{\ell=1}^{\bbb_1-1}\varsigma_{1,\ell}\lambda_{1,\ell}.
$$
Hence, in the integral against $\mu_0$, all winding sectors vanish except this single
compatible value of $\lambda$.

We next compare the magnetic backgrounds. By Proposition~\ref{prop:magnetic-background},
we may choose $\nu^i_{\z_i,\m_i,(\blambda_i,\lambda)}$ on $\Sigma_i$
so that near the glued boundary component, $\nu^i_{\z_i,\m_i,(\blambda_i,\lambda)}
=
-\varsigma_{i,\bbb_i}\lambda\,\d\theta .$
Because the two boundary parametrizations are glued with opposite orientation, these two
local expressions agree after gluing. Thus the two forms glue to a closed
$\a$-valued form $\nu_{\z,\m,(\blambda_1,\blambda_2)}=\nu^1_{\z_1,\m_1,(\blambda_1,\lambda_{1\bbb_1})}\bone_{\Sigma_1}+\nu^2_{\z_2,\m_2,(\blambda_2,\lambda_{2\bbb_2})}\bone_{\Sigma_2}$
on $\Sigma$, with the prescribed singularities at all interior magnetic insertions and
with the prescribed boundary winding data on the remaining boundary components.

We choose the defect graphs compatibly, exactly as in the closed case: on each
$\Sigma_i$, exactly one arc ends at the distinguished point of the glued boundary
component, and after gluing these two arcs concatenate to one smooth arc. Keeping all the
other arcs unchanged gives a defect graph $\mathcal D_{\bv,\bxi}$ on the glued surface.

We separate into two subcases: $\bbb_2\geq2$ and $\bbb_2=1$. We begin with the case where $\bbb_2\geq 2.$ By \cite[Lemma 3.7]{Guillarmou:2023exh}, we have a basis of the relative homology $H_1(\Sigma,\partial\Sigma)$ $$\bsigma=\bsigma_1\#\bsigma_2$$ by collecting the curves for $i=1,\,2$: $a_{ij},b_{ij},$ $j=1,\ldots,2\bbg_i,$ $d_{ij},$ $j=1,\ldots,\bbb_i-1$, and the curve $d_{1(\bbb_1-1)}-d_{2(\bbb_2-1)}$. This gives rise to a basis of $H^1(\Sigma,\partial\Sigma)$ $$\eta^{1,c}_1,\ldots,\eta^{1,c}_{2\bbg_1+\bbb_1-1},\eta^{2,c}_1,\ldots,\eta^{2,c}_{2\bbg_1+\bbb_1-2},$$ consisting of smooth real-valued closed 1-forms, dual to $\bsigma$ with compact support.
We will denote the corresponding $2\pi\Lambda$-periodic $\a$-valued 1-form by $\Omega^c_{\blambda^c}:=\Omega^{1,c}_{\blambda_1^c}+\Omega^{2,c}_{\blambda^c_2}$ for $(\blambda^c_1,\blambda^c_2)\in\Lambda^{2\bbg+\bbb_1-1}\times\Lambda^{2\bbg+\bbb_2-2}$.

The left-hand side of \eqref{eq:CITT-free-boson-two-surface-gluing} takes the form 
$$
\begin{aligned}
&\mathcal A^0_{\Sigma,g,\z,\m,\bzeta}
\left(
F_1\otimes F_2,
\widetilde{\bm\varphi}^{\blambda_1}_1,
\widetilde{\bm\varphi}^{\blambda_2}_2
\right)
\\&=
Z_{\Sigma,g}
\mathcal A^0_{\Sigma,g}
\left(
\widetilde{\bm\varphi}_1,\widetilde{\bm\varphi}_2
\right)
\sum_{\blambda^c\in\Lambda^{2\bbg+\bbb_1+\bbb_2-3}}
e^{-\frac{1}{4\pi}
\left\|
\nu_{\z,\m,\blambda}
+
\Omega^c_{\blambda^c}
\right\|_{g,0}^2}
\mathcal{B}_{\Sigma,g}(F_1\otimes F_2,\widetilde{\bm{\varphi}}_1,\widetilde{\bm{\varphi}}_2,\blambda,\blambda^c),
\end{aligned}
$$
with 
$$\mathcal{B}_{\Sigma,g}(F_1\otimes F_2,\widetilde{\bm{\varphi}}_1,\widetilde{\bm{\varphi}}_2,\blambda,\blambda^c):=
\mathbb E\left[
e^{
-\frac{1}{2\pi}
\left\langle
\d \x_{g,D}
+
\d P\widetilde{\bm\varphi},
\nu_{\z,\m,\blambda}
+
\Omega^c_{\blambda^c}
\right\rangle_2
}F_1\otimes F_2(\Phi_g)
\right],$$
where the Toda field is $\Phi_g
=
\x_{g,D}
+
P\widetilde{\bm\varphi}
+
I^{\bsigma}_{x_0}(\Omega^c_{\blambda^c})
+
I^{\bxi}_{x_0}
\left(
\nu_{\z,\m,\blambda}
\right),$ and the expectation is over the $\a$-valued Dirichlet GFF $\x_{g,D}.$

The Markov property of the Dirichlet GFF (Proposition \ref{prop:Markov-property}) gives the decomposition
$$\x_{g,D}\overset{\mathrm{law}}{=}\x_1+\x_2+P\mathbf{Y},$$
where, $\mathbf{Y}$ is the restriction of $\x_{g,D}$ to the glued component $\mathcal{C}$, and the
fields $\x_{1}$ and $\x_{2}$ are independent Dirichlet GFFs on $\Sigma_1$ and $\Sigma_2$. Denoting by $h_{\mathcal{C}}$ the restriction of the harmonic function $P(\widetilde{\bm\varphi}_1,\widetilde{\bm\varphi}_2)$ to $\mathcal{C}$, Lemma \ref{lem:boundary-curvature-additivity-gluing} implies that
\begin{align*}
    &\mathcal{B}_{\Sigma,g}(F_1\otimes F_2,\widetilde{\bm{\varphi}}_1,\widetilde{\bm{\varphi}}_2,\blambda,\blambda^c)\\
    &=\int \mathcal{B}_{\Sigma_1,g_1}(F_1,\widetilde{\bm{\varphi}}_1,\widetilde{{\varphi}}+h_{\mathcal{C}},\blambda_1,\lambda_{1\bbb_1},\blambda_1^c)\mathcal{B}_{\Sigma_2,g_2}(F_2,\widetilde{\bm{\varphi}}_2,\widetilde{{\varphi}}+h_{\mathcal{C}},\blambda_2,\lambda_{2\bbb_2},\blambda_2^c)\d\P_{\mathbf{Y}}(\widetilde\varphi)
\end{align*}
where $\P_{\mathbf{Y}}$ is the law of $\mathbf{Y}$, and 
$$
\begin{aligned}\mathcal{B}_{\Sigma_i,g_i}(F_i,\widetilde{\bm{\varphi}}_i,\widetilde{{\varphi}},\blambda_i,\lambda_{i\bbb_i},\blambda_i^c)
:=
\mathbb E_i[
&e^{
-\frac{1}{2\pi}
\langle
\d \x_{i}
+
\d P
\left(
\widetilde{\bm\varphi}_i,
\widetilde\varphi\right),
\nu^i_{\z_i,\m_i,(\blambda_i,\lambda_{i\bbb_i})}
+
\Omega^{i,c}_{\blambda^c_i}
\rangle_2}
F_i
(\Phi_i
)
],
\end{aligned}
$$ with the expectation $\E_i$ with respect to $\x_i$ and $$\Phi_i=
\x_{i}
+
P
\left(
\widetilde{\bm\varphi}_i,
\widetilde\varphi
\right)
+
I^{\bsigma_i}_{x_0}(\Omega^{i,c}_{\blambda^c_i})
+
I^{\bxi_i}_{x_0^i}
(
\nu^i_{\z_i,\m_i,(\blambda_i,\lambda)}).$$
Applying \cite[Lemmas 5.3 \& 5.4]{guillarmou2021segal}
componentwise in the orthonormal basis $(\eps_\ell)_{\ell=1}^r$, the law of the trace $\mathbf Y+h_{\mathcal{C}}$ can be written as
$$
\d\P_{\mathbf Y+h_{\mathcal{C}}}(\widetilde\varphi)
=
C_r\,
\frac{
Z_{\Sigma_1,g_1}Z_{\Sigma_2,g_2}
}{
Z_{\Sigma,g}
}
\frac{
\mathcal A^0_{\Sigma_1,g_1}
(\widetilde{\bm\varphi}_1,c+\varphi)
\mathcal A^0_{\Sigma_2,g_2}
(\widetilde{\bm\varphi}_2,c+\varphi)
}{
\mathcal A^0_{\Sigma,g}
(\widetilde{\bm\varphi}_1,\widetilde{\bm\varphi}_2)
}
\,\d c\otimes\d\P_{\T}(\varphi),
$$
where in the case $\partial\Sigma\neq\emptyset$ the constant is
$C_r=(\sqrt 2\,\pi)^{-r}.$ This implies that$$
\begin{aligned}
&Z_{\Sigma,g}\mathcal A^0_{\Sigma,g}
(\widetilde{\bm\varphi}_1,\widetilde{\bm\varphi}_2)
\,\d\P_{\mathbf Y}(\widetilde\varphi)
\\
&\qquad =
C_r\,
Z_{\Sigma_1,g_1}\mathcal A^0_{\Sigma_1,g_1}
(\widetilde{\bm\varphi}_1,\widetilde\varphi+h_{\mathcal C})
\,
Z_{\Sigma_2,g_2}\mathcal A^0_{\Sigma_2,g_2}
(\widetilde{\bm\varphi}_2,\widetilde\varphi+h_{\mathcal C})
\,\d\mu_0(\widetilde\varphi^\lambda).
\end{aligned}
$$
Substituting this identity into the preceding expression gives
$$
\begin{aligned}
&\mathcal A^0_{\Sigma,g,\z,\m,\bzeta}
\left(
F_1\otimes F_2,
\widetilde{\bm\varphi}^{\blambda_1}_1,
\widetilde{\bm\varphi}^{\blambda_2}_2
\right)
\\
&=
C_r
\sum_{\blambda^c\in\Lambda^{2\bbg+\bbb-1}}
e^{-\frac{1}{4\pi}
\left\|
\nu_{\z,\m,\blambda}
+
\Omega^c_{\blambda^c}
\right\|_{g,0}^2}
Z_{\Sigma_1,g_1}Z_{\Sigma_2,g_2}\int_{\R^r}H_0(c,\blambda^c)\d c
\\&=
C_r
\sum_{\substack{\blambda^c\in\Lambda^{2\bbg+\bbb-1}\\\l\in\Lambda}}
e^{-\frac{1}{4\pi}
\left\|
\nu_{\z,\m,\blambda}
+
\Omega^c_{\blambda^c}
\right\|_{g,0}^2}
Z_{\Sigma_1,g_1}Z_{\Sigma_2,g_2}
\int_{\T(\gamma)}
H_0(c+2\pi \ell,\blambda^c)\,\d c ,
\end{aligned}
$$where 
\begin{align*}
H_0(c,\blambda^c)=\E[&\mathcal A^0_{\Sigma_1,g_1}
(\widetilde{\bm\varphi}_1,c+\varphi)
\mathcal A^0_{\Sigma_2,g_2}
(\widetilde{\bm\varphi}_2,c+\varphi)
\\
&\times
\mathcal{B}_{\Sigma_1,g_1}
(F_1,\widetilde{\bm\varphi}_1,c+\varphi,
\blambda_1,\lambda,\blambda_1^c)
\\
&\times
\mathcal{B}_{\Sigma_2,g_2}
(F_2,\widetilde{\bm\varphi}_2,c+\varphi,
\blambda_2,\lambda,\blambda_2^c)].    
\end{align*} 
For $\ell\in\Lambda$, let $P^i_\ell$, $i=1,2$, be the harmonic $\a$-valued
function on $\Sigma_i$ with boundary value $2\pi\ell$ on the glued boundary
component $\mathcal C$, and boundary value $0$ on all the other boundary components.
Thus
$P(\widetilde{\bm\varphi}_i,c+\varphi+2\pi\ell)
=
P(\widetilde{\bm\varphi}_i,c+\varphi)+P^i_\ell ,$ and so 
\begin{equation}\label{eq:free-field-amplitude P^i_l}
    \begin{aligned}
&\mathcal A^0_{\Sigma_i,g_i}
(\widetilde{\bm\varphi}_i,c+\varphi+2\pi\ell)
=
\mathcal A^0_{\Sigma_i,g_i}
(\widetilde{\bm\varphi}_i,c+\varphi)
e^{
-\frac{1}{2\pi}
\langle
\d P(\widetilde{\bm\varphi}_i,c+\varphi),
\d P^i_\ell
\rangle_2
-\frac{1}{4\pi}\|\d P^i_\ell\|_2^2
}.
\end{aligned}
\end{equation}

We now prove that the topological part transforms compatibly with this shift. By the
Hodge decomposition with relative boundary condition, we may write
$\Omega^{i,c}_{\blambda_i^c}
=
\Omega^{i,\mathrm{h}}_{\blambda_i^c}
+
\d u_i,$ for some smooth $\a$-valued function 
$u_i$ with
$u_i|_{\partial\Sigma_i}=0,$
where $\Omega^{i,\mathrm{h}}_{\blambda_i^c}$ is co-closed and satisfies the relative boundary
condition $\iota_{\partial\Sigma_i}^*\Omega^{1,\mathrm{h}}_{\blambda^c_i}=0$. Since $P^i_\ell$ is harmonic, Stokes' formula on $\Sigma_i$ gives
$\langle
\Omega^{i,c}_{\blambda_i^c},
\d P_i^\ell
\rangle_2=0.$
The same argument applies to the magnetic background. Namely, using the decomposition
of the magnetic representative into its co-closed part and a Dirichlet exact part, we write
$\nu^i_{\z_i,\m_i,(\blambda_i,\lambda)}
=
\nu^{i,\mathrm{h}}_{\z_i,\m_i,(\blambda_i,\lambda)}
+
\d u_i',$
for some smooth $\a$-valued function $u_i'$ with $
u_i'|_{\partial\Sigma_i}=0,$
where $\nu^{i,\mathrm{h}}_{\z_i,\m_i,(\blambda_i,\lambda)}$ is co-closed away from the marked
points and satisfies the absolute boundary condition. Again, Stokes' formula gives
$\langle
\nu^i_{\z_i,\m_i,(\blambda_i,\lambda)},
\d P^i_\ell
\rangle_2=0.$ The above identities can be combined to show, for $g_i:=g|_{\Sigma_i},$ 
\begin{equation}\label{eq:computation-regarding-2pil-harmonic-function}
\begin{aligned}
\|
\nu^i_{\z_i,\m_i,(\blambda_i,\lambda)}
+
\Omega^{i,c}_{\blambda_i^c}
\|_{g_i,0}^2
+
\|\d P^i_\ell\|_2^2
=
\|
\nu^i_{\z_i,\m_i,(\blambda_i,\lambda)}
+
\Omega^{i,c}_{\blambda_i^c}
+
\d P^i_\ell
\|_{g_i,0}^2.
\end{aligned}
\end{equation}
Now this
orthogonality identity together with \eqref{eq:free-field-amplitude P^i_l} yields 
$$
\begin{aligned}
&\mathcal A^0_{\Sigma,g,\z,\m,\bzeta}
\left(
F_1\otimes F_2,
\widetilde{\bm\varphi}^{\blambda_1}_1,
\widetilde{\bm\varphi}^{\blambda_2}_2
\right)
\\
&=
C_r
\sum_{\substack{\blambda\in\Lambda^{2\bbg+\bbb-1}\\ \ell\in\Lambda}}
e^{-\frac{1}{4\pi}(
\|
\nu^1_{\z_1,\m_1,(\blambda_1,\lambda)}
+
\Omega^{1,c}_{\blambda_1^c}
+
\d P^1_\ell
\|_{g_1,0}^2
+\|
\nu^2_{\z_2,\m_2,(\blambda_2,\lambda)}
+
\Omega^{2,c}_{\blambda_2^c}
+
\d P^2_\ell
\|_{g_2,0}^2)}
\\
&\times
Z_{\Sigma_1,g_1}Z_{\Sigma_2,g_2}
\int_{\T(\gamma)}
H_1(c,\ell,\blambda^c)\,\d c ,
\end{aligned}
$$
where
$$
\begin{aligned}
H_1(c,\ell,\blambda^c)
:=
\mathbb E_{\mathbb T}\Big[
&
\mathcal A^0_{\Sigma_1,g_1}
\left(
\widetilde{\bm\varphi}_1,c+\varphi
\right)
\mathcal A^0_{\Sigma_2,g_2}
\left(
\widetilde{\bm\varphi}_2,c+\varphi
\right)
\\
&\times
\hat{\mathcal B}_{\Sigma_1,g_1}
\left(
F_1,
\widetilde{\bm\varphi}_1,
c+\varphi,
\blambda_1,
\lambda,
\blambda_1^c,
\ell
\right)
\\
&\times
\hat{\mathcal B}_{\Sigma_2,g_2}
\left(
F_2,
\widetilde{\bm\varphi}_2,
c+\varphi,
\blambda_2,
\lambda,
\blambda_2^c,
\ell
\right)
\Big],
\end{aligned}
$$
with
$$
\begin{aligned}
&\hat{\mathcal B}_{\Sigma_i,g_i}
\left(
F_i,
\widetilde{\bm\varphi}_i,
\widetilde\varphi,
\blambda_i,
\lambda,
\blambda_i^c,
\ell
\right)
 :=
\mathbb E_i\left[
F_i(
\Phi_i)
e^{-\frac{1}{2\pi}
\langle
\d\x_i
+
\d P
\left(
\widetilde{\bm\varphi}_i,\widetilde\varphi
\right),
\nu^i_{\z_i,\m_i,(\blambda_i,\lambda)}
+
\Omega^{i,c}_{\blambda_i^c}
+
\d P^i_\ell
\rangle_2
}\right].
\end{aligned}
$$

We now identify the exact form $\d P_i^\ell$ with a relative cohomology sector. Let $\bm{\l}^{i}=(0,\ldots,0,\l)
\in \Lambda^{2\bbg_i+\bbb_i-1}$. Since $P^i_\ell$ has
boundary value $2\pi\ell$ on the glued boundary component $\mathcal{C}$ and $0$ on all other boundary
components, its periods satisfy
$$
\int_{a_{i,j}}\d P_i^\ell=0,\qquad
\int_{b_{i,j}}\d P_i^\ell=0,
$$
and
$$
\int_{d_{i,j}}\d P_i^\ell=0
\quad
\text{for }j=1,\ldots,\bbb_i-1,
\qquad
\int_{d_{i,\bbb_i}}\d P^i_\ell=2\pi\ell .
$$
Therefore $\d P_i^\ell$ has the same relative periods as
$\Omega^{i,c}_{\bm{\l}^i}$, and hence there exists a smooth $\a$-valued function
$f_i$, vanishing on $\partial\Sigma_i$, such that
$\d P^i_\ell
=
\Omega^{i,c}_{\bm\l^i}
+
\d f^i_\ell .$
We absorb the Dirichlet exact part by the Cameron--Martin formula with the shift $\x_i\mapsto \x_i-f_i^\ell$. Applying the Girsanov transform to the factor $e^{
-\frac{1}{2\pi}
\langle
\d\x_i,\d f^i_\ell
\rangle_2
-\frac{1}{4\pi}
\left\|\d f^i_\ell\right\|_2^2},
$
and using $\langle
\d P
\left(
\widetilde{\bm\varphi}_i,\widetilde\varphi
\right),
\d f^i_\ell
\rangle_2=0,$
we get
$$
\begin{aligned}
e^{-\frac{1}{4\pi}
\|
\nu^i_{\z_i,\m_i,(\blambda_i,\lambda)}
+
\Omega^{i,c}_{\blambda_i^c}
+
\d P^i_\ell
\|_{g_i,0}^2
}
&\hat{\mathcal B}_{\Sigma_i,g_i}
\left(
F_i,
\widetilde{\bm\varphi}_i,
\widetilde\varphi,
\blambda_i,
\lambda,
\blambda_i^c,
\ell
\right)
\\
&=
e^{-\frac{1}{4\pi}
\|
\nu^i_{\z_i,\m_i,(\blambda_i,\lambda)}
+
\Omega^{i,c}_{\blambda_i^c}
+
\Omega^{i,c}_{\bm\l^i}
\|_{g_i,0}^2
}\hat{\mathcal B}'_{\Sigma_i,g_i}
\left(
F_i,
\widetilde{\bm\varphi}_i,
\widetilde\varphi,
\blambda_i,
\lambda,
\blambda_i^c+\bm\l^i
\right),
\end{aligned}
$$
where
$$
\begin{aligned}
&\hat{\mathcal B}'_{\Sigma_i,g_i}
\left(
F_i,
\widetilde{\bm\varphi}_i,
\widetilde\varphi,
\blambda_i,
\lambda,
\blambda_i^c+\bm\l^i
\right)
\\&:=
\mathbb E_i[
F_i
(\Phi_i+I^{\bsigma_i}_{x_0}(\Omega^{i,c}_{\bm\l^i}))
e^{-\frac{1}{2\pi}
\langle
\d\x_i
+
\d P
\left(
\widetilde{\bm\varphi}_i,\widetilde\varphi
\right),
\nu^i_{\z_i,\m_i,(\blambda_i,\lambda)}
+
\Omega^{i,c}_{\blambda_i^c+\mathbf n_i(\ell)}
\rangle_2
}].
\end{aligned}
$$ Indeed, the above expression follows as soon as one observes that $\Omega^{i,c}_{\blambda^c_{i}}+\Omega^{i,c}_{\bm\l^i}=\Omega^{i,c}_{\blambda^c_{i}+\bm\l^i}$. Finally reindexing the lattice sums $(\blambda_i^c,\ell)\mapsto \blambda_i^c+\bm\l^i$, we obtain the desired gluing identity in the subcase $\bbb_2\geq2.$

It remains to indicate the minor modification when $\bbb_2=1$. In this case the second
surface has no remaining boundary component after gluing, so the relative basis on the
glued surface is obtained from the basis of $\Sigma_1$, together with the interior cycles
of $\Sigma_2$. The same argument applies: the neutrality condition again fixes the unique
compatible intermediate winding $\lambda$, the two magnetic backgrounds glue, and the
boundary shift $2\pi\ell$ is represented by a relative cohomology class on the
$\Sigma_1$-side. The exact Dirichlet remainder is again absorbed by Girsanov, and the
sum over $\blambda^c$ is unchanged by the resulting lattice reindexing. Hence the same
identity holds when $\bbb_2=1$. This completes the proof of the case
$\partial\Sigma\neq\emptyset$, and hence the proposition.

It remains to treat the subcase $\bbb_2=1$. In this case the glued surface still has
non-empty boundary, but the second surface contributes no boundary-to-boundary relative arc
after the gluing. The relative homology basis of $H_1(\Sigma,\partial\Sigma)$ is obtained
by taking the interior cycles $a_{i,j},b_{i,j}$ from both surfaces, together with the arcs
$d_{1,j}$ of $\Sigma_1$ for $j=1,\ldots,\bbb_1-2$. Correspondingly, we choose compactly
supported dual representatives
$$
\eta^{1,c}_1,\ldots,\eta^{1,c}_{2\bbg_1+\bbb_1-2},
\eta^{2,c}_1,\ldots,\eta^{2,c}_{2\bbg_2}.
$$
Thus, for $\blambda^c=(\blambda_1^c,\blambda_2^c)
\in
\Lambda^{2\bbg_1+\bbb_1-2}\times \Lambda^{2\bbg_2},$
we write
$\Omega^c_{\blambda^c}
:=
\Omega^{1,c}_{(\blambda_1^c,0)}
+
\Omega^{2,c}_{\blambda_2^c}.$

The magnetic neutrality conditions are now as follows. The only non-trivial case is
$$
\sum_{j=1}^{n_{\mathfrak m,1}}m_{1,j}
+
\sum_{j=1}^{n_{\mathfrak m,2}}m_{2,j}
+
\sum_{\ell=1}^{\bbb_1-1}\varsigma_{1,\ell}\lambda_{1,\ell}
=0.
$$
Together with the neutrality conditions on the two unglued pieces,
 $\sum_{j=1}^{n_{\mathfrak m,1}}m_{1,j}
+
\sum_{\ell=1}^{\bbb_1}\varsigma_{1,\ell}\lambda_{1,\ell}
=0,$ $\sum_{j=1}^{n_{\mathfrak m,2}}m_{2,j}
+
\varsigma_{2,1}\lambda_{2,1}
=0$,
and the fact that the glued boundary components have opposite signs, this forces $\lambda_{1,\bbb_1}=\lambda_{2,1}=:\lambda.$
All other winding sectors in the $\mu_0$-integration give zero on at least one of the two
sides of the gluing identity. Therefore, as before, only the compatible intermediate
winding $\lambda$ contributes.

With this choice of $\lambda$, the magnetic representatives
$\nu^1_{\z_1,\m_1,(\blambda_1,\lambda)}
\text{ and }
\nu^2_{\z_2,\m_2,\lambda}$ can be chosen so that near the glued boundary components they are respectively equal to
$-\varsigma_{1,\bbb_1}\lambda\,\d\theta$ and
$-\varsigma_{2,1}\lambda\,\d\theta$. Since the boundary parametrizations are glued with
opposite orientations, these local expressions match, and hence they glue to the magnetic form
$\nu_{\z,\m,\blambda}
=
\nu^1_{\z_1,\m_1,(\blambda_1,\lambda)}\mathbf 1_{\Sigma_1}
+
\nu^2_{\z_2,\m_2,\lambda}\mathbf 1_{\Sigma_2}$
on the glued surface $\Sigma$, where now
$\blambda=(\lambda_{11},\ldots,\lambda_{1\bbb_1-1})$
is the remaining boundary winding data. The defect graphs are chosen exactly as before:
one arc on each side ends at the distinguished point of the glued boundary component, and
these two arcs concatenate to give the defect graph on $\Sigma$.

We now repeat the preceding computation. The only difference from the case
$\bbb_2\ge2$ is the treatment of the shift
$\widetilde\varphi\mapsto \widetilde\varphi+2\pi\ell$
on $\Sigma_2$. Since $\Sigma_2$ has only one boundary component, the harmonic
extension of this shift is the constant function $2\pi\ell$. Therefore
$\d P(\widetilde\varphi+2\pi\ell)
=
\d P\widetilde\varphi,$
and
$\mathcal A^0_{\Sigma_2,g_2}
(\widetilde\varphi+2\pi\ell)
=
\mathcal A^0_{\Sigma_2,g_2}
(\widetilde\varphi).$
On $\Sigma_1$, the shift is non-constant, but one can apply the above argument to remedy this issue. After reindexing the lattice sum, we complete the proof of the subcase $\bbb_2=1.$
Together with the case
$\bbb_2\ge2$, this completes the proof of the compactified-boson two-surface gluing
identity when $\partial\Sigma\neq\emptyset$, and hence the proof of the case of gluing two surfaces.

The case of self-gluing follows from a similar argument in \cite[Proposition 8.16]{Guillarmou:2023exh}. We leave the details to the reader.

\end{proof}

Now we are in a position to prove Propositions \ref{prop:CITT-boundary-amplitudes-properties} and \ref{prop:CITT-Segal-gluing}. We begin with the former.
\begin{proof}[Proof of Proposition \ref{prop:CITT-boundary-amplitudes-properties}]

We denote by $\Sigma'$ the surface $\Sigma$ with reversed orientation and observe that:
\begin{lemma}\label{lem:reversed-orientatoin}
    The following holds true:
    \begin{equation*}
        \A^0_{\Sigma,g,\z,\m,\bzeta}(F,\widetilde{\bm\varphi}^{\blambda})=\A^0_{\Sigma',g,\z,-\m,\bzeta}(F,\widetilde{\bm\varphi}^{\blambda})
    \end{equation*}
\end{lemma}
\begin{proof}
    By Proposition \ref{prop:magnetic-background}, we have $\nu^{\Sigma'}_{\z,\m,\blambda}=\nu^{\Sigma}_{\z,-\m,\blambda},$ and the claim follows directly from the definition \eqref{eq:free-compactified-boson-amplitude}.
\end{proof}
Next, we show that the amplitudes of surfaces with boundary are well-defined. We consider the regularized amplitudes:
\begin{equation}\label{eq:def-boundary-regularized-amplitude}
    \begin{aligned}
        &\mathcal A^\eps_{\Sigma,g,\bx,\bv,\balpha,\m,\bzeta}(F,\widetilde{\bm\varphi}^{\blambda})\\
        &:=\delta_0(\sum_{j=1}^{n_{\m}+\bbb}m_j(\blambda))\sum_{\blambda^c\in\Lambda^{2\bbg+\bbb-1}}e^{-\frac{1}{4\pi}\norm{\nu_{\z,\m,\blambda}+\Omega^c_{\blambda^c}}^2_{g,0}}Z_{\Sigma,g}\A^0_{\Sigma,g}(\widetilde{\bm\varphi})\\
        &\times\E\left[e^{-\frac{\la d\x_{g,D}+\d P\widetilde{\bm\varphi},\nu_{\z,\m,\blambda}+\Omega^c_{\blambda^c}\ra}{2\pi}}F(\Phi_g)\prod^{n_{\m}}_{j=1}V^g_{\alpha_j,\eps}(x_j)e^{-\frac{\bi}{4\pi}\la QK_g,\Phi_g\ra^{\reg}_{g}-\sum^r_{i=1}\mu_iM^g_{\gamma e_i}(\Phi_g,\Sigma)}\right]
    \end{aligned}
\end{equation}

\begin{lemma}\label{lem:boundary-regularized-amplitude-convergence}As $\eps\rightarrow0$ and $\bx\rightarrow\z$ in the direction $\bv=((z_1,v_1),\ldots,(z_n,v_n))\in(T\Sigma)^n$, $$\mathcal A^\eps_{\Sigma,g,\bx,\bv,\balpha,\m,\bzeta}(F,\widetilde{\bm\varphi}^{\blambda})\rightarrow\mathcal A_{\Sigma,g,\bv,\balpha,\m,\bzeta}(F,\widetilde{\bm\varphi}^{\blambda}),\quad(\d c\otimes\P_{\T})^{\bbb}\text{ almost surely},$$ where \begin{equation}\label{eq:boundary-amplitude}
    \begin{aligned}
        &\mathcal A_{\Sigma,g,\bv,\balpha,\m,\bzeta}(F,\widetilde{\bm\varphi}^{\blambda})\\
&=\delta_0(\sum_{j=1}^{n_{\m}+\bbb}m_j(\blambda))Z_{\Sigma,g}\A^0_{\Sigma,g}(\widetilde{\bm\varphi})e^{-\sum_{j<j'}\la\alpha_j\alpha_{j'}\ra G_{g,D}(z_j,z_{j'})-\sum_j\frac{|\alpha_j|^2}{2}W_g(z_j)}\prod_{j=1}^{n_{\m}}e^{\bi\la\alpha_j,P\widetilde{\bm\varphi}(z_j)\ra}\\&\times\sum_{\blambda^c\in\Lambda^{2\bbg+\bbb-1}}e^{-\frac{1}{4\pi}\norm{\nu_{\z,\m,\blambda}+\Omega^c_{\blambda^c}}^2_{g,0}}\prod_{j=1}^{n_{\m}}e^{i\la\alpha_j,I^{\bsigma}_{x_0}(\Omega^c_{\blambda^c})+I^{\bxi}_{x_0}(\nu_{\z,\m,\blambda})\ra(z_j)}\times\\
        &\E\left[e^{-\frac{\la d\x_{g,D}+\d P\widetilde{\bm\varphi}+\d u_{\z},\nu_{\z,\m,\blambda}+\Omega^c_{\blambda^c}\ra}{2\pi}}F(\Phi_g+u_{\z})e^{-\frac{\bi}{4\pi}\la QK_g,\Phi_g+u_{\z}\ra^{\reg}_{g}-\sum^r_{i=1}\mu_iM^g_{\gamma e_i}(\Phi_g+u_{\z},\Sigma)}\right],
    \end{aligned}
\end{equation} with $u_{\z}=\sum_{j=1}^{n_{\m}}i\alpha_jG_{g,D}(\cdot,z_j),$ the Toda field $\Phi_g=\x_{g,D}+P\widetilde{\bm\varphi}+I^{\bxi}_{x_0}(\nu_{\z,\m\blambda})+I^{\bsigma}_{x_0}(\Omega^c_{\blambda^c}),$ the expectation over the Dirichlet GFF $\x_{g,D},$ and the evaluation of $I^{\bxi}_{x_0}(\nu_{\z,\m,\blambda})$ at the points $z_j$ done in the direction $v_j.$
\end{lemma}
\begin{proof}
    The proof is similar to that of Theorems \ref{thm:mixed operators} and \ref{thm:electro-magnetic operators}. One starts with the Girsanov transform, where an analogous statement of Lemma \ref{lem:toda-order-2-distribution} is needed (see \cite[Lemma 8.18]{Guillarmou:2023exh}), and the exponential estimate is provided by Proposition \ref{prop:exp-moment-shifted}. Upon rewriting the regularized amplitude \eqref{eq:def-boundary-regularized-amplitude} using the Girsanov transform, one follows again the proof of Theorems \ref{thm:mixed operators} and \ref{thm:electro-magnetic operators} to obtain that the difference $\Delta_{\eps,\bx}$ between regularized amplitudes satisfies the bound $$|\Delta_{\eps,\bx}|\leq\sum_{\blambda\in\Lambda^{2\bbg}}\A^0_{\Sigma,g}(\widetilde{\bm\varphi})e^{-\frac{1}{4\pi}\norm{\prod^c_1\Omega^c_{\blambda^c}}^2_2+C|\blambda|-\frac{1}{2\pi}\la\d P\widetilde{\bm\varphi},\prod^c_{\blambda^c}\Omega^c_{\blambda^c}\ra}C_{\eps,\bx,\z}(F,\widetilde{\bm{\varphi}})$$ for some constant $C_{\eps,\bx,\z}(F,\widetilde{\bm{\varphi}})$ that vanishes $\mu_0^{\otimes\bbb}$ almost surely in $\widetilde{\bm\varphi}$ as $\eps\rightarrow0$ and $\bx\rightarrow\z.$ Moreover, the quantity $\la\d P\widetilde{\bm\varphi},\prod^c_{\blambda^c}\Omega^c_{\blambda^c}\ra$ is at most linear in $|\blambda|$, and this proves the claim.
\end{proof}

Finally, we claim that the amplitudes of surfaces with boundary with fixed functional $F$ belong to the Hilbert space $\H^{\otimes\bbb}$:
\begin{lemma}\label{lem:boundary-amplitude-L2}
    If $\sum_{j=1}^{n_{\m}}m_j=0$ and $\alpha_j-Q\in \mathcal{C}_+,$ then for $F\in\mathcal{E}^{\m}_\Lambda(\Sigma;\a),$ we have $$\A_{\Sigma,g,\bv,\balpha,\m,\bzeta}(F,\cdot)\in\H^{\otimes\bbb}.$$
\end{lemma}

\begin{proof}
    Let $\Sigma'$ denote a copy of $\Sigma$ with reversed orientation. We glue the $i$-th component of $\Sigma$ to that of $\Sigma'$ to obtain the closed double surface $\Sigma^{\#2}$, with an induced involution $\tau:\Sigma^{\#2}\rightarrow\Sigma^{\#2}$ sending a point $x\in\Sigma$ to the corresponding copy in $\Sigma'$, and the extended metric $g$ on $\Sigma^{\#2}$ symmetric under the involution $\tau.$

    We recall that $u_{\z}=\sum_{j=1}^{n_{\m}}i\alpha_jG_{g,D}(\cdot,z_j)$ and set $$F_{\z}(\Phi):=e^{-\frac{1}{2\pi}\la \d u_\z,\nu_{\z,\m,\blambda}+\Omega^c_{\blambda^c}\ra}|F(\Phi+u_\z)||e^{-\sum^r_{j=1}\mu_jM^g_{\gamma e_j}(\Phi+u_{\z},\Sigma)}|.$$ Note that the representation \eqref{eq:boundary-amplitude} gives that $$\A_{\Sigma,g,\bv,\balpha,\m,\bzeta}(F,\widetilde{\bm\varphi}^{\blambda})\leq C\A^0_{\Sigma,g,\bv,\m,\bzeta}(F_{\z},\widetilde{\bm\varphi}^{\blambda})$$ for some constant $C$ independent of $\widetilde{\bm\varphi}^{\blambda}$. Now Proposition \ref{prop:CITT-free-boson-two-surface-gluing} shows that 
$$C'\int\A^0_{\Sigma,g,\z,\m,\bzeta}(F_{\z},\widetilde{\bm\varphi}^{\blambda})\A^0_{\Sigma',\tau^*g,\z,-\m,\bzeta}(F_{\z},\widetilde{\bm\varphi}^{\blambda})\d\mu_0^{\otimes\bbb}(\widetilde{\bm\varphi}^{\blambda})=\A^0_{\Sigma^{\#2},g,\hat{\z},\hat\m}(\hat F_{\z},\widetilde{\bm\varphi}^{\blambda})$$ for some explicit constant $C'$ obtained from the constants in Proposition \ref{prop:CITT-free-boson-two-surface-gluing}, where $\hat\z=(\z,\tau(\z)),\, \hat\m=(\m,-\m),$ and $$\hat F_{\z}(\Phi):=e^{-\frac{1}{2\pi}\la \d \hat u_\z,\nu_{\z,\m}+\Omega_{\blambda}\ra}|F(\Phi+\hat u_\z|_{\Sigma})||F(\Phi+\hat u_\z|_{\Sigma'})||e^{-\sum^r_{j=1}\mu_jM^g_{\gamma e_j}(\Phi+\hat u_{\z},\Sigma^{\#2})}|,$$ with $\hat u_{\z}:=u_{\z}\bone_{\Sigma}+\tau^*u_{\z}\bone_{\Sigma'}.$ One can follow the argument of the proof of Theorem \ref{thm:electro-magnetic operators} to show that $\A^0_{\Sigma^{\#2},g,\hat{\z},\hat\m}(\hat F_{\z},\widetilde{\bm\varphi}^{\blambda})$ is finite. It then follows from Lemma \ref{lem:reversed-orientatoin} that 
\begin{align*}
    &\int|\A_{\Sigma,g,\bv,\balpha,\m,\bzeta}(F,\widetilde{\bm\varphi}^{\blambda})|^2\d\mu_0^{\otimes\bbb}(\widetilde{\bm\varphi}^{\blambda})\\
    &\leq C^2\int|\A^0_{\Sigma,g,\z,\m,\bzeta}(F_{\z},\widetilde{\bm\varphi}^{\blambda})|^2\d\mu_0^{\otimes\bbb}(\widetilde{\bm\varphi}^{\blambda})\\
&=C^2\int\mathcal{A}^0_{\Sigma,g,\z,\m,\bzeta}(F_\z,\widetilde{\bm\varphi}^{\blambda})\mathcal{A}^0_{\Sigma',\tau^*g,\z,-\m,\bzeta}(F_\z,\widetilde{\bm\varphi}^{\blambda})\d\mu_0^{\otimes\bbb}(\widetilde{\bm\varphi}^{\blambda})\\&=\frac{C^2}{C'}\A^0_{\Sigma^{\#2},g,\hat{\z},\hat\m}(\hat F_{\z},\widetilde{\bm\varphi}^{\blambda})<\infty
\end{align*}
\end{proof}
This proves the first claim of Proposition \ref{prop:CITT-boundary-amplitudes-properties}.

Next we show that the amplitude is invariant under the translations
$$
\widetilde{\bm\varphi}^{\blambda}
\mapsto
\widetilde{\bm\varphi}^{\blambda}
+
2\pi\ell\,\mathbf 1_{\partial_j\Sigma},
\qquad \ell\in\Lambda,\quad j=1,\ldots,\bbb.
$$ We assume for simplicity that the base point $x_0$ lies on the first boundary component
$\partial_1\Sigma$. First consider the simultaneous translation
$$
c_j\mapsto c_j+2\pi\ell,
\qquad j=1,\ldots,\bbb,
\qquad \ell\in\Lambda.
$$
This amounts to replacing the bulk field $\Phi_g$ by $\Phi_g+2\pi\ell$. The Toda
interaction is unchanged since,  the
test functional $F$ is $2\pi\Lambda$-periodic, and the terms depending only on
$\d P\widetilde{\bm\varphi}$ and on
$\mathcal A^0_{\Sigma,g}(\widetilde{\bm\varphi})$ are unchanged. The only term which
requires comment is the curvature phase. Under $\Phi_g\mapsto\Phi_g+2\pi\ell$, it changes by
$$
e^{
-\frac{\bi}{4\pi}
\int_\Sigma 2\pi\langle Q,\ell\rangle K_g\,\d v_g
}=
e^{-\bi\pi\langle Q,\ell\rangle\chi(\Sigma)
}=1.
$$
thanks to the assumption $Q\in\Lambda^*$.
Thus the amplitude is invariant under simultaneous translation of all boundary constants
by the same element $2\pi\ell$.

It remains to prove invariance under a translation of a single boundary component. If
$j=1$, then translating only $c_1$ by $2\pi\ell$ can be reduced to translating all
components by $2\pi\ell$, followed by translating each $c_j$, $j>1$, by
$-2\pi\ell$. Hence it suffices to prove invariance under
$c_j\mapsto c_j+2\pi\ell$ for $j>1$. Let $h$ be an $\a$-valued function on $\partial\Sigma$, locally constant on each
boundary component, and set
$W(h)
:=
\mathcal A_{\Sigma,g,\bv,\balpha,\m,\bzeta}
\left(F,\widetilde{\bm\varphi}^{\blambda}+h\right).$
We shall prove
$$
W(h_j)=W(0),
\qquad
h_j:=2\pi\ell\,\mathbf 1_{\partial_j\Sigma},
\qquad j>1,\quad \ell\in\Lambda.
$$
For notational simplicity, write $\Theta_{\blambda^c}
:=
\nu_{\z,\m,\blambda}+\Omega^c_{\blambda^c}.$
Also absorb into a single symbol $H$ all the factors in the amplitude which are
$2\pi\Lambda$-periodic functions of the field, namely
$$
H(\Phi):=F(\Phi)\prod_{j=1}^{n_{\mathfrak m}}V^g_{\alpha_j,\varepsilon}(x_j(t))
e^{-\frac{\bi}{4\pi}\langle QK_g,\Phi\rangle^{\reg}_g
-\sum_{i=1}^r\mu_iM^g_{\gamma e_i}(\Phi,\Sigma)},$$
and the amplitude with boundary shift $h_j$ can be written as
$$
\begin{aligned}
W(h_j)=&\sum_{\blambda^c\in\Lambda^{2\bbg+\bbb-1}}\mathcal A^0_{\Sigma,g}(\widetilde{\bm\varphi}+h_j)e^{-\frac{1}{4\pi}\|\Theta_{\blambda^c}\|_{g,0}^2}\mathbb E\left[e^{-\frac{1}{2\pi}\langle\d \x_{g,D}+\d P\widetilde{\bm\varphi}+\d P h_j,\Theta_{\blambda^c}\rangle_2}H(\Phi_g+P h_j)
\right],
\end{aligned}
$$
where
$\Phi_g=\x_{g,D}+P\widetilde{\bm\varphi}+I^{\bsigma}_{x_0}(\Omega^c_{\blambda^c})+I^{\bxi}_{x_0}(\nu_{\z,\m,\blambda}).$
Here and below the normalizing factor $Z_{\Sigma,g}$, the electric regularization limits,
and the magnetic neutrality factor are suppressed, since they play no role in this
periodicity argument.

We observe that$$
\mathcal A^0_{\Sigma,g}(\widetilde{\bm\varphi}+h_j)
=e^{-\frac{1}{4\pi}\int_\Sigma |\d P h_j|_g^2\,\d v_g
-\frac{1}{2\pi}\langle\d P h_j,\d P\widetilde{\bm\varphi}\rangle_2}
\mathcal A^0_{\Sigma,g}(\widetilde{\bm\varphi}),$$
and so 
$$
\begin{aligned}
W(h_j)=&\sum_{\blambda^c\in\Lambda^{2\bbg+\bbb-1}}
\mathcal A^0_{\Sigma,g}(\widetilde{\bm\varphi})
e^{-\frac{1}{4\pi}
\|\Theta_{\blambda^c}+\d P h_j\|_{g,0}^2}
\mathbb E\left[
e^{-\frac{1}{2\pi}\langle\d \x_{g,D}+\d P\widetilde{\bm\varphi},
\Theta_{\blambda^c}+\d P h_j\rangle_2}H(\Phi_g+P h_j)\right].
\end{aligned}
$$

We also note that $\d P h_j$. Since $h_j$ belongs to the same relative cohomology class as $\eta^c_{2\bbg+j-1}\otimes \ell.$ Consequently there exists a smooth $\a$-valued function $f_j$, with
$f_j|_{\partial\Sigma}=0$, such that
$\d P h_j
=
\eta^c_{2\bbg+j-1}\otimes \ell+\d f_j .$
After reindexing the lattice summation, we obtain
$$
\begin{aligned}
W(h_j)
=
\sum_{\blambda^c\in\Lambda^{2\bbg+\bbb-1}}
\mathcal A^0_{\Sigma,g}(\widetilde{\bm\varphi})
e^{-\frac{1}{4\pi}
\|\Theta_{\blambda^c}+\d f_j\|_{g,0}^2
}
\mathbb E\left[
e^{-\frac{1}{2\pi}
\left\langle
\d \x_{g,D}
+\d P\widetilde{\bm\varphi},
\Theta_{\blambda^c}+\d f_j
\right\rangle_2
}H(\Phi_g+f_j)
\right].
\end{aligned}
$$
Expanding the regularized norm gives
$$
\|\Theta_{\blambda^c}+\d f_j\|_{g,0}^2
=
\|\Theta_{\blambda^c}\|_{g,0}^2
+
2\langle \Theta_{\blambda^c},\d f_j\rangle_2
+
\|\d f_j\|_2^2.
$$
Since $f_j$ vanishes on $\partial\Sigma$, integration by parts gives
$$
\|\d f_j\|_2^2
=
\langle f_j,\Delta_g f_j\rangle_2,
\qquad
\langle \d \x_{g,D},\d f_j\rangle_2
=
\langle \x_{g,D},\Delta_g f_j\rangle_2,
$$
and thus
$$
\begin{aligned}
W(h_j)
=&
\sum_{\blambda^c\in\Lambda^{2\bbg+\bbb-1}}
\mathcal A^0_{\Sigma,g}(\widetilde{\bm\varphi})
e^{-\frac{1}{4\pi}
\|\Theta_{\blambda^c}\|_{g,0}^2
}\\&\times\mathbb E\left[
e^{-\frac12
\langle \x_{g,D},\Delta_g f_j\rangle
}
e^{-\frac{1}{2\pi}
\left\langle
\d(X_{g,D}+f_j+P\widetilde{\bm\varphi}),
\Theta_{\blambda^c}
\right\rangle_2
}H(\Phi_g+f_j)
\right].
\end{aligned}
$$
Now one applies the Girsanov transform to the Dirichlet GFF $\x_{g,D}$. This is allowed
because $f_j$ vanishes on $\partial\Sigma$. The
Girsanov shift absorbs the replacement $\x_{g,D}\mapsto \x_{g,D}+f_j$, and hence gives $W(h_j)=W(0).$ The same argument can be applied to show that the amplitude does not depend on the relative cohomology representatives.

Changing relative homology bases amounts to reindexing the lattice sum, after which the corresponding relative cohomology representatives have the same relative cohomology class. By the previous paragraph, replacing one representative by the
other does not change the non-curvature part of the summand.
The fact that the curvature term is invariant under this change is given by
Lemma~\ref{lem:boundary-relative-curvature-independence} and the assumption $Q\in\Lambda^*$.

We turn to the independence of the representative of
$\nu_{\z,\m,\blambda}$. This point is slightly more delicate since an arbitrary exact change of
$\nu_{\z,\m,\blambda}$ is not necessarily of Dirichlet type and therefore cannot be
absorbed directly by the Dirichlet GFF.
Let
$\nu_{\z,\m,\blambda}$
and $\nu'_{\z,\m,\blambda}$
be two magnetic background forms satisfying the assumptions of
Proposition~\ref{prop:magnetic-background} and \eqref{nu-relative-period}. Since they represent the same absolute
cohomology class and have the same prescribed singularities and periods, their difference is
exact:
$\nu_{\z,\m,\blambda}-\nu'_{\z,\m,\blambda}=\d f$
for some smooth $\a$-valued function $f$ on $\Sigma_\z$. Moreover, by the local
normalization of the magnetic forms near the marked points and near the boundary
components, the difference actually vanishes in these neighborhoods. Hence $f$ is locally
constant near each boundary component and near each marked point. We normalize $f$ by
adding a constant so that $f(x_0)=0$.

The locally constant boundary jumps of $f$ are encoded by a relative cohomology class.
Thus there exists
$\blambda^c_0\in\Lambda^{2\bbg+\bbb-1}$
such that
$\d h=\Omega^c_{\blambda^c_0}+\d f_{\blambda^c_0}$
for some smooth $\a$-valued function $h$, constant near the boundary components and
marked points with values in $2\pi\Lambda$, and some smooth function
$f_{\blambda^c_0}$ vanishing on $\partial\Sigma$. Equivalently, after decomposing $f=C+h+\bar f,$
with $C$ constant, $h$ locally constant near the boundary components and marked points
with values in $2\pi\Lambda$, and $\bar f|_{\partial\Sigma}=0$, the non-Dirichlet part
of $\d f$ is precisely represented by $\Omega^c_{\blambda^c_0}$. In particular,
$\d f=\Omega^c_{\blambda^c_0}+\d f_{\blambda^c_0}+\d\bar f .$

We now insert $\nu_{\z,\m,\blambda}=\nu'_{\z,\m,\blambda}+\d f$
into the definition of the amplitude. In the summand indexed by
$\blambda^c$, the topological form becomes
$$
\nu_{\z,\m,\blambda}+\Omega^c_{\blambda^c}=\nu'_{\z,\m,\blambda}+
\Omega^c_{\blambda^c+\blambda^c_0}+\d(f_{\blambda^c_0}+\bar f).
$$
After this reindexing $\blambda^c\mapsto \blambda^c-\blambda^c_0$
, the
only remaining difference between the two expressions is the exact Dirichlet-type shift
$\d(f_{\blambda^c_0}+\bar f),\,
(f_{\blambda^c_0}+\bar f)|_{\partial\Sigma}=0,$
which can be absorbed by the Dirichlet GFF, again, via the Girsanov transform with the shift
$\x_{g,D}\mapsto \x_{g,D}+f_{\blambda^c_0}+\bar f$. 
The locally constant part $h$ only changes the field by an element of $2\pi\Lambda$, leaving the amplitude unchanged, as shown above. Moreover, the curvature term is also unchanged thanks to the assumption $Q\in\Lambda^*.$ This proves the independence of the representative of $\nu_{\z,\m,\blambda},$ and
the independence of the defect graph is given by Lemma~\ref{lem:boundary-defect-graph-independence} and again the assumption $Q\in\Lambda^*.$

We proceed to the conformal anomaly. Let $g'=e^\rho g$, with $\rho\in C^\infty(\Sigma)$ and
$\rho|_{\partial\Sigma}=0$. Since the boundary value of $\rho$ vanishes, the
boundary parametrizations, the boundary measure $\mu_0$, and the boundary field
$\widetilde{\bm\varphi}^{\blambda}$ are unchanged. Moreover the Dirichlet GFF is
conformally invariant in law:
$$
\x_{g',D}\overset{\mathrm{law}}{=}\x_{g,D},
$$
and the harmonic extension of the boundary field is also unchanged. Thus the only changes in the amplitude come from the determinant of the Dirichlet
Laplacian, the vertex operators, the Toda interaction term, the curvature coupling, and the
regularized norms of the singular magnetic forms.

We first recall from \eqref{eq:weyl IGMC} that $M^{g'}_{\gamma e_i}(\x_{g',D},\d x)=e^{(1-\gamma^2\la e_i,e_i\ra/4)\rho(x)}M^g_{\gamma e_i}(\x_{g,D},\d x)$ and that the Polyakov--Alvarez formula \cite[Section 1]{OSGOOD1988148} gives
\begin{equation}\label{eq:weyl-det-dirichlet-vector}
Z_{\Sigma,g'}
=
e^{\frac{r}{96\pi}
\int_\Sigma
\left(
|d\rho|_g^2+2K_g\rho
\right)\d \v_g
}Z_{\Sigma,g}.
\end{equation}
Next, the electric vertex renormalization gives, exactly as in the closed case,
\begin{equation}\label{eq:weyl-electric-renormalization}
\prod_{j=1}^{n_{\mathfrak m}}
V^{g'}_{\alpha_j,\varepsilon}(x_j(t))
=
e^{-\sum_{j=1}^{n_{\mathfrak m}}
\frac{|\alpha_j|^2}{4}\rho(z_j)
+o_{\varepsilon,t}(1)
}\prod_{j=1}^{n_{\mathfrak m}}
V^{g}_{\alpha_j,\varepsilon}(x_j(t)),
\end{equation}
as $\varepsilon\to0$ and then $t\to1$. By Lemma \ref{lem:boundary-curvature-conformal-change}, the contribution of the magnetic part is\begin{equation}\label{eq:weyl-magnetic-energy}
e^{
-\frac{1}{4\pi}
\|\nu_{\z,\m,\blambda}+\Omega^c_{\blambda^c}\|_{g',0}^2
}=
e^{-\frac{1}{4}
\sum_{j=1}^{n_{\mathfrak m}}|m_j|^2\rho(z_j)
}e^{-\frac{1}{4\pi}
\|\nu_{\z,\m,\blambda}+\Omega^c_{\blambda^c}\|_{g,0}^2
}.
\end{equation}
It remains to treat the curvature coupling. We use
$K_{g'}\,\d \v_{g'}
=
\left(K_g+\Delta_g\rho\right)\d \v_g .$
to obtain$$
\int_\Sigma
\left\langle QK_{g'},\x_{g,D}+P\widetilde{\bm\varphi}\right\rangle
\d \v_{g'}
=
\int_\Sigma
\left\langle QK_g,\x_{g,D}+P\widetilde{\bm\varphi}\right\rangle
\d \v_g+\int_\Sigma
\left\langle Q\Delta_g\rho,\x_{g,D}\right\rangle
\d \v_g,$$
where we have use the Green identity to get $$\int_\Sigma\Delta_g\rho P\widetilde{\bm\varphi}\d \v_g=\int_{\partial\Sigma}\partial_\nu\rho P\widetilde{\bm\varphi}\d \v_g, $$ which vanishes since $g,\,g'$ are admissible so that $\partial_\nu\rho=0.$ Moreover, Lemma~\ref{lem:boundary-curvature-conformal-change}
gives
$$
\begin{aligned}
&\int_{\Sigma_{\bsigma}}^{\reg}
I^{\bsigma}_{x_0}(\Omega^c_{\blambda^c})K_{g'}\,\d \v_{g'}
+
\int_{\Sigma}^{\reg}
I^{\bxi}_{x_0}(\nu_{\z,\m,\blambda})K_{g'}\,\d \v_{g'}
\\
&\qquad =
\int_{\Sigma_{\bsigma}}^{\reg}
I^{\bsigma}_{x_0}(\Omega^c_{\blambda^c})K_g\,\d \v_g
+
\int_{\Sigma}^{\reg}
I^{\bxi}_{x_0}(\nu_{\z,\m,\blambda})K_g\,\d \v_g
+
\left\langle
\d\rho,\Omega^c_{\blambda^c}+\nu_{\z,\m,\blambda}
\right\rangle_2 .
\end{aligned}
$$
Applying the Girsanov transform with the shift $\x_{g,D}
\mapsto
\x_{g,D}-\frac{\bi}{2}Q\rho $ (since the shift is imaginary, one has to go through an analytic continuation as in Theorem \ref{thm:mixed operators}) and checking the factors thus produced, we obtain the desired formula.

Finally, the diffeomorphism invariance and the spin covariance follow from a similar argument to the proof in Theorem \ref{thm:mixed operators}. This completes the proof.

\end{proof}

We conclude this section by proving Proposition \ref{prop:CITT-Segal-gluing}.
\begin{proof}[Proof of Proposition \ref{prop:CITT-Segal-gluing}]
For the case of gluing two surfaces, we apply Proposition \ref{prop:CITT-free-boson-two-surface-gluing}
with $F_i(\Phi):=G_i(\Phi)\prod^{n_\m^i}_{j=1}V_{\alpha_{ij},\eps}^g(x_j)e^{-\frac{\bi}{4\pi}\langle QK_{g_i},\Phi\rangle^{\reg}_{g_i}-\sum_{i=1}^r\mu_iM^{g_i}_{\gamma e_i}(\Phi,\Sigma_i)},$
to obtain the gluing identity for regularized amplitudes:
\begin{align*}
&\mathcal{A}^\eps_{\Sigma,g,\bx,\bv,\balpha,\m,\bzeta}(G_1\otimes G_2,\widetilde{\bm\varphi}^{\blambda_1}_1,\widetilde{\bm\varphi}^{\blambda_2}_2)\\&=C\int\mathcal{A}^\eps_{\Sigma_1,g_1,\bx_1,\bv_1,\balpha_1,\m_1,\bzeta_1}(G_1,\widetilde{\bm\varphi}^{\blambda_1}_1,\widetilde{\varphi}^\lambda)\mathcal{A}^\eps_{\Sigma_2,g_2,\bx_2,\bv_2,\balpha_2,\m_2,\bzeta_2}(G_2,\widetilde{\bm\varphi}^{\blambda_2}_2,\widetilde{\varphi}^\lambda)\d\mu_0(\widetilde{\varphi}^\lambda),
\end{align*} where $C$ is the constant given in Proposition \ref{prop:CITT-Segal-gluing}. We note that we need Lemma \ref{lem:boundary-curvature-additivity-gluing} to gurantee that the functionals $F_1$ and $F_2$ glue to $$F_1\otimes F_2(\Phi)=G_1\otimes G_2(\Phi)\prod_{i=1,2}(\prod^{n_m^i}_{j=1}V^{g}_{\alpha_{ij},\eps}(x_j))e^{-\frac{\bi}{4\pi}\langle QK_{g},\Phi\rangle^{\reg}_{g}-\sum_{i=1}^r\mu_iM^{g}_{\gamma e_i}(\Phi,\Sigma)}.$$
It remains to justify the passage to the limit as $\eps\rightarrow0$ and $\bx\rightarrow\z$ in the prescribed directions. This follows from 
\begin{lemma}
    The following convergence holds in $\H^{\otimes\bbb}$,
    $$\lim_{\bx\rightarrow\z}\lim_{\eps\rightarrow0}\mathcal{A}^\eps_{\Sigma,g,\bx,\bv,\balpha,\m,\bzeta}(F,\widetilde{\bm\varphi}^{\blambda})=\mathcal{A}_{\Sigma,g,\bv,\balpha,\m,\bzeta}(F,\widetilde{\bm\varphi}^{\blambda}),$$ where $\bx$ approach $\z$ in the direction of $\bv.$ 
\end{lemma}
We omit the proof since it follows from a similar argument to \cite[Lemma 8.21]{Guillarmou:2023exh}. This proves the gluing identity in the case of gluing two surfaces. 

The case of self-gluing is treated in the same way. The only additional point is that self-gluing is expressed as a partial trace. This trace is well-defined by the same Hilbert--Schmidt argument as in \cite[Proposition 8.14]{Guillarmou:2023exh}: after
cutting out small annuli around the two boundary components to be identified, the regularized amplitude is written as a composition of Hilbert--Schmidt operators, since
annular amplitudes are $L^2$-kernels. The convergence lemma above then allows one
to pass to the limits $\eps\to0$ and $\bx\to\z$, while the behavior of the regularized
curvature term under self-gluing is given by Lemma
\ref{lem:boundary-curvature-additivity-gluing}. This proves the self-gluing identity
and completes the proof of Proposition \ref{prop:CITT-Segal-gluing}.
\end{proof}

\appendix

\section{Proof of Proposition \ref{prop:closed formula}}\label{Appendix}
\subsection{The complex twin and meromorphic continuation}
We first introduce a few notations. We denote by $\Lambda_{\C}$ the collection of pairs $(a,a')\in\C^2$ such that $a-a'\in\Z$ and write $z^{a|a'}=z^a\bar z^{a'}$ with the choice of branch such that $(-1)^{a|a'}=(-1)^{a-a'}.$ We then define the Gamma function of the complex field \cite{IMGelfand_2004} by \begin{equation}\label{eq:complex Gamma}\Gamma^\C(a|a'):=\frac{1}{\pi}\int_{\C}z^{a-1|a'-1}e^{2\bi\Re z}\d z=\bi^{a-a'}\frac{\Gamma(a)}{\Gamma(1-a')}=\frac{\bi^{a'-a}}{\pi}\Gamma(a)\Gamma(a')\sin(\pi a').\end{equation} This integral converges conditionally if $0<\frac{1}{2}\Re(a+a')<1$ and admits a meromorphic continuation to the complex lines $a-a'\in\Z$. Finally, for $\z=(z_1,\ldots,z_n)\in\C^n$, we let $D_n(\z)=\prod_{1\leq i<j\leq n}(z_i-z_j)$ and $\d\nu_n(z)=(\pi^n n!)^{-1}\prod_{i=1}^n\d z_i$.  

 We also record the following identity: for $(\lambda,\lambda')\in\Lambda_{\C}$ satisfying $-2<\Re(\lambda+\lambda')<0$ and $a\in\C^\times$, we have
\begin{equation}\label{eq:one-variable-Fourier}
\frac1\pi\int_{\C} z^{\lambda|\lambda'}e^{\,a\bar z-\bar a z}\d z
=
\Gamma_{\C}(1+\lambda\,|\,1+\lambda')\,a^{-1-\lambda'\,|\,-1-\lambda}.
\end{equation} This can be obtained from \eqref{eq:complex Gamma} via the change of variable $w=\bar{a}z$.

To follow the argument in \cite[Appendix A]{fateev2007correlation}, we have to establish the complex twin of \cite[eq. (A.7)]{fateev2007correlation} since the integrand in \eqref{eq:DF-0-1-infty} is no longer the modulus of, say, $1-x^{(i)}_a$ due to the nonzero magnetic charge $m_2.$
\begin{lemma}\label{lem:complex-twin-A7-final}
Let $p,q\in\N_0$, set $M:=p+q+1$, and let $u_1,\dots,u_M\in\C$ be pairwise distinct. Let $(\lambda_1,\lambda_1'),\dots,  (\lambda_M,\lambda_M')\in\Lambda_{\C}.$ Then
\begin{equation}\label{eq:complex-twin-A7-final-statement}
\begin{aligned}
&\int_{\C^p}
D_p(z)^{1|1}
\prod_{\alpha=1}^{p}\prod_{j=1}^{M}(z_\alpha-u_j)^{\lambda_j|\lambda_j'}
\d\nu_p(z)
\\&=
(-1)^{\sum_{j=1}^{M}(j-1)(\lambda_j-\lambda_j')}
\frac{
\prod_{j=1}^{M}\Gamma_{\C}(1+\lambda_j\,|\,1+\lambda_j')
}{
\Gamma_{\C}\!\left(
p+1+\sum_{j=1}^{M}\lambda_j
\ \Big|\
p+1+\sum_{j=1}^{M}\lambda_j'
\right)
}
\\
&\prod_{1\le i<j\le M}
(u_j-u_i)^{1+\lambda_i+\lambda_j\,|\,1+\lambda_i'+\lambda_j'}
\int_{\C^q}
D_q(w)^{1|1}
\prod_{\beta=1}^{q}\prod_{j=1}^{M}(w_\beta-u_j)^{-1-\lambda_j\,|\,-1-\lambda_j'}
\d\nu_q(w),
\end{aligned}
\end{equation} where both integrals converge absolutely if  
\begin{equation}\label{eq:Omega-conv-final}
-2<\Re(\lambda_j+\lambda_j')<0,\quad
-2p-2<\sum_{j=1}^M \Re(\lambda_j+\lambda_j')<-2p.
\end{equation}
\end{lemma}
\begin{proof}
Under the condition \eqref{eq:Omega-conv-final}, both integrals converge absolutely: the inequalities $-2<\Re(\lambda_j+\lambda_j')<0$ ensure local integrability at the finite singularities, and the bounds $$ -2p-2<\sum_{j=1}^M \Re(\lambda_j+\lambda_j')<-2p $$ give sufficient decay at infinity.

Set $U(t):=\prod_{j=1}^{M}(u_j-t)$ and $P_z(t):=\prod_{\alpha=1}^{p}(z_\alpha-t).$ For $1\le j\le M$, define
$$
\omega_j:=\frac{P_z(u_j)}{-U'(u_j)}
=
\frac{\prod_{\alpha=1}^{p}(z_\alpha-u_j)}{\prod_{\ell\neq j}(u_\ell-u_j)}.
$$
Then
$$
\sum_{j=1}^{M}\frac{\omega_j}{u_j-t}=\frac{P_z(t)}{U(t)}=(-1)^{q+1}t^{-q-1}+O(t^{-q-2}),
$$
near infinity.
Writing $\frac1{u_j-t}=-\sum_{m\ge0}u_j^m t^{-m-1}$ and comparing coefficients imply that $\omega=(\omega_1,\dots,\omega_M)$ lies in the affine subspace
$$
\mathcal A_{p,q}(u):=
\Bigl\{\omega\in\C^M:\ \sum_{j=1}^{M}\omega_j u_j^m=0\ ,\,0\le m\le q-1,\ \sum_{j=1}^{M}\omega_j u_j^q=(-1)^q\Bigr\},
$$
and that $(\omega_1,\dots,\omega_p)$ form affine coordinates on $\mathcal A_{p,q}(u)$. Differentiating $\omega_j$ with respect to $z_\alpha$ gives
$$
\frac{\partial\omega_j}{\partial z_\alpha}
=
\frac{\omega_j}{z_\alpha-u_j}.
$$
Therefore
$$
\det\!\left[\frac{\partial\omega_j}{\partial z_\alpha}\right]_{1\le j,\alpha\le p}
=
\Bigl(\prod_{j=1}^{p}\omega_j\Bigr)
\det\!\left[\frac{1}{z_\alpha-u_j}\right]_{1\le j,\alpha\le p},
$$
and the Cauchy determinant formula yields
$$
\det\!\left[\frac{1}{z_\alpha-u_j}\right]_{1\le j,\alpha\le p}
=
\frac{\prod_{1\le\alpha<\beta\le p}(z_\beta-z_\alpha)\,
\prod_{1\le i<j\le p}(u_i-u_j)}
{\prod_{\alpha=1}^{p}\prod_{j=1}^{p}(z_\alpha-u_j)}.
$$
Using $\prod_{\alpha=1}^{p}(z_\alpha-u_j)=\omega_j(-U'(u_j)),$
one obtains
\begin{equation}\label{eq:left-to-affine-safe}
\begin{aligned}
\int_{\C^p}
D_p(z)^{1|1}
\prod_{\alpha=1}^{p}\prod_{j=1}^{M}(z_\alpha-u_j)^{\lambda_j|\lambda_j'}
\d\nu_p(z)=&\frac{\prod_{j=1}^p(-U'(u_j))^{1+\lambda_j|1+\lambda_j'}\prod_{j=p+1}^M(-U'(u_j))^{\lambda_j|\lambda_j'}}{\prod_{1\leq i<j\leq p}(u_i-u_j)^{1|1}}
\\&\times
\frac{1}{\pi^p }
\int_{\A_{p,q}(u)}
\prod_{j=1}^{M}\omega_j^{\lambda_j|\lambda_j'}
\prod_{j=1}^{p}\d\omega_j.
\end{aligned}
\end{equation}

For $0\le m\le q$, set $L_m(\omega):=\sum_{j=1}^{M}\omega_j u_j^m.$ By construction of $\mathcal A_{p,q}(u)$, we have $L_m(\omega)=0,\,0\le m\le q-1,$ and $L_q(\omega)=(-1)^q.$ Let $Q_\tau(t):=\sum_{m=0}^{q}\tau_m t^m.$
Then for every $\omega\in\mathcal A_{p,q}(u)$,
$$
\sum_{m=0}^{q}\bigl(\tau_m\overline{L_m(\omega)}-\bar\tau_mL_m(\omega)\bigr)
=(-1)^q
(\tau_q-\bar\tau_q).
$$
On the other hand,
\begin{equation}\label{eq:Q_tau}
    \sum_{m=0}^{q}\bigl(\tau_m\overline{L_m(\omega)}-\bar\tau_mL_m(\omega)\bigr)
=
\sum_{j=1}^{M}\bigl(Q_\tau(\overline{u_j})\overline{\omega_j}-\overline{Q_\tau(\overline{u_j})}\,\omega_j\bigr).
\end{equation}

Therefore, for every $\omega\in\mathcal A_{p,q}(u)$,

\begin{equation}\label{eq:constraint-identity}
e^{(-1)^q(\tau_q-\bar\tau_q)}
=
\prod_{j=1}^{M}
e^{\,Q_\tau(\overline{u_j})\overline{\omega_j}-\overline{Q_\tau(\overline{u_j})}\,\omega_j}.
\end{equation}
If we define $$\Phi(y):=\frac{1}{\pi^p }
\int_{\{L_m(\omega)=y_m,\,0\leq m\leq q\}}
\prod_{j=1}^{M}\omega_j^{\lambda_j|\lambda_j'}
\prod_{j=1}^{p}\d\omega_j,\quad y=(y_0,\ldots,y_q)\in\C^{q+1},$$ then the Fourier transform $\hat{\Phi}$ is given by 
\begin{equation}\label{eq:Phi(0,...,0,)}
\begin{aligned}
    \hat{\Phi}(\tau):=&\int_{\C^{q+1}}\Phi(y)e^{\sum_{m=0}^q(\tau_m\bar{y}_m-\bar{\tau}_my_m)}\prod_{m=0}^q\d y_m\\
    =&\frac{\prod_{p+1\leq i<j\leq M}(u_j-u_i)^{1|1}}{\pi^p }\int_{\C^M}\prod^M_{j=1}\omega_j^{\lambda_j|\lambda_j'}e^{\sum_{m=0}^q(\tau_m\overline{L_m(\omega))}-\bar{\tau}_mL_m(\omega))}\prod^M_{j=1}\d\omega_j\\
    =&\frac{\prod_{p+1\leq i<j\leq M}(u_j-u_i)^{1|1}}{\pi^p }\prod^M_{j=1}\int_\C\omega^{\lambda_j|\lambda_j'}_je^{Q_\tau(\overline{u_j})\overline{\omega_j}-\overline{Q_\tau(\overline{u_j})}\omega_j}\d\omega_j\\
=&\pi^{q+1}\prod_{p+1\leq i<j\leq M}(u_j-u_i)^{1|1}\prod_{j=1}^M\Gamma^\C(1+\lambda_j|1+\lambda_j')Q_\tau(\overline{u_j})^{-1-\lambda_j'|-1-\lambda_j},
\end{aligned}
\end{equation}
where, in the second equality, we have used the change of variables $y_j=L_j(\omega)$ for $j=0,\ldots,q$ with $(\omega_1,\ldots\omega_p)$ fixed, the penultimate equality follows from \eqref{eq:Q_tau}, and the last equality follows from \eqref{eq:one-variable-Fourier}. Now the Fourier inversion formula gives 

\begin{align*}
    \Phi(0,\ldots,0,(-1)^q&)=\frac{1}{\pi^{2(q+1)}}\int_{\C^{q+1}}\hat{\Phi}(\tau)e^{(-1)^{q+1}(\tau_q-\bar{\tau}_q)}\prod^q_{m=0}\d\tau_m\\
=&\frac{\prod_{p+1\leq i< j\leq M}(u_j-u_i)^{1|1}}{\pi^{q+1}}\int_{\C^{q+1}}e^{(-1)^{q+1}(\tau_q-\bar{\tau}_q)}\prod^M_{j=1}Q_\tau(\overline{u_j})^{-1-\lambda_j'|-1-\lambda_j}\prod^q_{m=0}\d\tau_m.\end{align*}
 Writing $Q_\tau(t)=c\prod_{\beta=1}^{q}(t-\overline{w_\beta})$, we have $\tau_q=c$ and \begin{equation}\label{eq:Q_tau(u_j)}
     \begin{aligned}
         Q_\tau(\overline{u_j})^{-1-\lambda'_j|-1-\lambda_j}=&c^{-1-\lambda'_j|-1-\lambda_j}\prod^q_{\beta=1}(u_j-w_\beta)^{-1-\lambda_j|-1-\lambda_j'}\\=&(-1)^{q\sum^M_{j=1}(\lambda_j'-\lambda_j)}c^{-1-\lambda'_j|-1-\lambda_j}\prod^q_{\beta=1}(w_\beta-u_j)^{-1-\lambda_j|-1-\lambda_j'}\end{aligned}\end{equation} The usual coefficient-to-roots Jacobian reads
$$
\prod_{m=0}^{q}\d\tau_m
=\frac{1}{q!}
|c|^{2q}D_q(w)^{1|1}\d c\prod_{\beta=1}^q\d w_\beta.
$$
Combining this, \eqref{eq:left-to-affine-safe}, \eqref{eq:Phi(0,...,0,)}, and \eqref{eq:Q_tau(u_j)}, we obtain that the left-hand side of \eqref{eq:complex-twin-A7-final-statement} is equal to the sign $(-1)^{q\sum^M_{j=1}(\lambda_j'-\lambda_j)}$ times
$$
\begin{aligned}
&\frac{\prod_{j=1}^p(-U'(u_j))^{1+\lambda_j|1+\lambda_j'}\prod_{j=p+1}^M(-U'(u_j))^{\lambda_j|\lambda_j'}}{\prod_{1\leq i<j\leq p}(u_i-u_j)^{1|1}}\prod_{p+1\leq i<j\leq M}(u_j-u_i)^{1|1}
\\&\times\prod_{j=1}^{M}\Gamma^{\C}(1+\lambda_j\,|\,1+\lambda_j')\int_{\C}e^{(-1)^{q+1}(c-\bar c)}
c^{-p-1-\sum_j\lambda_j\,|\,-p-1-\sum_j\lambda_j'}\d\nu_1(c) \times I,\end{aligned}
$$
where $I$ is the integral on the right-hand side of \eqref{lem:complex-twin-A7-final}.
Applying \eqref{eq:one-variable-Fourier} to the $c$-integral gives
$$\begin{aligned}
\int_{\C}
e^{(-1)^{q+1}(c-\bar c)}
c^{-p-1-\sum_j\lambda_j\,|\,-p-1-\sum_j\lambda_j'}\d \nu_1(c)
=&\frac{(-1)^{q\sum^M_{j=1}(\lambda_j-\lambda_j')}}{\Gamma^{\C}\!\left(
p+1+\sum_{j=1}^{M}\lambda_j
\ \Big|\
p+1+\sum_{j=1}^{M}\lambda_j'
\right)},
\end{aligned}$$
and so the sign cancels. 
Finally, one checks that $$\begin{aligned}&\frac{\prod_{j=1}^p(-U'(u_j))^{1+\lambda_j|1+\lambda_j'}\prod_{j=p+1}^M(-U'(u_j))^{\lambda_j|\lambda_j'}}{\prod_{1\leq i<j\leq p}(u_i-u_j)^{1|1}}\prod_{p+1\leq i<j\leq M}(u_j-u_i)^{1|1}\\&=(-1)^{\sum^M_{j=1}(j-1)(\lambda_j-\lambda_j')}\prod_{1\leq i<j\leq M}(u_j-u_i)^{1+\lambda_i+\lambda_j|1+\lambda_i'+\lambda_j'}\end{aligned}$$ This completes the proof.
\end{proof}

For later applications of Lemma~\ref{lem:complex-twin-A7-final}, we will need to evaluate both sides of \eqref{eq:complex-twin-A7-final-statement} for parameters outside the common absolute-convergence region \eqref{eq:Omega-conv-final}. We therefore establish a general meromorphic continuation theorem for complex local zeta integrals and then apply it to the present setting.

The argument follows the method of Atiyah \cite{atiyah1970resolution}, namely to reduce the integral, after resolution of singularities, to finitely many monomial Mellin-type integrals. For the resolution input, we use principalization of ideal sheaves in the complex-analytic category \cite{Hironaka,wlodarczyk2008resolution}. As a consequence, the polynomial case on $\C^n$ follows immediately, and this will imply the meromorphic continuation of both sides of \eqref{eq:complex-twin-A7-final-statement}.

\begin{theorem}[Meromorphic continuation of complex local zeta integrals]\label{thm:local-complex-zeta}
Let $X$ be a complex manifold of complex dimension $n$, let $f_1,\dots,f_m$ be holomorphic functions on $X,$ and let $\omega$ be a compactly supported $C^\infty$ top-degree form on $X$.
For $\blambda=((\lambda_1,\lambda_1'),\dots,(\lambda_m,\lambda_m'))\in \Lambda_{\C}^m,$ define
$$
Z(\blambda):=
\int_X \prod_{j=1}^m f_j(x)^{\lambda_j|\lambda_j'}\,\omega(x),
$$
whenever the integral converges absolutely. For each connected component of $\Lambda_{\C}^m$ on which the absolute-convergence
domain of $Z(\blambda)$ has non-empty interior, the function $Z(\blambda)$ admits a meromorphic
continuation to the whole component.
\end{theorem}

\begin{proof}
We work on one connected component of $\Lambda_{\C}^m$, so that each difference $\lambda_j-\lambda_j'\in\Z,\,1\le j\le m$ is fixed. Choose finitely many holomorphic coordinate charts ($U_\alpha$) covering the support of $\omega$, together with a smooth partition of unity $(\chi_\alpha)$ subordinate to $(U_\alpha)$. Writing $\omega=\sum_\alpha \chi_\alpha\omega,$ it is enough to prove meromorphic continuation for each local term. Thus we may assume $X\subset \C^n$ is an open set with coordinates $z=(z_1,\dots,z_n)$, and $\omega=\varphi(z)\d^{2n}z,\,\varphi\in C_c^\infty(X).$ We are reduced to
$$
Z(\blambda)=\int_X \varphi(z)\prod_{j=1}^m f_j(z)^{\lambda_j|\lambda_j'}\d^{2n}z.
$$

We follow Atiyah's argument \cite{atiyah1970resolution}. Set $D:=\bigcup_{j=1}^m \{f_j=0\}.$ By embedded resolution of singularities (for example, \cite[Theorem 2.0.3]{wlodarczyk2008resolution}), there exists a proper holomorphic map $\pi:Y\to X$
such that $Y$ is smooth, $\pi$ is biholomorphic over $X\setminus D$, and in local coordinates $y=(y_1,\dots,y_n)$ on $Y$, each pullback $f_j\circ\pi$ has the form
$$
f_j\circ\pi(y)=u_j(y)\,y_1^{a_{1j}}\cdots y_n^{a_{nj}},
$$
where $u_j$ is nowhere vanishing holomorphic and $a_{kj}\in\N_0$. Moreover,
$$
\pi^*(\d^{2n}z)=v(y)\prod_{k=1}^n |y_k|^{2b_k}\d^{2n}y,
$$
where $v$ is smooth and nowhere vanishing, and $b_k\in\N_0$.

Choose a smooth partition of unity on $\pi^{-1}(\mathrm{supp}\,\varphi)$. Since $\pi$ is proper and $\mathrm{supp}\,\varphi$ is compact, only finitely many coordinate charts on $Y$ are needed. Thus $Z(\blambda)$ is a finite sum of local terms of the form
\begin{equation}\label{eq:local-zeta-after-resolution}
\int_{K}
\psi(y)\,
\prod_{j=1}^m \bigl(f_j\circ \pi(y)\bigr)^{\lambda_j|\lambda_j'}
\prod_{k=1}^n |y_k|^{2b_k}\,\d^{2n}y,
\end{equation}
where $K\subset \C^n$ is compact and $\psi\in C_c^\infty(K)$.

Substituting the monomial form of $f_j\circ\pi$, we get
$$
\prod_{j=1}^m \bigl(f_j\circ\pi(y)\bigr)^{\lambda_j|\lambda_j'}
=
\prod_{j=1}^m u_j(y)^{\lambda_j|\lambda_j'}
\prod_{k=1}^n
y_k^{\sum_{j=1}^m a_{kj}\lambda_j \,\big|\, \sum_{j=1}^m a_{kj}\lambda_j'}.
$$
Since each $u_j$ is nowhere vanishing and holomorphic, $u_j(y)^{\lambda_j|\lambda_j'}$ is smooth in $y$ and entire in $(\lambda_j,\lambda_j')$.
Therefore each local term \eqref{eq:local-zeta-after-resolution} is of the form
\begin{equation}\label{eq:local-monomial-integral}
I(\blambda)
=
\int_{K}
\Psi(y,\blambda)\,
\prod_{k=1}^n y_k^{A_k(\blambda)\,|\,A_k'(\blambda)}
\,\d^{2n}y,
\end{equation}
where $K\subset \C^n$ is compact, $\Psi$ is smooth in $y$ with compact support contained in $K$ and holomorphic in $\blambda$, and $A_k(\blambda),\,A'_k(\blambda)$ are affine linear functions of $\blambda,$ with the difference $A_k(\blambda)-A'_k(\blambda)\in\Z$ constant on the chosen connected component of $\Lambda_{\C}^m$.

It remains to prove that \eqref{eq:local-monomial-integral} admits meromorphic continuation in $\blambda$. We do this by induction on $n$.
For $n=1$, writing $y=re^{i\theta}$, we have
$$
I(\blambda)
=
\int_0^\infty r^{A(\blambda)+A'(\blambda)+1}
\left(\int_0^{2\pi}\Psi(re^{i\theta},\blambda)e^{i\theta(A(\blambda)-A'(\blambda))}\,\d\theta\right)\d r.
$$
Let $\Psi^\sharp$ denote the raidal part of the integral, i.e.,
$$
\Psi^\sharp(r,\blambda):=
\int_0^{2\pi}\Psi(re^{\bi\theta},\blambda)e^{\bi\theta(A(\blambda)-A'(\blambda))}\,\d\theta.
$$
Since $A(\blambda)-A'(\blambda)\in\Z$ is fixed and $\Psi$ is $C_c^\infty$ in $y$ and holomorphic in $\blambda$, it follows that the function $\Psi^\sharp(\cdot,\blambda)$ is $C_c^\infty([0,\infty))$ and holomorphic in $\blambda$. Expanding it near $r=0$:
$$
\Psi^\sharp(r,\blambda)=\sum_{\nu=0}^{N-1} c_\nu(\blambda)\,r^\nu+r^N R_N(r,\blambda),
$$
where $c_\nu(\blambda)$ and $R_N(r,\blambda)$ are holomorphic in $\blambda$, and $R_N(\cdot,\blambda)$ is $C_c^\infty$ in $r$, we have
$$
I(\blambda)
=
\sum_{\nu=0}^{N-1}\frac{c_\nu(\blambda)}{A(\blambda)+A'(\blambda)+\nu+2}
+
\int_0^\infty r^{A(\blambda)+A'(\blambda)+N+1}R_N(r,\blambda)\,\d r.
$$
The remainder is holomorphic whenever $\Re\bigl(A(\blambda)+A'(\blambda)\bigr)>-N-2.$ Since $N$ is arbitrary, $I(\blambda)$ admits meromorphic continuation, with poles contained in the hyperplanes $A(\blambda)+A'(\blambda)\in -2-\N_0.$

Now assume the claim proved in dimension $n-1$, and consider
$$
I(\blambda)
=
\int_{K}
\Psi(y,\blambda)\,
\prod_{k=1}^n y_k^{A_k(\blambda)\,|\,A_k'(\blambda)}
\d^{2n}y.
$$
Write $y=(y_1,y') $ and $ y'=(y_2,\dots,y_n).$ In the polar coordinates $y_1=re^{\bi\theta},$ we have
$$
I(\blambda)
=
\int_{\C^{n-1}}
\left[
\int_0^\infty
r^{A_1(\blambda)+A_1'(\blambda)+1}
\Phi(r,y',\blambda)\,dr
\right]
\prod_{k=2}^n y_k^{A_k(\blambda)\,|\,A_k'(\blambda)}
\d^{2n-2}y',
$$
where
$$
\Phi(r,y',\blambda):=
\int_0^{2\pi}
\Psi(re^{\bi\theta},y',\blambda)\,
e^{\bi\theta(A_1(\blambda)-A_1'(\blambda))}\,d\theta.
$$
As in the base case, $\Phi(\cdot,y',\blambda)$ is $C_c^\infty([0,\infty))$ in $r$, compactly supported uniformly in $y'$, and holomorphic in $\blambda$. Expanding at $r=0$,
$$
\Phi(r,y',\blambda)
=
\sum_{\nu=0}^{N-1} c_\nu(y',\blambda)\,r^\nu
+
r^N R_N(r,y',\blambda),
$$
we obtain $I(\blambda)=\sum_{\nu=0}^{N-1} I_\nu(\blambda)+I_N^{\mathrm{rem}}(\blambda),$
where each $I_\nu(\blambda)$ is an integral in the $n-1$ variables $y'$ of the same type as the original one, with coefficient
$$
\frac{c_\nu(y',\blambda)\,R^{A_1(\blambda)+A_1'(\blambda)+\nu+2}}
{A_1(\blambda)+A_1'(\blambda)+\nu+2},
$$
and the remainder $I_N^{\mathrm{rem}}(\blambda)$ takes the form
$$
I_N^{\mathrm{rem}}(\blambda)
=
\int_{\C^{n-1}}
G_N(y',\blambda)\,
\prod_{k=2}^n y_k^{A_k(\blambda)\,|\,A_k'(\blambda)}
\d^{2n-2}y',
$$
with$$
G_N(y',\blambda):=\int_0^R r^{A_1(\blambda)+A_1'(\blambda)+N+1}R_N(r,y',\blambda)\d r,
$$
which is $C_c^\infty$ in $y'$ and holomorphic in $\blambda$ whenever
$\Re\bigl(A_1(\blambda)+A_1'(\blambda)\bigr)>-N-2.$

By the induction hypothesis, each $I_\nu(\blambda)$ admits meromorphic continuation. Since $N$ is arbitrary, the remainders are holomorphic on half-spaces exhausting the whole parameter space. Hence $I(\blambda)$ admits meromorphic continuation to all $\blambda$. This completes the proof.
\end{proof}

\begin{corollary}\label{cor:global-polynomial-zeta}
Let $f_1,\dots,f_m\in\C[x_1,\dots,x_n]$ be polynomials and set
$$
Z(\blambda):=
\int_{\C^n}\prod_{j=1}^m f_j(x)^{\lambda_j|\lambda_j'}\d^{2n}x .
$$
On each connected component of $\Lambda_{\C}^m$ on which the domain of
absolute convergence has non-empty interior, $Z(\blambda)$ admits a meromorphic continuation to the whole component.

\end{corollary}

\begin{proof}
We prove the statement componentwise in $\Lambda_{\C}^m$, keeping the
integers $\lambda_j-\lambda_j'$ fixed. Let $[X_0:\cdots:X_n]$ be homogeneous coordinates on $\P^n$, and identify
$\C^n$ with the affine chart $\{X_0\neq0\}$ by
$x_i=\frac{X_i}{X_0},\, i=1,\ldots,n.$ For each $j$, let $d_j:=\deg f_j$, and let $F_j$ be the homogenization of
$f_j$, so that
$$
F_j(X_0,\ldots,X_n)
=
X_0^{d_j}
f_j\!\left(\frac{X_1}{X_0},\ldots,\frac{X_n}{X_0}\right).
$$

We now describe what the original integral becomes after compactification. In the affine chart $U_\nu:=\{X_\nu\neq0\}$, $\nu\ge1$, set $y_i:=\frac{X_i}{X_\nu},\, i\neq \nu.$ On
$U_\nu\cap\{y_0\neq0\}$, the original affine coordinates are
$x_\nu=\frac{1}{y_0},
\,
x_i=\frac{y_i}{y_0}\quad (i\neq \nu,\ i\ge1),$ and
the complex Jacobian of this change of variables gives
$$
\d^{2n}x
=
|y_0|^{-2(n+1)}\d^{2n}y
=
y_0^{-(n+1)\,|\,-(n+1)}\d^{2n}y .
$$
Moreover, if $F_{j,\nu}(y)
:=
F_j(y_0,\ldots,y_{\nu-1},1,y_{\nu+1},\ldots,y_n),$
then $f_j(x)=y_0^{-d_j}F_{j,\nu}(y),$ and so
$f_j(x)^{\lambda_j|\lambda_j'}
=
F_{j,\nu}(y)^{\lambda_j|\lambda_j'}
y_0^{-d_j\lambda_j\,|\,-d_j\lambda_j'}.$
Combining these identities, the integrand in $U_\nu$ is
$$
\prod_{j=1}^m f_j(x)^{\lambda_j|\lambda_j'}\d^{2n}x
=
\prod_{j=1}^m
F_{j,\nu}(y)^{\lambda_j|\lambda_j'}
\,
y_0^{\mu(\blambda)\,|\,\mu'(\blambda)}
\d^{2n}y,
$$
where
$$
\mu(\blambda):=-\sum_{j=1}^m d_j\lambda_j-(n+1),
\qquad
\mu'(\blambda):=-\sum_{j=1}^m d_j\lambda_j'-(n+1).
$$
On the affine chart $U_0=\{X_0\neq0\}$, there is no additional infinity
factor: $X_0=1$, and the integrand remains $\prod_{j=1}^m f_j(x)^{\lambda_j|\lambda_j'}\d^{2n}x.$

Let $\Omega$ denote the set of
parameters in the connected component under consideration for which the
original integral over $\C^n$ is absolutely convergent, and assume that
$\Omega$ has non-empty interior. For $\blambda\in\Omega$, the change of
variables above is legitimate and the integral may be computed after
compactifying $\C^n$ to $\P^n$.

Choose a smooth partition of unity $(\chi_\nu)_{\nu=0}^n$ on $\P^n$
subordinate to the affine cover $(U_\nu)_{\nu=0}^n$. For $\blambda\in\Omega$, we have $Z(\blambda)=\sum_{\nu=0}^n Z_\nu(\blambda).$
For $\nu=0$, the localized contribution is
$$
Z_0(\blambda)
=
\int_{U_0}
\chi_0(x)
\prod_{j=1}^m f_j(x)^{\lambda_j|\lambda_j'}
\,\d^{2n}x,
$$and Theorem~\ref{thm:local-complex-zeta} gives a meromorphic continuation of
$Z_0$ to the whole connected component. On the other hand, if $\nu\ge1$, the above computation above gives
$$
Z_\nu(\blambda)
=
\int_{U_\nu}
\chi_\nu(y)
\prod_{j=1}^m
F_{j,\nu}(y)^{\lambda_j|\lambda_j'}
\,
y_0^{\mu(\blambda)|\mu'(\blambda)}
\,\d^{2n}y,
$$
and again Theorem~\ref{thm:local-complex-zeta} gives its meromorphic continuation to the same component. This proves the claim.
\end{proof}

\subsection{Proof of Proposition \ref{prop:closed formula}}
Let $\mathcal{I}_{\bs}(\sigma,\tau,B_1,B_1',\ldots, B_r,B_r')$ denote the integral
$$\begin{aligned}
\int_{\C^N}
\prod_{a=1}^{s_r}(x_a^{(r)})^{\sigma|\sigma}
\prod_{i=1}^{r}\prod_{a=1}^{s_i}(&1-x_a^{(i)})^{B_i|B_i'}
\prod_{i=1}^{r}D_{s_i}(x^{(i)})^{\tau|\tau}
\\&\times\prod_{i=1}^{r-1}\prod_{a=1}^{s_i}\prod_{b=1}^{s_{i+1}}
(x_b^{(i+1)}-x_a^{(i)})^{-\tau/2|-\tau/2}
\prod_{i=1}^{r}\d\nu_{s_i}(x^{(i)}),\end{aligned}
$$ where $$B_i=\frac{\gamma}{2}\la e_i,\alpha_2+m_2\ra,\quad B'_i=\frac{\gamma}{2}\la e_i,\alpha_2-m_2\ra.$$
We recall that the Cartan matrix of the lie algebra $\mathfrak{sl}_{r+1}$ is given by \begin{equation*}
    A_{ij}=\begin{cases}
        2 &\text{if }i=j,\\
        -1 &\text{if } |i-j|=1,\\
        0 &\text{otherwise},
    \end{cases}
\end{equation*}
and so the $A_r$-Dotsenko--Fateev integral \eqref{eq:DF-0-1-infty} becomes 
\begin{equation}\label{eq:to be proved}
\begin{aligned}
 \prod_{i=1}^r(\pi^{s_i} s_i!)\mathcal{I}_{\bs}\left(\frac{\kappa\gamma}{2},\gamma^2, B_1,B'_1,\ldots,B_r,B_r'\right).
\end{aligned}
\end{equation}
It remains to show that the square of the integral in \eqref{eq:to be proved} coincides with the product of the imaginary Fateev--Litivinov formula as showcased in \eqref{eq:FL with magnetic charges}.

We follow the argument in \cite[Appendix]{fateev2007correlation} with Proposition \ref{lem:complex-twin-A7-final} in place of \cite[eq. (A.7)]{fateev2007correlation} to first obtain a recurrent relation \eqref{eq:recurrent relation} and then apply the shift equation \eqref{eq:shift} of the special function $\Upsilon$ to conclude the proof.

The applications of Lemma~\ref{lem:complex-twin-A7-final} below are first
carried out in a non-empty auxiliary open subset of the parameter space
$(\sigma,\tau,B_1,B_1',\ldots,B_r,B_r')\in \C^2\times\Lambda_{\C}^r,$
where $\sigma$ and $\tau$ are kept diagonal, i.e. they correspond to the exponents $\sigma|\sigma$ and $\tau|\tau$. On this subset all
intermediate integrals are absolutely convergent, and the applications of
Lemma~\ref{lem:complex-twin-A7-final} give the recurrence relation
\eqref{eq:recurrent relation}. By
Corollary~\ref{cor:global-polynomial-zeta}, both sides of this recurrence
relation admit componentwise meromorphic continuations to
$\C^2\times\Lambda_{\C}^r$. Hence the recurrence relation, and therefore its iteration, holds meromorphically. Only after the full iterated formula has been obtained do we specialize to the Toda parameters. At exceptional parameter values, the identity is understood by meromorphic continuation.

Although in Steps~1--4 below we set $\sigma=\frac{\kappa\gamma}{2},\,\tau=\gamma^2$, this is only a notational simplification. The argument should be understood as the preceding meromorphic-continuation argument with
$\sigma$ and $\tau$ first treated as independent diagonal complex-field
parameters. The specialization $\sigma=\frac{\kappa\gamma}{2},
\,\tau=\gamma^2$ is imposed only after the recurrence and its iteration have been established as meromorphic identities.

\paragraph{\textbf{Step 1: Removing the row $x^{(1)}$}}
We begin by rewriting the Vandermonde factor in the first row. Lemma \ref{lem:complex-twin-A7-final} with $p=s_1-1,\,q=0,$ and $\lambda_j=\gamma^2/2-1$ implies that 

\begin{equation}\label{eq:first-row-m0}
D_{s_1}(x^{(1)})^{\gamma^2-1|\gamma^2-1}
=
C_0
\int_{\C^{s_1-1}}
D_{s_1-1}(y^{(1)})^{1|1}
\prod_{\beta=1}^{s_1-1}\prod_{a=1}^{s_1}
(y_\beta^{(1)}-x_a^{(1)})^{\frac{\gamma^2}{2}-1|\frac{\gamma^2}{2}-1}
\d\nu_{s_1-1}(y^{(1)}),
\end{equation}
with
\begin{equation}\label{eq:c1}
C_0=
\frac{
\Gamma^{\C}\!\left(\frac{s_1\gamma^2}{2}\,\Big|\,\frac{s_1\gamma^2}{2}\right)
}{
\Gamma^{\C}\!\left(\frac{\gamma^2}{2}\,\Big|\,\frac{\gamma^2}{2}\right)^{s_1}
}=\frac{l\left(\frac{s_1\gamma^2}{2}\right)}{l\left(\frac{\gamma^2}{2}\right)^{s_1}},
\end{equation}
where we have used \eqref{eq:complex Gamma} and $l(x)=\Gamma(x)/\Gamma(1-x)$. We note that the sign in \eqref{eq:complex-twin-A7-final-statement} is positive since there is only one exponent $\lambda_j|\lambda_j'$ that is not of the form $\lambda|\lambda$, and so the exponent $\sum_{j=1}^M(j-1)(\lambda_j-\lambda_j')=0$. This argument applies to all the applications of Lemma \ref{lem:complex-twin-A7-final} below.

Using 
$$D_{s_1}(x^{(1)})^{\gamma^2|\gamma^2}=D_{s_1}(x^{(1)})^{\gamma^2-1|\gamma^2-1}D_{s_1}(x^{(1)})^{1|1}$$ and
substituting \eqref{eq:first-row-m0} into the integral in \eqref{eq:to be proved}, we isolate the remaining $x^{(1)}$-integral:
\begin{equation}\label{eq:J1-definition}
\begin{aligned}
J_1(y^{(1)},x^{(2)})
:=
\int_{\C^{s_1}}
& D_{s_1}(x^{(1)})^{1|1}
\prod_{a=1}^{s_1}(1-x_a^{(1)})^{B_1|B_1'}
\\
&\times
\prod_{\beta=1}^{s_1-1}\prod_{a=1}^{s_1}
(y_\beta^{(1)}-x_a^{(1)})^{\frac{\gamma^2}{2}-1|\frac{\gamma^2}{2}-1}
\prod_{a=1}^{s_1}\prod_{b=1}^{s_2}
(x_b^{(2)}-x_a^{(1)})^{-\frac{\gamma^2}{2}|-\frac{\gamma^2}{2}}
\d\nu_{s_1}(x^{(1)}).
\end{aligned}
\end{equation}

Now we apply Lemma~\ref{lem:complex-twin-A7-final} to $J_1$, with external points $1,\ y_1^{(1)},\dots,y_{s_1-1}^{(1)},\ x_1^{(2)},\dots,x_{s_2}^{(2)},$
and exponents
$$
B_1|B_1',
\qquad
\underbrace{\frac{\gamma^2}{2}-1|\frac{\gamma^2}{2}-1,\dots,\frac{\gamma^2}{2}-1|\frac{\gamma^2}{2}-1}_{s_1-1\text{ times}},
\qquad
\underbrace{-\frac{\gamma^2}{2}|-\frac{\gamma^2}{2},\dots,-\frac{\gamma^2}{2}|-\frac{\gamma^2}{2}}_{s_2\text{ times}}
$$
to obtain

\begin{equation}\label{eq:J1}
\begin{aligned}
J_1=&C_1 
\prod_{\beta=1}^{s_1-1}
(1-y_\beta^{(1)})^{B_1+\frac{\gamma^2}{2}|B_1'+\frac{\gamma^2}{2}}
D_{s_1-1}(y^{(1)})^{\gamma^2-1|\gamma^2-1}
\prod_{a=1}^{s_2}
(1-x_a^{(2)})^{1+B_1-\frac{\gamma^2}{2}|1+B_1'-\frac{\gamma^2}{2}}
\\
&\times D_{s_2}(x^{(2)})^{1-\gamma^2|1-\gamma^2}
\int_{\C^{s_2-1}}
\prod_{\delta=1}^{s_2-1}
(1-y_\delta^{(2)})^{-1-B_1\,|\,-1-B_1'}
D_{s_2-1}(y^{(2)})^{1|1}
\\&\prod_{\beta=1}^{s_1-1}\prod_{\delta=1}^{s_2-1}
(y_\delta^{(2)}-y_\beta^{(1)})^{-\frac{\gamma^2}{2}|-\frac{\gamma^2}{2}}
\times\prod_{\delta=1}^{s_2-1}\prod_{a=1}^{s_2}
(x_a^{(2)}-y_\delta^{(2)})^{-1+\frac{\gamma^2}{2}|-1+\frac{\gamma^2}{2}}
\d\nu_{s_2-1}(y^{(2)}),
\end{aligned}
\end{equation}
with the constant $C_1$ given by
\begin{equation}\label{eq:C1}
\begin{aligned}
    C_1=&
\frac{
\Gamma^{\C}(1+B_1\,|\,1+B_1')\,
\Gamma^{\C}(\frac{\gamma^2}{2}|\frac{\gamma^2}{2})^{s_1-1}\,
\Gamma^{\C}\!\left(1-\frac{\gamma^2}{2}\Big|1-\frac{\gamma^2}{2}\right)^{s_2}
}{
\Gamma^{\C}\!\left(
s_1+1+B_1+(s_1-1)(\frac{\gamma^2}{2}-1)-\frac{\gamma^2}{2}s_2
\Big|\
s_1+1+B_1'+(s_1-1)(\frac{\gamma^2}{2}-1)-\frac{\gamma^2}{2}s_2
\right)
}\\
=&\frac{\Gamma^{\C}(1+B_1\,|\,1+B_1')
l\left(\frac{\gamma^2}{2}\right)^{s_1-1}l\left(1-\frac{\gamma^2}{2}\right)^{s_2}
}{\Gamma^{\C}\!\left(
s_1+1+B_1+(s_1-1)(\frac{\gamma^2}{2}-1)-\frac{\gamma^2}{2}s_2
\Big|\
s_1+1+B_1'+(s_1-1)(\frac{\gamma^2}{2}-1)-\frac{\gamma^2}{2}s_2
\right)
}
\end{aligned}
\end{equation}
We note that the term $D_{s_1-1}(y^{(1)})$ in \eqref{eq:first-row-m0} and \eqref{eq:C1} are combined to have exponent $\gamma^2|\gamma^2.$

\paragraph{\textbf{Step 2: Removing the row $x^{(2)}$}}
We observe that the $x^{(2)}$-dependent part becomes
\begin{equation*}
\begin{aligned}
J_2(y^{(2)},x^{(3)})
:=
\int_{\C^{s_2}}
& D_{s_2}(x^{(2)})^{1|1}
\prod_{a=1}^{s_2}(1-x_a^{(2)})^{Q_2|Q_2'}
\\
&\times
\prod_{\delta=1}^{s_2-1}\prod_{a=1}^{s_2}
(x_a^{(2)}-y_\delta^{(2)})^{-1+\frac{\gamma^2}{2}|-1+\frac{\gamma^2}{2}}
\prod_{a=1}^{s_2}\prod_{b=1}^{s_3}
(x_b^{(3)}-x_a^{(2)})^{-\frac{\gamma^2}{2}|-\frac{\gamma^2}{2}}
\d\nu_{s_2}(x^{(2)}),
\end{aligned}
\end{equation*}
with $$
Q_2:=1+B_1+B_2-\frac{\gamma^2}{2},
\qquad
Q_2':=1+B_1'+B_2'-\frac{\gamma^2}{2}.$$
Integrating out the second row $x^{(2)}$ using Lemma \ref{lem:complex-twin-A7-final} yields
\begin{equation}\label{eq:J2}
    \begin{aligned}
   J_2=& C_2
\prod_{a=1}^{s_2-1}
(1-y_a^{(2)})^{1+B_1+B_2\,|\,1+B_1'+B_2'}
D_{s_2-1}(y^{(2)})^{-1+\gamma^2|-1+\gamma^2}
\\
&\times
\prod_{b=1}^{s_3}
(1-x_b^{(3)})^{2+B_1+B_2-\gamma^2|2+B_1'+B_2'-\gamma^2}
D_{s_3}(x^{(3)})^{1-\gamma^2|1-\gamma^2}
\\
&\times\int_{\C^{s_3-1}}
\prod_{\delta=1}^{s_3-1}
(1-y_\delta^{(3)})^{-2-B_1-B_2+\frac{\gamma^2}{2}|-2-B_1'-B_2'+\frac{\gamma^2}{2}}
D_{s_3-1}(y^{(3)})^{1|1}
\\
&\times
\prod_{\eta=1}^{s_2-1}\prod_{\delta=1}^{s_3-1}
(y_\delta^{(3)}-y_\eta^{(2)})^{-\frac{\gamma^2}{2}|-\frac{\gamma^2}{2}}
\prod_{\delta=1}^{s_3-1}\prod_{a=1}^{s_3}
(x_a^{(3)}-y_\delta^{(3)})^{-1+\frac{\gamma^2}{2}|-1+\frac{\gamma^2}{2}}
\d\nu_{s_3-1}(y^{(3)}),
    \end{aligned}
\end{equation}
with the constant $C_2$ given by
\begin{equation}\label{eq:C2}
C_2=
\frac{
\Gamma^{\C}(1+Q_2\,|\,1+Q_2')l\left(\frac{\gamma^2}{2}\right)^{s_2-1}
l\left(1-\frac{\gamma^2}{2}\right)^{s_3}
}{
\Gamma^{\C}\!\left(
s_2+1+Q_2+(s_2-1)(\frac{\gamma^2}{2}-1)-\frac{\gamma^2}{2}s_3
\ \Big|\
s_2+1+Q_2'+(s_2-1)(\frac{\gamma^2}{2}-1)-\frac{\gamma^2}{2}s_3
\right)
}.
\end{equation}
We note that the terms $(1-y^{(2)}_a)$ and $D_{s_2-1}(y^{(2)})$ in \eqref{eq:J1} and \eqref{eq:J2} can be combined to have exponents $B_2|B_2'$ and $\gamma^2|\gamma^2,$ respectively. 

\paragraph{\textbf{Step 3: Removing the row $x^{k},\,3\leq k\leq r-1$}}
We define recursively
\begin{equation}\label{eq:Qk-recursion}
Q_{k+1}=1+Q_k+B_{k+1}-\frac{\gamma^2}{2},
\qquad
Q_{k+1}'=1+Q_k'+B_{k+1}'-\frac{\gamma^2}{2},
\qquad k\ge 2,
\end{equation} and thus
\begin{equation}\label{eq:Qk-explicit}
Q_k=(k-1)+\sum_{j=1}^k B_j-(k-1)\frac{\gamma^2}{2},
\qquad
Q_k'=(k-1)+\sum_{j=1}^k B_j'-(k-1)\frac{\gamma^2}{2}.
\end{equation}
For $3\le k\le r-1$, the $x^{(k)}$-dependent part is
\begin{equation}\label{eq:Jk}
\begin{aligned}
J_k(y^{(k)},x^{(k+1)})
:=
&\int_{\C^{s_k}}
 D_{s_k}(x^{(k)})^{1|1}
\prod_{a=1}^{s_k}(1-x_a^{(k)})^{Q_k|Q_k'}
\\
&\times
\prod_{\alpha=1}^{s_k-1}\prod_{a=1}^{s_k}
(x_a^{(k)}-y_\alpha^{(k)})^{-1+\frac{\gamma^2}{2}|-1+\frac{\gamma^2}{2}}
\prod_{a=1}^{s_k}\prod_{b=1}^{s_{k+1}}
(x_b^{(k+1)}-x_a^{(k)})^{-\frac{\gamma^2}{2}|-\frac{\gamma^2}{2}}
\d\nu_{s_k}(x^{(k)}).
\end{aligned}
\end{equation}

Applying Lemma~\ref{lem:complex-twin-A7-final} to $J_k$, we obtain a new row $y^{(k+1)}=(y_1^{(k+1)},\dots,y_{s_{k+1}-1}^{(k+1)}),$
and the constant
\begin{equation}\label{eq:Ck}
C_k=
\frac{
\Gamma^{\C}(1+Q_k\,|\,1+Q_k')l\left(\frac{\gamma^2}{2}\right)^{s_k-1}
l\left(1-\frac{\gamma^2}{2}\right)^{s_{k+1}}
}{
\Gamma^{\C}\!\left(
s_k+1+Q_k+(s_k-1)(\frac{\gamma^2}{2}-1)-\frac{\gamma^2}{2}s_{k+1}
\ \Big|\
s_k+1+Q_k'+(s_k-1)(\frac{\gamma^2}{2}-1)-\frac{\gamma^2}{2}s_{k+1}
\right)
}.
\end{equation}

We note that the terms $(1-y^{(k)}_a)$ and $D_{s_k-1}(y^{(k)})$ have exponents $B_k|B_k'$ and $\gamma^2|\gamma^2,$ respectively. 

\paragraph{\textbf{Step 4: Removing the row $x^{(r)}$} }
Finally the $x^{(r)}$-dependent term takes the form
\begin{equation}\label{eq:Jr}
\begin{aligned}
J_r(y^{(r)})
:=
\int_{\C^{s_r}}
 D_{s_r}(x^{(r)})^{1|1}
&\prod_{a=1}^{s_r}(x_a^{(r)})^{\frac{\kappa\gamma}{2}|\frac{\kappa\gamma}{2}}(1-x_a^{(r)})^{Q_r|Q_r'}\\&\times
\prod_{\alpha=1}^{s_r-1}\prod_{a=1}^{s_r}
(x_a^{(r)}-y_\alpha^{(r)})^{-1+\frac{\gamma^2}{2}|-1+\frac{\gamma^2}{2}}
\d\nu_{s_r}(x^{(r)}).
\end{aligned}
\end{equation} We apply again Lemma \ref{lem:complex-twin-A7-final} to obtain that \begin{equation}
    \begin{aligned}
        J_r=C_r\prod_{a=1}^{s_r-1}\left(y_a^{(r)}\right)^{\frac{\kappa\gamma+\gamma^2}{2}|\frac{\kappa\gamma+\gamma^2}{2}}\left(1-y^{(r)}_a\right)^{Q_r+\frac{\gamma^2}{2}|Q_r+\frac{\gamma^2}{2}}D_{s_r-1}(y^{(r)})^{-1+\gamma^2|-1+\gamma^2},
\end{aligned}
\end{equation}
with the constant $C_r$ given by
\begin{equation}\label{eq:Kr-explicit}
C_r=
\frac{
\Gamma^{\C}(1+Q_r\,|\,1+Q_r')l\left(1+\frac{\kappa\gamma}{2}\right)
l\left(\frac{\gamma^2}{2}\right)^{s_r-1}
}{
\Gamma^{\C}\!\left(
s_r+1+\frac{\kappa\gamma}{2}+Q_r+(s_r-1)\!\left(-1+\frac{\gamma^2}{2}\right)
\ \Big|\
s_r+1+\frac{\kappa\gamma}{2}+Q_r'+(s_r-1)\!\left(-1+\frac{\gamma^2}{2}\right)
\right)
}.
\end{equation}

We combine all the above results to conclude that 
\begin{equation}\label{eq:recurrent relation}
\mathcal{I}_{\bs}\left(\frac{\kappa\gamma}{2},\gamma^2,B_1,B_1',\ldots,B_r,B_r'\right)=\prod_{k=0}C_k\mathcal{I}_{\bs-\bone}\left(\frac{\kappa\gamma+\gamma^2}{2},\gamma^2,B_1+\frac{\gamma^2}{2},B_1'+\frac{\gamma^2}{2},\ldots,B_r,B_r'\right),
\end{equation}
with $\bs-\bone=(s_1-1,\ldots,s_r-1).$

\paragraph{\textbf{Step 5: Conclusion}}
By \eqref{eq:complex Gamma}, one has 
$$
        \Gamma^{\C}(a|a')\Gamma^{\C}(a'|a)
        =
        \frac{\Gamma(a)}{\Gamma(1-a)}
        \frac{\Gamma(a')}{\Gamma(1-a')}
        =
        l(a)l(a'), \quad (a,a')\in\Lambda_{\C}.
$$
Let $\mathcal R_{\bs}:=C_0C_1\cdots C_r$
be the recurrence factor in \eqref{eq:recurrent relation}. Then writing $B^+_j=B_j$ and $B^-_j=B'_j$, we have
$\mathcal{R}_{\bs}^2=\mathcal{R}^+_{\bs}\mathcal{R}^-_{\bs}$,
with
$$
\begin{aligned}
\mathcal R_{\bs}^{\pm}
&=
l\left(\frac{s_1\gamma^2}{2}\right)
l\left(1+\frac{\kappa\gamma}{2}\right)
l\left(\frac{\gamma^2}{2}\right)^{-r}
\\
&\quad\times
\prod_{k=1}^{r-1}
\frac{
l\left(
k+\sum_{j=1}^{k}B_j^{\pm}
-(k-1)\frac{\gamma^2}{2}
\right)
}{
l\left(
k+1+\sum_{j=1}^{k}B_j^\pm
+\frac{\gamma^2}{2}(s_k-s_{k+1}-k)
\right)
}
\\
&\quad\times
\frac{
l\left(
r+\sum_{j=1}^{r}B_j^\pm
-(r-1)\frac{\gamma^2}{2}
\right)
}{
l\left(
r+1+\frac{\kappa\gamma}{2}
+\sum_{j=1}^{r}B_j^\pm
+\frac{\gamma^2}{2}(s_r-r)
\right)
},
\end{aligned}
$$
where we have used
$$
        l\left(1-\frac{\gamma^2}{2}\right)
        =
        l\left(\frac{\gamma^2}{2}\right)^{-1}.
$$

We now rewrite the sums $\sum_{j=1}^{k}B^\pm_j$ in terms of the weights of the fundamental
representation of highest weight $\omega_r$. Let
$h_1,\ldots,h_{r+1}$ be such weights, ordered so that $h_{k+1}-h_k=e_k,\,1\le k\le r.$ Then $h_{k+1}-h_1=\sum_{j=1}^{k}e_j,$ and hence
$\sum_{j=1}^{k}B^\pm_j=\left\langle h_{k+1}-h_1,\alpha_2\pm m_2\right\rangle.$
This shows that
\begin{equation}\label{eq:Rbs-pm-h}
\begin{aligned}
\mathcal R_{\bs}^{\pm}
&=
l\left(\frac{s_1\gamma^2}{2}\right)\,
l\left(1+\frac{\kappa\gamma}{2}\right)\,
l\left(\frac{\gamma^2}{2}\right)^{-r}
\\
&\quad\times
\prod_{k=1}^{r-1}
\frac{
l\!\left(
k+\bigl\langle h_{k+1}-h_1,\alpha_2\pm m_2\bigr\rangle-(k-1)\frac{\gamma^2}{2}
\right)
}{
l\!\left(
k+1+\bigl\langle h_{k+1}-h_1,\alpha_2\pm m_2\bigr\rangle+\frac{\gamma^2}{2}(s_k-s_{k+1}-k)
\right)
}
\\
&\quad\times
\frac{
l\!\left(
r+\bigl\langle h_{r+1}-h_1,\alpha_2\pm m_2\bigr\rangle-(r-1)\frac{\gamma^2}{2}
\right)
}{
l\!\left(
r+1+\frac{\kappa\gamma}{2}+\bigl\langle h_{r+1}-h_1,\alpha_2\pm m_2\bigr\rangle+\frac{\gamma^2}{2}(s_r-r)
\right)
}.
\end{aligned}
\end{equation}

We recall the special function $\Upsilon$ defined in \eqref{eq:Upsilon} and the shift relation \eqref{eq:shift}:
\begin{equation}\label{eq:Upsilon-shift-proof-detail}
\Upsilon(z+\chi)
=
l\!\left(\frac{\chi z}{2}\right)
\left(\frac{\chi}{\sqrt2}\right)^{1-\chi z}\Upsilon(z), \quad\chi\in\left\{\gamma\frac{2}{\gamma}\right\}.
\end{equation}
Iterating the shift equation \eqref{eq:Upsilon-shift-proof-detail}, we obtain that, for $z\in\C$ and $M\in\N$, \begin{equation}\label{eq:finite-l-product}
\prod_{j=0}^{M-1}l\!\left(z+j\frac{\gamma^2}{2}\right)
=
\left(\frac{\gamma}{\sqrt2}\right)^{M(2z-1)+\frac{\gamma^2}{2}M(M-1)}
\frac{
\Upsilon\!\left(\frac{2z}{\gamma}+M\gamma\right)
}{
\Upsilon\!\left(\frac{2z}{\gamma}\right)
}
\end{equation}

Iterating \eqref{eq:recurrent relation} first $s_1$ times, then on the truncated chain of length $r-1$, and so on, yields
\begin{equation}\label{eq:iterated-recurrence}
\mathcal I_{\bs}\!\left(\frac{\kappa\gamma}{2},\gamma^2;B_1^+,B_1^-,\ldots,B_r^+,B_r^-\right)
=
\prod_{\ell=1}^{r}\ \prod_{j=0}^{s_{\l}-s_{\l-1}-1}\mathcal R_{\ell,j}\,,
\end{equation}
where $s_0:=0$, and $\mathcal R_{\ell,j}$ is obtained from the recurrence factor by replacing
$$
r\mapsto r-\ell+1,\qquad
\kappa\mapsto \kappa+\gamma(s_{\ell-1}+j),\qquad
s_i\mapsto s_i-s_{\ell-1}-j,
$$
and
$$
B_i^\pm \mapsto \frac{\gamma}{2}\langle e_{\ell+i-1},\alpha_2\pm m_2\rangle.
$$
As in the previous step, we have $\mathcal R_{\ell,j}^2=\mathcal R_{\ell,j}^+\mathcal R_{\ell,j}^-,$
with
\begin{equation}\label{eq:Rljpm}
\begin{aligned}
\mathcal R_{\ell,j}^{\pm}
&=
l\!\left(\frac{(s_\l-s_{\l-1}-j)\gamma^2}{2}\right)\,
l\!\left(1+\frac{\gamma(\kappa+\gamma(s_{\ell-1}+j))}{2}\right)\,
l\!\left(\frac{\gamma^2}{2}\right)^{-(r-\ell+1)}
\\
&\quad\times
\prod_{k=1}^{r-\ell}
\frac{
l\!\left(
k+\frac{\gamma}{2}\langle h_{\ell+k}-h_\ell,\alpha_2\pm m_2\rangle
-(k-1)\frac{\gamma^2}{2}
+\frac{(s_{\ell-1}+j)\gamma^2}{2}
\right)
}{
l\!\left(
k+1+\frac{\gamma}{2}\langle h_{\ell+k}-h_\ell,\alpha_2\pm m_2\rangle
+\frac{\gamma^2}{2}(s_{\ell+k-1}-s_{\ell+k}-k)
\right)
}
\\
&\quad\times
\frac{
l\!\left(
r-\ell+1+\frac{\gamma}{2}\langle h_{r+1}-h_\ell,\alpha_2\pm m_2\rangle
-(r-\ell)\frac{\gamma^2}{2}
+\frac{(s_{\ell-1}+j)\gamma^2}{2}
\right)
}{
l\!\left(
r-\ell+2+\frac{\kappa\gamma}{2}
+\frac{\gamma}{2}\langle h_{r+1}-h_\ell,\alpha_2\pm m_2\rangle
+\frac{\gamma^2}{2}(s_r-r+\ell-1)
\right)
}.
\end{aligned}
\end{equation}
Hence
\begin{equation}\label{eq:factorized-square}
\left(
\mathcal I_{\bs}\!\left(\frac{\kappa\gamma}{2},\gamma^2;B_1^+,B_1^-,\ldots,B_r^+,B_r^-\right)
\right)^2
=
\mathcal R_{\bs}^{+,\mathrm{tot}}\,
\mathcal R_{\bs}^{-,\mathrm{tot}},
\end{equation}
where $\mathcal R_{\bs}^{\pm,\mathrm{tot}}:=\prod_{\ell=1}^{r}\ \prod_{j=0}^{s_\ell-s_{\l-1}-1}\mathcal R_{\ell,j}^{\pm}.$

Applying \eqref{eq:finite-l-product}, using the relation $h_{\ell+k}-h_\ell=e_\ell+\cdots+e_{\ell+k-1}$ and the neutrality $$2Q-\kappa\omega_r-(\alpha_2\pm m_2)-(\alpha_3\pm m_3)=
\gamma\sum_{i=1}^r s_i e_i,$$ one checks that

\begin{equation}\label{eq:Rplus-total-final}
\begin{aligned}
\mathcal R_{\bs}^{\pm,\mathrm{tot}}
&=
\left[
l\!\left(\frac{\gamma^2}{2}\right)
\left(\frac{\gamma}{\sqrt2}\right)^{2-\gamma^2}
\right]^{
\frac{2}{\gamma}\langle 2Q-\kappa\omega_r-(\alpha_2\pm m_2)-(\alpha_3\pm m_3),\rho\rangle}
\Upsilon(\gamma)^r\Upsilon(\kappa)
\\
&\quad\times
\frac{
\displaystyle\prod_{e>0}
\Upsilon\!\bigl(\langle Q-(\alpha_2\pm m_2),e\rangle\bigr)\,
\Upsilon\!\bigl(\langle Q-(\alpha_3\pm m_3),e\rangle\bigr)
}{
\displaystyle\prod_{i,j=1}^{r+1}
\Upsilon\!\left(
\frac{\kappa}{r+1}
+\langle \alpha_2\pm m_2-Q,h_i\rangle
+\langle \alpha_3\pm m_3-Q,h_j\rangle
\right)
}.
\end{aligned}
\end{equation}
Finally, using the conditions $m_1=0$ and $m_1+m_2+m_3=0$, we complete the proof.

\bibliographystyle{alpha}
\bibliography{CITT}
\end{document}